\documentclass[autodetect-engine,ja=standard,a4paper,10.5pt]{bxjsarticle}
\usepackage[pdftex]{graphicx}  % pdflatex用にpdftexオプションも追加
\usepackage{amsmath}
\usepackage{amssymb}
\usepackage{ascmac}
\usepackage{siunitx}
\usepackage{float}
\usepackage{url}
\usepackage{braket}
\usepackage{appendix}
\usepackage{cases}
\usepackage{mathtools}
\usepackage{docmute}
\usepackage{mathrsfs}
\usepackage{subfiles}
\usepackage{amsthm}
\usepackage{fancyhdr}
\usepackage{enumerate}
\usepackage{paralist}
\usepackage{color}
\usepackage{authblk} %所属を表示

\usepackage[T1]{fontenc}
\usepackage[utf8]{inputenc}

 \usepackage[backend=biber,style=phys,biblabel=brackets]{biblatex}
	\bibliography{Classification_of_global_structures_of_evaporating_regular_black_holes_with_infinite_periods_of_time.bib}

\everymath{\displaystyle}

\newcommand{\two}{I\hspace{-1.2pt}I}

\theoremstyle{definition}
\newtheorem{dfn}{Definition} 
\newtheorem{thm}{Theorem}
\newtheorem{example}{Example}
\newtheorem{prop}{Proposition}
\newtheorem{corollary}{Corollary}
\newtheorem{lemma}{Lemma}

\usepackage[colorlinks=true, bookmarks=true,
bookmarksnumbered=true, bookmarkstype=toc, linkcolor=black,
urlcolor=black, citecolor=black]{hyperref}
\usepackage[version=3]{mhchem}
\numberwithin{equation}{section}

\title{Classification of global structures of evaporating regular \\ black holes with infinite periods of time}

\author[1]{Kensuke Sueto}
\author[1,2]{Hirotaka Yoshino}

\affil[1]{Department of Physics, Graduate School of Science, Osaka Metropolitan University,, Osaka 558-8585, Japan}

\affil[2]{Nambu Yoichiro Institute of Theoretical and Experimental Physics (NITEP),
Osaka Metropolitan University, Osaka 558-8585, Japan}

\date{}

\begin{document}
\maketitle

\begin{abstract}
  We carry out model independent analyses for global structures of spherically symmetric regular black holes that evaporate and approach the extremal state spending infinite periods of time due to Hawking radiation. 
  We assume the radius of the outer apparent horizon (outer AH) to be a decreasing function of ingoing null coordinate, and consider the three cases, Cases 1, 2, and 3, where the radius of the inner AH is a constant, increases, and decreases,  respectively. 
  A complete classification of the Penrose diagrams is presented in each case, taking account not only of the presence and absence of the event horizon but also of the relative positions of the inner and outer AHs.
  Sufficient conditions that lead to each type are proved, and the examples of each type are summarized. 
  We also study the behavior of outgoing null geodesics in models where the null geodesic equation is solvable, and figure out the important factors to determine the type of a spacetime using concepts in the area of dynamical systems. 
  The formation of a Cauchy horizon is also established in the solvable models when the event horizon is present. 
 \end{abstract}

%
%======================================%
%<<<<<<<<<<<< SECTION I  >>>>>>>>>>>>>>%
%======================================%
%
 \section{\label{sec-introduction}Introduction}
 The theory of general relativity has broad applicability to gravitational phenomena such as cosmology, black holes and gravitational waves, but it is not expected to be a complete theory of gravity because the singularity theorems predict the formation of spacetime singularities in fairly generic situations such as the inside of black holes~\cite{Penrose:1964wq}. For example, in a Schwarzschild spacetime, a spacelike singularity is located at $r= 0$, and the spacetime is not extendible beyond that singularity. 
 
 Another problem of general relativity arises if we consider the effects of quantum fields in a curved spacetime. Hawking has shown that a black hole emits quantum particles with a thermal spectrum~\cite{Hawking:1975vcx}, --Hawking radiation-- leading, once backreaction is included, to black hole evaporation. This leads to the well-known information loss problem~\cite{Hawking:1976ra}. For these reasons, many physicists think that the theory of general relativity is an incomplete gravitational theory.

 No one knows whether a complete theory of gravity is quantized or can be formulated as a classical effective theory. In the latter case, it is expected that black holes predicted by such a theory would not have any singularities and would have different properties from those currently known. 
 The hypothetical black holes without spacetime singularities are called ``regular black holes,'' and here, we provide a brief review on this topic.
 
 Regular black hole was first introduced by Bardeen~\cite{Bardeen-2677}. 
 Subsequently, Dymnikova and Hayward independently introduced new regular black hole metrics, which established the foundation for contemporary research on regular black holes~\cite{Hayward:2005gi,Dymnikova:1992ux,Dymnikova:2003vt}. One approach to realizing these regular black holes involves incorporating nonlinear electromagnetism into the Lagrangian~\cite{Ayon-Beato:2000mjt,Fan:2016hvf}. However, constructing regular black holes from the Einstein-Hilbert action with some nonlinear electromagnetism action requires fine-tuning of integration constants and coupling constants~\cite{Chinaglia:2017uqd,Maeda:2021jdc}, making it difficult to consider them as natural black hole models whose regularity is realized in generic situations. 
 Notably, in the past year, authors of Ref.~\cite{Bueno:2024dgm} have developed a Lagrangian that derives spherically symmetric regular black holes from pure gravity without matter fields. 
 This theory enables the construction of spherically symmetric regular black holes in five or more dimensions without fine-tuning of integration constants, and thereby, some regular black holes have been derived as the exact solutions of this theory~\cite{Bueno:2024eiga,Bueno:2025gjg,DiFilippo:2024mwm,Frolov:2024hhe,Konoplya:2024kih}. Thus, regular black holes can be predicted from classical theories, and investigation of their properties is expected to enhance our understanding not only of black holes but also of gravity in general. For interested readers, we refer to Refs.~\cite{Maeda:2021jdc,Lan:2023cvz} in which the history and properties of regular black holes are explained.

 The global structure of black holes plays a fundamental role in investigating their properties. 
 This significance stems from the fact that black holes are defined by the global structure of spacetime, and this global structure is essential when discussing issues of predictability such as global hyperbolicity and the information loss problem~\cite{Kodama:1979vm,Lesourd:2018ekn}. 
 For these reasons, in this article, we focus on the global structure of evaporating regular black holes. 
 The global structures of Hayward black holes and charged Hayward black holes that completely evaporate have already been investigated, demonstrating that they are globally hyperbolic and do not give rise to the information loss problem~\cite{Hayward:2005gi,Frolov:2014jva,Sueto:2023ztw}. 
 In those studies, the complete evaporation is assumed by hand based on the expectation that the backreaction effects of the Hawking radiation would become significant to violate the semiclassical approximation at the last stage of evaporation.
 By contrast, it has been established that regular black holes do not completely evaporate if the semiclassical approximation holds throughout the evaporation~\cite{Carballo-Rubio:2018pmi} since a regular black hole generally possesses an even number of horizons, and under the Stefan-Boltzmann law, its surface gravity approaches zero over an infinite time. 
 We have independently confirmed this conclusion by studying the time evolution of an evaporating Hayward spacetime as presented in Appendix~\ref{sec-Hayward}
 (see also \cite{Chen:2014jwq} for a review on black hole remnants). 
 Therefore, to develop our understanding of regular black holes further, in the main text of this paper, we investigate the global structure of spherically symmetric regular black holes that undergo evaporation over infinite periods of time.

 The contents of this paper are as follows.
 In the former part of this paper, we classify the possible types of Penrose diagrams 
 without relying on specific models in the spacetimes of evaporating regular black holes spending infinite periods of time.
 Our argument is fairly general in the sense that the only primary assumption is that
 the radius of the outer apparent horizon (outer AH) is a monotonically decreasing function of the outgoing null coordinate,
 while that of the inner AH is constant, monotonically increasing, and monotonically decreasing (which we call Cases 1, 2, and 3, respectively).
 The classification is carried out taking account not only of whether the event horizon exists or not,
 but also of the relative positions of the outer and inner AHs. 
 There are two distinct types of spacetimes in Case~1, two in Case~2 and four in Case~3.
 Interested readers are suggested to look at the Penrose diagrams in Figs.~\ref{pic-case1-diagrams},
 \ref{pic-case2-diagrams}, and \ref{pic-case3_diagrams}, in advance.
 For each case, we also derive some sufficient conditions to determine the type of a given spacetime,
 and summarize the examples of the functional forms of the radii of the inner and outer AHs
 that lead to each type.

 In the latter part of this paper, we study the models in which the equation for outgoing null geodesics is solvable, and interpret the behavior of the outgoing null geodesics using the concepts in the area of dynamical systems to figure out important factors to determine the type of a spacetime.
 Using these models, we examine whether a Cauchy horizon is formed or not in the types where the event horizon is present.

 This paper is organized as follows. In Section~\ref{sec-settings}, we explain the class of metrics of evaporating regular black holes that are studied in this paper, and introduce the concepts that are necessary for the classification of spacetimes. 
 We then carry out the classification of the spacetime in Cases~1, 2, and 3 in Secs.~\ref{sec-case1}, \ref{sec-case2} and \ref{sec-case3}, respectively, and prove some sufficient conditions to classify the given spacetimes into these types. 
 Examples of each type are also presented.
 In Sec.~\ref{Sec:Solvable-models}, we study the solvable models and interpret the results using the concepts of the dynamical systems, and in Sec.~\ref{sec-extendibility}, we examine the extendibility beyond the Cauchy horizon using the solvable models.
 Section \ref{sec-conclusions} is devoted to a conclusion. 
 In Appendix~\ref{sec-Hayward}, we show an explicit calculation for outgoing null geodesics and the extendibility of the spacetime of the evaporating Hayward black hole. 
 In Appendix~\ref{Examples-proof},  the examples of each type presented in the main text are proved to satisfy the sufficient condition for the corresponding type.
 Appendix~\ref{Finiteness-of-affine-parameters} presents the detailed calculation for the extendibility beyond the Cauchy horizon in solvable models.

\pagebreak

 % By Penrose's singularity theorem, black holes generally have the singularity in the flame of general relativity. This fact implies that 
 % there appears an extremely curved spacetime region where quantum effects of gravity dominate and the Einstein equation is no longer valid.

%
%======================================%
%<<<<<<<<<<<< SECTION II  >>>>>>>>>>>>>>%
%======================================%
%
\section{\label{sec-settings}Settings and strategies}

In this section, we explain the setup of the problem and our strategies.

\subsection{Setup}
 We assume a general dynamically spherically symmetric spacetime whose metric is given by
 \begin{align}
     \label{eq-metric}
     ds^2=-f(v,r)A^2(v,r)dv^2+2A(v,r)dr dv+r^2d\Omega^2,
  \end{align}
 where $d\Omega^2$ refers to the metric of a two-dimensional unit sphere.
 Here, $A(v,r)>0$ is assumed in order to avoid the coordinate singularity. 
 We only treat asymptotically flat spacetimes, and hence, we require the condition 
 $\lim_{r\rightarrow \infty} f(v,r)=\lim_{r\rightarrow \infty} A(v,r)=1$. 
 The regularity conditions ensuring the absence of a curvature singularity at the center $r=0$
 are $f(v,0)=1, \left.\partial_r f(v,r)\right|_{r=0}=0$ and $\left.\partial_r A(v,r)\right|_{r=0}=0$. 
 The conditions $f(v,\infty)=f(v,0)=1$ imply that the equation $f(v,r)=0$ generally has an even number of solutions. 
 The solutions of $f(v,r)=0$ are the locations of the AHs (i.e., the surfaces with zero expansion) 
 since the expansion of the outgoing null geodesics is proportional to $f(v,r)$. 
 Therefore, asymptotically flat regular spacetimes generally have an even number of AHs. 
 In this paper, we assume that there be only two AHs and denote them as $r_\pm(v) \ (r_+(v)>r_-(v))$. 
 We henceforth refer to them simply as the outer AH and the inner AH. 
 From this assumption, the function $f(v,r)$ can always be expressed using a positive-valued function $F(v,r)$ as
 \begin{align}
     f(v,r)=F(v,r)(r-r_+(v))(r-r_-(v)).
     \label{formula-f(v,r)}
 \end{align}
 Since the normal vector to the AHs is 
 \begin{align}
     n\propto \mathrm{d}f=\partial_vf\mathrm{d}v+\partial_rf\mathrm{d}r,
 \end{align}
 where $\mathrm{d}$ is the external derivative, 
 the norm of this vector is 
 \begin{align}
     g^{-1}(n,n)&\propto \frac{2}{A}\partial_rf\partial_vf. 
 \end{align}
 As $\partial_r f$ is positive and negative on the outer and inner AHs, respectively, 
 the characteristics of the AHs are determined by the sign of $\left.\partial_v f\right|_{r_\pm}$, and we can easily evaluate them as
 \begin{align}
     \begin{cases}
         \left.\partial_vf\right|_{r_+}&=-F(v,r_+)(r_+-r_-)\partial_v r_+,\nonumber\\
         \left.\partial_v f\right|_{r_-}&=F(v,r_-)(r_+-r_-)\partial_v r_-.
     \end{cases}
 \end{align}
 Therefore, the outer and inner AHs become null when $\partial_v r_{\pm}=0$, timelike when $\partial_v r_{\pm}<0$, and spacelike when $\partial_v r_{\pm}>0$. 
 From the above discussions, we can classify the spacetimes of evaporating black holes into the following three cases;
 \begin{enumerate}[\textrm{Case} 1:]
     \item $\partial_vr_+<0, \partial_v r_-=0$, the spacetime has a timelike outer AH and a null inner AH;
     \item $\partial_v r_+<0,\partial_v r_->0$, the spacetime has a timelike outer AH and a spacelike inner AH;
     \item $\partial_v r_+<0,\partial_v r_-<0$, the spacetime has a timelike outer AH and a timelike inner AH.
 \end{enumerate}
 The known models such as the Bardeen or Hayward black holes are classified as Case~2
 (see Appendix~\ref{sec-Hayward} for the case of the Hayward black hole). 
 
 In the subsequent sections~\ref{sec-case1}, \ref{sec-case2}, and \ref{sec-case3}, 
 Penrose diagrams for each case will be constructed. 
 By analyzing the behavior of outgoing null geodesics, these diagrams can be classified without assuming concrete expressions for the functions in the metric, such as  $A(v,r), F(v,r)$, and $r_\pm(v)$. 
 The equation for the outgoing null geodesics is given by
 \begin{align}
     \label{eq-outgoing}
     \frac{dr}{dv}=\frac{1}{2}A(v,r)f(v,r),
 \end{align}
and all analyses in this paper are based on this equation. 
This right-hand side is expressed as $A(v,r)F(v,r)(r-r_+(v))(r-r_-(v))$ where $A(v,r)F(v,r)>0$ is satisfied. Therefore, we can define a positive-valued function $h(v,r)$ as 
\begin{equation}
h(v,r)\coloneqq A(v,r)F(v,r)
\label{Def:h(v,r)}
\end{equation} 
and we would analyze the equation
\begin{align}
  \label{eq-null-satisfy}
  \frac{dr}{dv}\,=\,\frac{1}{2}h(v,r)(r-r_+(v))(r-r_-(v))
\end{align}
instead of Eq. \eqref{eq-outgoing} without loss of generality.

Since we are interested in the evaporation phase, we impose the following conditions:
\begin{inparaenum}[(1)]
    \item $\partial_vr_+<0$;
    \item $\lim_{v\rightarrow \infty}r_{\pm}=r_c$; 
    \item The characteristic of each AH is never changed; and
    \item The limits $v\to\infty$ of the functions $A(v,r)$, $F(v,r)$ and $h(v,r)$ behave as analytic functions of $r$.
\end{inparaenum}
The first condition is based on the physical assumption that the radius of the outer AH decreases during the evaporation phase.  The second condition is a representation of the assumption that regular black holes have infinite evaporation time and asymptote to extremal black holes.  The third and fourth conditions are imposed for simplicity; one can discard these conditions when considering more complex evaporation scenarios. 
%Additionally, we assume that $h(v,r)$ satisfy the following two conditions;
%\begin{inparaenum}[(a)]
%  \item $h(\infty,r)=\lim_{v\rightarrow \infty} h(v,r)$ behaves as an analytic function of $r$ ; and 
%  \item $\lim_{v\rightarrow\infty} h(v,r_c)$ converges to a finite value $h_c$.
%\end{inparaenum}
By the fourth condition, all of $\partial_r^nh(v,r)$ with $n=0,1,2,...$ converge to finite values in the limit $v\to \infty$. In particular,
we define $h(\infty,r):=\lim_{v\rightarrow \infty} h(v,r)$ and
\begin{equation}
h_c:=\lim_{v\to\infty}h(v,r_c)
\end{equation} 
for later convenience.

Here, we introduce a shifted radial coordinate $x$ by
\begin{equation}
x \, = \, r-r_c.
\label{shifted-radial-coordinate}
\end{equation}
Correspondingly, the positions of the outer and inner AHs are given by $x=x_\pm(v)$ with
\begin{equation}
x_\pm(v) \, = \, r_\pm(v)-r_c.
\label{horizons-shifted-radial-coordinate}
\end{equation}
The merit of introducing this shifted radial coordinate is that $x_\pm(v)$ become simple functions, i.e., $\lim_{v\to\infty}x_\pm(v) = 0$.
Hereafter, when we handle the equation of outgoing null geodesics,
the coordinate $x$ will be used rather than $r$ (but we use the original radial coordinate $r$ when we consider the construction of Penrose diagrams). We redefine the functions
$h(v,r(x))$, $F(v,r(x))$, and $A(v,r(x))$ as $h(v,x)$, $F(v,x)$, and $A(v,x)$. 
Equation~\eqref{eq-null-satisfy} for the outgoing null geodesics is rewritten as
\begin{equation}
  \frac{dx}{dv}=\frac{1}{2}h(v,x)(x-x_+(v))(x-x_-(v)).
  \label{eq-null-x}
\end{equation}
For later convenience, we introduce the primitive functions of $x_\pm(v)$ with
\begin{equation}
X_\pm(v):=\int x_\pm(v)dv.
\label{Primitive-function-xpm}
\end{equation}
These primitive functions play important roles to specify the types of Penrose diagrams.
Note that our analyses do not depend on the choice of the integral constants.
In addition, in the case that an event horizon is present, 
we denote its worldline as $x=x_{\mathrm{EH}}(v)$. 

The behavior of the outgoing null geodesics strongly depends on their relative position to the outer and inner AHs. 
For outgoing null geodesics located outside or on the outer AH, the following lemma holds irrespectively of the specific case:
\begin{lemma}
  \label{lemma-1}
  Outgoing null geodesics initially located outside or on the outer AH never intersect the outer AH at any later time and reach the future null infinity $\mathscr{I}^+$.
\end{lemma}

\begin{proof}
Since in this region $\frac{dx}{dv}\geq 0$ and $\frac{dx_+}{dv}<0$ are satisfied, on every point on the outer AH, the outgoing null geodesics cross the horizon from the inside region to the outside region.
Then, if an outgoing null geodesic initially located outside or on the outer AH 
arrives at some point on the outer AH,
it immediately contradicts the uniqueness of solutions of first-order differential equations. 
Therefore, if an outgoing null geodesic is outside or on the outer AH at a given time, they will not intersect it at later finite advanced times. 
Next, we examine the behavior of the integral curve in the limit $v\rightarrow \infty$. 
For the sake of contradiction, we assume that an outgoing null geodesic satisfy $\lim_{v\rightarrow \infty} x(v)=x_{\mathrm{lim}} (\neq \infty)$. 
Since the right-hand side of Eq.~\eqref{eq-null-x} converges to $h(\infty,x)x^2$ in the limit $v\rightarrow \infty$, the values of $dx/dv$ of these integral curves are
\begin{align}
  \lim_{v\rightarrow \infty}\frac{dx}{dv}
  \,=\,\frac{1}{2}h(\infty,x_{\mathrm{lim}})x_{\mathrm{lim}}^2
  \,>\,0.
\end{align}
This contradicts the assumption that $x_{\mathrm{lim}}$ is the limit value for $v\rightarrow \infty$. This leads to the conclusion that these outgoing null geodesics must satisfy $\lim_{v\rightarrow \infty} x(v)=\infty$, and thus must arrive at the future null infinity $\mathscr{I}^+$.
\end{proof}  

This Lemma indicates that if there is the event horizon, it must be formed inside the outer AH. 
The behavior of outgoing null geodesics in other positions requires case-by-case analyses. Therefore, we will examine Cases 1, 2, and 3 separately in Sects.~\ref{sec-case1}, \ref{sec-case2}, and \ref{sec-case3}.

\subsection{Compactified outgoing null coordinate}
\label{Subsec:compactified-retarded}

The purpose of this paper is to classify all Penrose diagrams in Cases 1, 2, and 3,  focusing attention not only on the existence/nonexistence of the event horizon,
but also on the relative locations of the outer and inner AHs.
In order to characterize the location of the outer and inner AHs,
it is helpful to introduce the compactified outgoing null coordinate $U$.

We consider the part of the spacetime given by $v\ge v_0$ for some value of $v_0$. 
Then, we consider the union of the set $v=v_0$ and $r\ge 0$ and the set $r=0$ and $v\ge v_0$.
These sets give the lower left boundary of the Penrose diagram.
Defining the value of $U$ on these sets as
\begin{equation}
U \, = \begin{cases}
 \frac{2}{\pi}\arctan\left(\frac{v-v_0}{r_c}\right) & (r=0,\, v\ge v_0),\\
 -\frac{2}{\pi}\arctan\left(\frac{r}{r_c}\right) & (r>0,\, v= v_0),
 \end{cases}
\end{equation}
the coordinate $U$ has been introduced on the union of the two sets.
 Then, we require $U$ to be constant along each radially outgoing null geodesic.
By this procedure, the coordinate $U$ is extended to the whole range of $v>v_0$.
Note that $U$-constant lines do not cross each other because of the uniqueness properties
of the first-order differential equations. 
Since we consider Penrose diagrams, 
we require the coordinate range to be $-1\le U\le+1$, by adding the points $U=\pm 1$ that correspond to
 $v=\infty$ on the set $r=0, v\ge v_0$ and to $r=\infty$ on the set $v=v_0, \, r\ge 0$, respectively.
In what follows,  the coordinate range of $U$ is denoted as the interval,
 \begin{equation}
 J\,=\, \left[-1,\,+1\right].
 \end{equation}

We suppose the outer and inner AHs to be located at 
\begin{equation}
U\, = \, U^{(\pm)} (v)
\end{equation}
with the functions $U^{(\pm)}$. 
The domain of definition of these functions is $v_0\le v<\infty$,
and we define
\begin{equation}
U^{(\pm)}_\infty \, = \, \lim_{v\to\infty}U^{(\pm)} (v).
\label{Upm_infty}
\end{equation}
Note that the inside regions of the outer and inner AHs are given by
$U>U^{(\pm)}(v)$, while the outside regions of the outer and inner AHs
are given by $U<U^{(\pm)}(v)$, respectively.

For the outer AH, the function $U^{(+)}(v)$ is a monotonically increasing function
because of its timelike property. We define the range of the function $U^{(+)}(v)$
to be $I^{(+)}$. Introducing $U^{(\pm)}_0 = U^{(\pm)}(v_0)$, the interval $I^{(+)}$ is given by
\begin{equation}
I^{(+)} \, = \,  \left[ U^{(+)}_0, U^{(+)}_\infty \right).
\end{equation}
Similarly, we define the range of the function $U^{(-)}(v)$ to be $I^{(-)}$.
Since the inner horizon is null, spacelike, and timelike in Cases 1, 2, and 3,
the function $U^{(-)}(v)$ is constant, monotonically decreasing,
and monotonically increasing, respectively. Therefore, we have
\begin{equation}
I^{(-)} \, = \,
\begin{cases}
 \left\{U^{(-)}_0\right\} & (\textrm{Case~1});\\
  \left (U^{(-)}_\infty,\,U^{(-)}_0\right] &  (\textrm{Case~2});\\
   \left[U^{(-)}_0,\, U^{(-)}_\infty\right) &  (\textrm{Case~3}),
 \end{cases}
\end{equation}
where $\left\{U^{(-)}_0\right\}$ denotes the singleton set.
Note that since the outer AH is located outside the inner AH,
we have $U^{(+)}(v)<U^{(-)}(v)$, and hence
\begin{equation}
U^{(+)}_0<U^{(-)}_0, \qquad \textrm{and} \qquad U^{(+)}_\infty\le U^{(-)}_\infty.
\label{Upm-constraint}
\end{equation}
The equality in the latter inequality is possible because the values of $U^{(\pm)}_\infty$
are defined by the limit of $v\to\infty$ in Eq.~\eqref{Upm_infty}.

We now have the three sets, $J$ and $I^{(\pm)}$. 
Each of $I^{(\pm)}$ indicates the coordinate range of $U$ where
the outer or inner AH is present. 
Our basic strategy is to define the types of the Penrose diagrams
in each of the three cases based on the relative positions of these three sets.
The concrete procedures will be explained in each of the following three sections.

It is worth pointing out the following lemma:
\begin{lemma}
  \label{lemma-2}
  The event horizon exists if and only if $U^{(+)}_\infty\,<\,+1$. If this condition is satisfied,
  the event horizon is located at $U=U^{(+)}_\infty$. 
\end{lemma}
\begin{proof}
If $U^{(+)}_\infty\,<\,+1$, an arbitrary outgoing null geodesic in the range $U^{(+)}_\infty\,<\,U\,<\,+1$
do not cross the outer AH, and hence, remains inside the outer AH.
This means that such a geodesic is trapped and is in the black hole region, indicating the presence of the event horizon.
As for the ``only if'' part, we prove the contraposition. If $U^{(+)}_\infty = +1$, all null geodesic
in the domain $U^{(+)}_0\le U<+1$ cross the outer AH, and hence escape to future null infinity
from Lemma~\ref{lemma-1}. 
Since we know that all null geodesics in the outside domain $-1<U<U^{(+)}_0$ escape to future null infinity as well,
this means that there is no event horizon.

If the condition $U^{(+)}_\infty\,<\,+1$ is satisfied, all null geodesics in the region $U^{(+)}_\infty\le U<+1$
do not escape to infinity. Since the event horizon is the outermost null geodesic congruence, 
it is given by $U=U^{(+)}_\infty$.
\end{proof}
It would be worth presenting the following corollary:
\begin{corollary}
  \label{corollary-Sec2}
  The black hole is absent if and only if $U^{(+)}_\infty\,=\,+1$. 
  \end{corollary}

We now move to the case-by-case analyses for Cases 1, 2, and 3.

\newpage

%
%======================================%
%<<<<<<<<<<<< SECTION III  >>>>>>>>>>>>>>%
%======================================%
%
\section{\label{sec-case1} Case~1 : Null inner AH}
In this section, we study Case~1 where the inner AH is a null surface.
Since $x_-(v)=0$ in this case, the analysis becomes simplified 
compared to the other two cases 2 and 3, and this is a good starting point.
For outgoing null geodesics, the following lemma holds:
\begin{lemma}
  \label{lemma-case1}
  The inner AH consists of radially outgoing null geodesics.
  Outgoing null geodesics except on the inner AH never intersect the inner AH. 
  Furthermore, outgoing null geodesics inside the inner AH asymptotically approach $x\to 0$ as $v\rightarrow\infty$. Consequently, an event horizon exists in a Case~1 spacetime.
\end{lemma}
\begin{proof}
  For a Case~1 spacetime, the differential equation \eqref{eq-null-x} of outgoing null geodesics becomes 
  \begin{align}
    \frac{dx}{dv}=\frac{1}{2}h(v,x)(x-x_+(v))x,
  \end{align}
  and therefore, $x(v)=0$ is a solution to this equation. 
  Since the right-hand side of this equation is continuous, the uniqueness of solutions to differential equations holds. 
  Therefore, there is no solution $x(v)$ that becomes zero except the solution $x(v)=0$. 
  This leads to the conclusion that outgoing null geodesics except on the inner AH never intersect the inner AH.

  Combining this result with $x_-(v)=0$, outgoing null geodesics inside the inner AH 
  must stay in the region $x<0$. Since the function $x=x(v)$ for any outgoing null geodesic in this region
  is monotonically increasing, it must converge to a nonpositive value in the limit $v\to\infty$, i.e. 
  $\lim_{v\rightarrow \infty} x(v)\leq 0$. For the sake of contradiction, we assume that one of them satisfy 
  $\lim_{v\rightarrow \infty} x(v)=x_{\mathrm{sup}} (< 0)$. 
  Then, the value of $dx/dv$ satisfies
  \begin{align}
    \lim_{v\rightarrow \infty}\frac{dx}{dv}=\frac{1}{2}h(\infty,x_{\mathrm{sup}})x_{\mathrm{sup}}^2>0.
  \end{align}
  This is a contradiction to the assumption that $x_{\mathrm{sup}}$ is the limit value for $v\rightarrow\infty$. This leads to the conclusion that these outgoing null geodesics satisfy $\lim_{v\rightarrow\infty} x(v)=0$.
\end{proof}

\begin{figure}[t]
  \centering
  \includegraphics[width=0.5\linewidth]{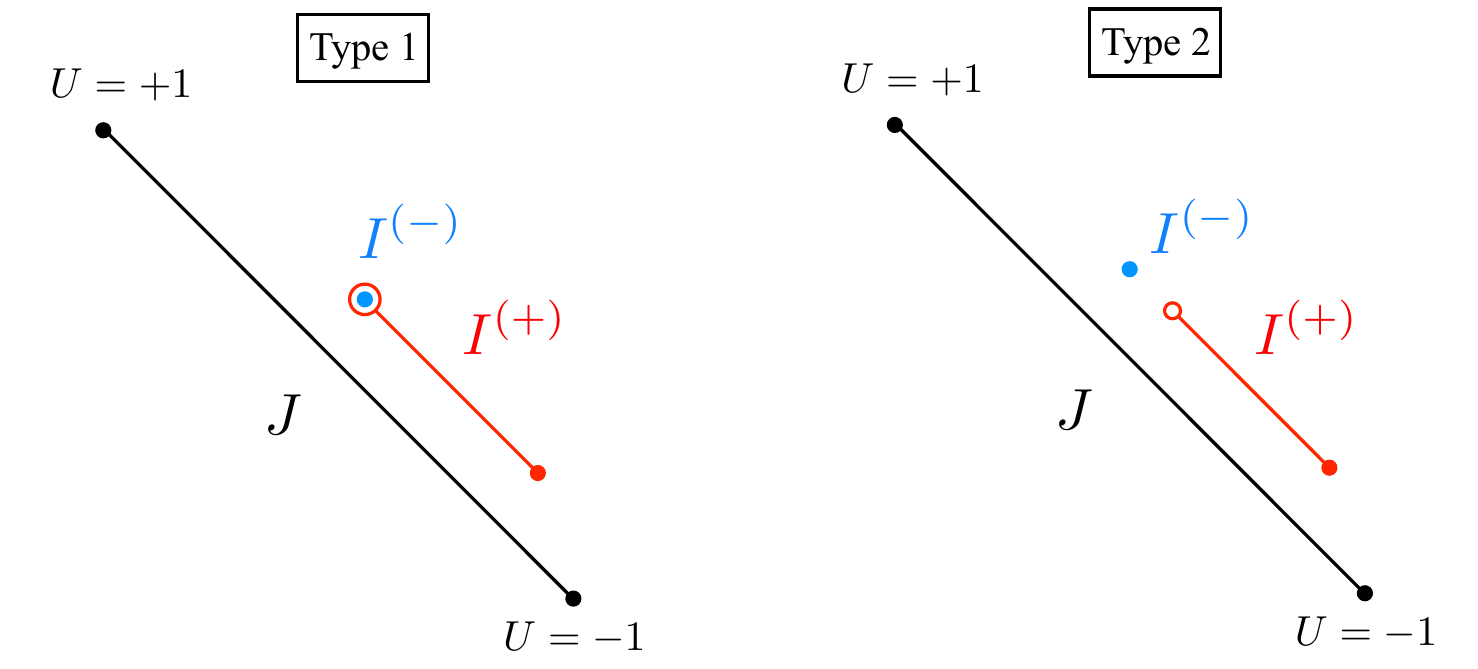}
  \caption{The sets $J$ and $I^{(\pm)}$ in Type~1 (left) and in Type~2 (right) in Case~1. The segments are drawn
  at the angle of $135^{\circ}$ to the horizontal direction 
  to indicate that we are working in the double null coordinates, $(U,v)$.}
  \label{Rod-structure-case1}
\end{figure}

From this lemma, the region $x<0$ is the black hole region because photons inside this region never
reach the future null infinity. Then, we can consider two possibilities: 
One is that $x=0$ is the event horizon, and the other is that an event horizon exists in the region $x>0$.
In fact, it will turn out that which possibility is realized depends on the form of $A(v,x)f(v,x)$. 
To mathematically formulate the difference between the two configurations,
let us recall the three sets, $J$ and $I^{(\pm)}$ introduced in Sec.~\ref{Subsec:compactified-retarded}.
In Case~1, we have shown $I^{(-)}=\left\{U^{(-)}_0\right\}$ and  $I^{(+)}=\left[U^{(+)}_0, U^{(+)}_\infty\right)$.
Since $U=U^{(-)}_0$ corresponds to the outgoing null geodesic on the inner AH
while the value of $U$ that belongs to $I^{(+)}$ corresponds to the outgoing null geodesic
that escapes to infinity, there must not be an overlap between $I^{(-)}$ and $I^{(+)}$.
Therefore, the inequality
\begin{equation}
-1<U^{(+)}_0<U^{(+)}_\infty\le U^{(-)}_0<+1
\end{equation}
must hold. Then, we define Type~1 and Type~2 of the spacetimes in Case~1 as follows:
\begin{dfn}
The spacetime of Case~1 is defined to be of Type~1 if  $U^{(+)}_\infty= U^{(-)}_0$ holds,
while the spacetime of Case~1 is defined to belong to Type~2 if  $U^{(+)}_\infty< U^{(-)}_0$ holds.
\label{Def:Case1-Type1-Type2}
\end{dfn}
The intervals $J$ and $I^{(\pm)}$ in Type~1 and Type~2 are depicted
in Fig.~\ref{Rod-structure-case1}. 
Since $U=U^{(+)}_\infty$ corresponds to the event horizon from Lemma~\ref{lemma-2}, the inner AH becomes the event horizon in a Type~1 spacetime,
while there is an event horizon outside the inner AH in a Type~2 spacetime.
Since there is no possibility other than Type~1 and Type~2, we have the following theorem;
\begin{thm}
  \label{thm-case1}
  The classification of Definition~\ref{Def:Case1-Type1-Type2} is complete
  in the sense that any spacetime of Case~1 belongs to either Type~1 or Type~2.
  \end{thm}

\begin{figure}[t]
  \centering
  \includegraphics[width=0.5\linewidth]{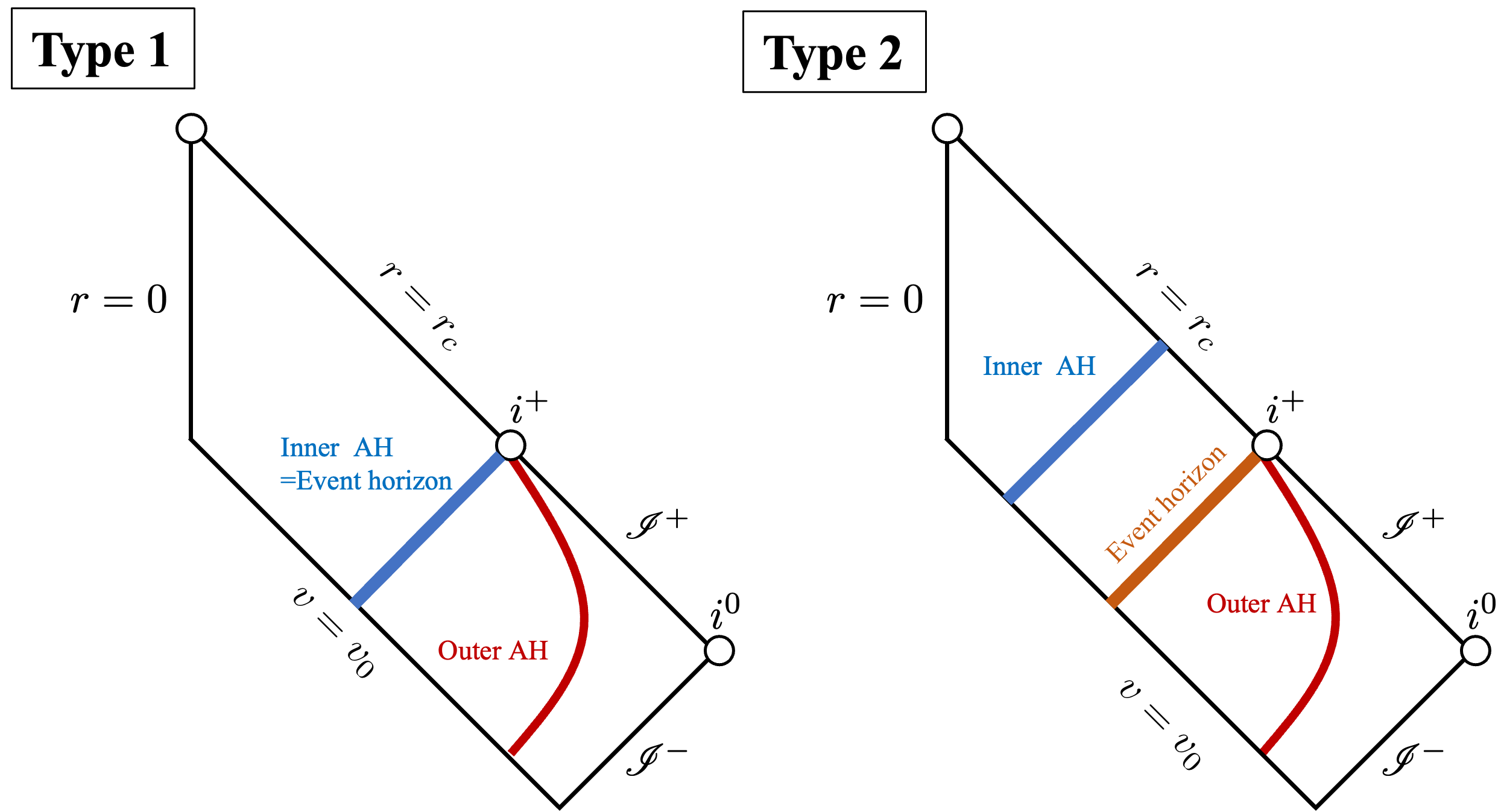}
  \caption{The possible Penrose diagrams of Case~1 spacetimes of Type~1 (left panel) and of Type~2 (right panel). 
  The red and blue curves represent the outer and inner AHs, and the orange curves represent the event horizon, respectively. }
  \label{pic-case1-diagrams}
\end{figure}

Consequently, we can construct two types of Penrose diagrams of Case~1 spacetimes. 
To draw the diagram, firstly, we consider a null hypersurface $v=v_0$ (the lower line intersecting the horizontal line at the angle of $135^{\circ}$) 
and a timelike worldline of the central point, $r=0$ (the vertical line). 
Secondly, we consider outgoing null geodesics emitted from the hypersurface $v=v_0$. 
These null geodesics are drawn as lines each of which crosses the horizontal line at the angle of $45^{\circ}$.
These lines are extended up to the surface $v=\infty$, which is the upper line parallel to the line of $v=v_0$. 
From Lemma \ref{lemma-1}, Lemma \ref{lemma-case1} and Fig.~\ref{Rod-structure-case1}, we know 
the qualitative behavior of all outgoing null geodesics, particularly in their relative positions to the AHs,
once its Type is specified.
In the case of Type~1, the inner AH is the event horizon, 
and all null geodesics between the inner and outer AHs cross the outer AH
which is timelike. 
On the other hand, in the case of Type~2, the event horizon is located at a separate position from the inner AH.
In the limit $v\to\infty$, the values of $r$ of all outgoing null geodesics in the event horizon asymptote to $r_c$.
Therefore, in the black hole region, the surface $v=\infty$ satisfies $r=r_c$ at the same time. 
The resultant Penrose diagrams are shown in Figure \ref{pic-case1-diagrams}.

Note that there are two points that correspond to timelike infinity, $i^+$. 
One is the point at which $v=\infty$ and the event horizon intersect.
The other is the point at which $v=\infty$ and the worldline of $r=0$ intersect, because 
the proper time along $r=0$ diverges in the limit $v\to\infty$.  
To complete the Penrose diagrams, we have to determine whether the surface $v=\infty$ in the black hole region
is a Cauchy horizon or null infinity.
We postpone this issue for a while and give brief discussions on this point
later in Sec.~\ref{sec-extendibility}.

\subsection{\label{subsec-case1-1}Sufficient conditions to classify as Type~1 or Type~2.}
Here, we derive some sufficient conditions to determine which Type of the Penrose diagrams corresponds to a given metric.
Our arguments frequently rely on the $(\varepsilon, \delta)$-definition of limit. In the next proposition, for example, 
we use the fact that $\lim_{v\to\infty}h(v,x(v))=h_c$ holds along any curve $x=x(v)$ in the region $0<x<x_+(v)$
(because $\lim_{v\to\infty}x(v)=0$ holds),
and express this condition in the $(\varepsilon,\delta)$-form as ``for any $\varepsilon$, there exists $v_0$ such that for all $v>v_0$, 
the inequality $|h(v,x)-h_c|<\varepsilon h_c$ holds.''

Recall that the difference between Type~1 and Type~2 is whether all of the outgoing null geodesics outside the inner AH intersect the outer AH or not. It would also be instructive to recall the definitions
of $X_\pm(v)$ given in Eq.~\eqref{Primitive-function-xpm}. 
The following proposition gives the condition of the Case~1 spacetime to be of Type~1:
\begin{prop}
  \label{prop-Case1-Type1}
  If there exists a positive constant $\varepsilon$ for which 
  \begin{align}
    \label{eq-prop-Case1-Type1}
    \lim_{v\rightarrow\infty}x_+(v)\exp\left[\frac{1+\varepsilon}{2}h_cX_+(v)\right]\,=\,0
  \end{align}
  is satisfied, then all outgoing null geodesics initially located between the two AHs intersect the outer AH eventually, i.e., the spacetime is of Type~1.
\end{prop}

\begin{proof}
  Consider outgoing null geodesics which initially satisfy $0<x(v_0)<x_+(v_0)$. 
  For the sake of contradiction, we assume the existence of an outgoing null geodesic
  that stays in the range $0<x(v)< x_+(v)$. 
  This null geodesic satisfies
  \begin{align}
    \frac{dx}{dv}\,=\,\frac{1}{2}h(v,x)(x-x_+)x\,>\,-\frac{1}{2}h(v,x)x_+x\,>\,-\frac{1+\varepsilon}{2}h_cx_+x
  \end{align}
  in the region $v\ge v_0$
  for a sufficiently large $v_0$, because $h(v,x)<(1+\varepsilon)h_c$ holds in the region $0<x(v)< x_+(v)$. 
  Since $x(v)>0$ holds, we can divide both sides with $x(v)$ and integrate the resultant inequality as
  \begin{align}
    \ln{\left(\frac{x}{x_0}\right)}\,>\,-\frac{1+\varepsilon}{2}h_c\int_{v_0}^{v}{x_+(v^\prime)dv^\prime},\nonumber
  \end{align}
  where $x_0\coloneqq x(v_0)$. This inequality is equivalent to
  \begin{align}
    x(v)>x_0\exp{\left[-\frac{1+\varepsilon}{2}h_c(X_+(v)-X_+(v_0))\right]},
  \end{align}
  because the exponential function is a monotonically increasing function. 
    Since $x_+(v)>x(v)$ holds from the assumption, we have
  \begin{align}
    x_+(v)\exp{\left[\frac{1+\varepsilon}{2}h_cX_+(v)\right]}\,>\,x_0\exp{\left[\frac{1+\varepsilon}{2}h_cX_+(v_0)\right]}.
  \end{align}
  However, from the assumption of Eq.~\eqref{eq-prop-Case1-Type1}, the left-hand side of this inequality converges to zero for $v\rightarrow\infty$. This is a contradiction. Since this contradiction is caused by the assumption
  of the presence of an outgoing null geodesic staying in the region $0<x<x_+(v)$,
  such an outgoing null geodesic must not exist. Thus, any
  outgoing null geodesic with $0<x_0<x_+(v_0)$ must intersect the outer AH, i.e., the spacetime is of Type~1.
\end{proof}

Next, we prove the proposition that gives the sufficient condition of the Case~1 spacetime to be of Type~2:
\begin{prop}
  \label{prop-Case1-prop2}
  If there exists a parameter $\alpha$ in the range $0<\alpha<1$ which satisfies
  \begin{align}
    \label{eq-Case1-sufficient-2-1}
    \lim_{v\rightarrow \infty}x_+(v)\exp{\left[\frac{1-\alpha}{2}h_c X_+(v)\right]}\,>\,0,
  \end{align}
    then the spacetime is of Type~2. Here, Eq.~\eqref{eq-Case1-sufficient-2-1}
    includes the case that the left-hand side diverges to $+\infty$.
\end{prop}
\begin{proof}
  For the sake of contradiction, we assume that all outgoing null geodesics between the inner and outer AH cross the outer AH
  at later times. 
  From this assumption, all outgoing null geodesics leaving from 
  the points $(v_0,x_0)$ with $0<x_0<x_+(v_0)$ must intersect the outer AH $x=x_+(v)$. 
  Therefore, they also intersect the surface $x=\bar{\alpha} x_+(v)$ for any $\bar{\alpha}$ satisfying $0<\bar{\alpha}<1$. Let $v_{\bar{\alpha}}$ be the time of this intersection. Then, for the range of $v_0<v<v_{\bar{\alpha}}$, $x(v)<\bar{\alpha}x_+(v)$ holds. Since $h(v,x)>(1-\varepsilon)h_c$ holds in the range $0<x(v)<\bar{\alpha}x_+(v)$ for an arbitrary $\varepsilon$ satisfying $0<\varepsilon<1$ by adopting a sufficiently large $v_0$, $x(v)$ satisfies
  \begin{align}
    \frac{dx}{dv}\,=\,\frac{1}{2}h(v,x)(x-x_+)x
    \,<\,\frac{(1-\varepsilon)(\bar{\alpha}-1)}{2}h_c x_+x.
   \end{align}
   Here, we choose the values of $\varepsilon$ and $\bar{\alpha}$ so that they satisfy $1-\alpha= (1-\varepsilon)(1-\bar{\alpha})$
   for the parameter $\alpha$ given in the proposition. Note that
   by choosing a sufficiently small $\varepsilon$, it is possible to make $\bar{\alpha}$ satisfy $0<\bar{\alpha}<1$.
   Then, we have
   \begin{align}
       \frac{dx}{dv}\,<\,-\frac{1-\alpha}{2}h_cx_+x.
  \end{align}
  Dividing both sides by $x$ 
  and integrating the resultant inequality, we have
  \begin{align}
    \ln{\left(\frac{x}{x_0}\right)}\,<\,-\frac{1-\alpha}{2}h_c \int_{v_0}^{v}{x_+(v^\prime)dv^\prime}
  \end{align}
  where $x_0 \coloneqq x(v_0)$. This inequality is equivalent to
  \begin{align}
    \label{eq-Case1-Type2}
      x<x_0\exp{\left[-\frac{1-\alpha}{2}h_c (X_+(v)-X_+(v_0))\right]}.
  \end{align}
Since $x(v)$ intersects $\bar{\alpha} x_+(v)$ at $v=v_{\bar{\alpha}}$, the right-hand side of Eq.~\eqref{eq-Case1-Type2} becomes $\bar{\alpha} x_+(v)$ at $v=\hat{v}$ for some $\hat{v}$ satisfying $\hat{v}<v_{\bar{\alpha}}$. 
Therefore, for each $x_0$ in the range $0<x_0<\bar{\alpha}x_+(v_0)$, there exists $\hat{v}$ satisfying
\begin{align}
  \label{eq-Case1-sufficient-2}
  x_0\,=\,\bar{\alpha}x_+(\hat{v})\exp{\left[\frac{1-\alpha}{2}h_c(X_+(\hat{v})-X_+(v_0))\right]}.
\end{align}
The right-hand side of this equation is a strictly positive continuous function of $\hat{v}$, 
and from the assumption of Eq.~\eqref{eq-Case1-sufficient-2-1}, 
this asymptotes to a positive value or diverges to $+\infty$ for $\hat{v}\rightarrow \infty$.
Therefore, this right-hand side has a positive minimum or a positive infimum in the range $v_0\le \hat{v}$, which we denote as $x_0^{(\mathrm{min})}$. 
Consequently, for outgoing null geodesics with $x_0<x_0^{(\mathrm{min})}$, there is no $\hat{v}$ satisfying Eq.~\eqref{eq-Case1-sufficient-2}, and hence, those geodesics do not intersect the curve $x=\bar{\alpha}x_+(v)$.
This is a contradiction. Since this contradiction is caused by the assumption that all outgoing null geodesics
between the inner and outer AHs cross the outer AH at later times,  
we have the conclusion that if Eq.~\eqref{eq-Case1-sufficient-2-1} holds, then there is an event horizon in the region $x>0$, i.e., the spacetime is of Type~2.
\end{proof}

\begin{corollary}
    \label{prop-case1-sufficient-2-2}
  If there exist $v_0$ and $\alpha$ in the range $0<\alpha<1$ which satisfy
  \begin{align}
    \label{eq-Case1-sufficient-2-2}
    \frac{dx_+}{dv}\,>\,-\frac{h_c}{2}(1-\alpha)x_+^2
  \end{align}
  for $v>v_0$, then the spacetime is of Type~2.
\end{corollary}
\begin{proof}
  If the inequality of Eq.~\eqref{eq-Case1-sufficient-2-2} is satisfied, then 
  \begin{align}
    &\frac{d}{dv}\left(x_+(v)\exp{\left[\frac{h_c}{2}(1-\alpha)X_+(v)\right]}\right)\nonumber\\
    \,=\,&\exp{\left[\frac{h_c}{2}(1-\alpha)X_+(v)\right]}\left[\frac{dx_+}{dv}+\frac{h_c}{2}(1-\alpha)x_+^2\right]\,>\,0
  \end{align}
  holds. Therefore, $x_+(v)\exp{\left[\frac{h_c}{2}(1-\alpha)X_+(v)\right]}$ is a monotonically increasing function of $v$. Since this function is strictly positive, it satisfies the condition of Eq.~\eqref{eq-Case1-sufficient-2-1}, i.e., the spacetime is of Type~2.
\end{proof}

\begin{corollary}
  \label{corollary-Case1-1}
  If $(dx_+/dv)/x_+^2$ converges in the limit $v\to\infty$ and
  \begin{align}
    \label{eq-Case1-corollary-1}
    \lim_{v\rightarrow\infty}\frac{dx_+/dv}{x_+^2(v)}\,>\,-\frac{h_c}{2}
  \end{align}
  holds, then the spacetime is of Type~2. 
\end{corollary}
\begin{proof}
Defining
  \begin{align}
    A\coloneqq \lim_{v\rightarrow\infty}\frac{dx_+/dv}{x_+^2(v)},
  \end{align}
Eq.~\eqref{eq-Case1-corollary-1} is rewritten as $A>-\frac{h_c}{2}$. 
 %Under the condition of Eq.~\eqref{eq-Case1-corollary-1}, $A$ is a negative value or zero because $(dx_+/dv)/x_+^2$ is strictly negative for arbitrary $v\ge v_0$. 
 Here, by adopting a sufficiently small $\alpha$, 
 it is possible to make the inequality $A>-\frac{h_c}{2}(1-\alpha)$ satisfied. Therefore, for sufficiently large $v_0$, 
  \begin{align}
    \left|\frac{dx_+/dv}{x_+^2}-A\right|\,<\,A+\frac{h_c}{2}(1-\alpha)
  \end{align}
  holds for $v>v_0$ and this leads to
  \begin{align}
    -\frac{h_c}{2}(1-\alpha)x_+^2\,<\,\frac{dx_+}{dv}.
  \end{align}
  This inequality is equivalent to Eq.~\eqref{eq-Case1-sufficient-2-2}, i.e., the spacetime is of Type~2.
\end{proof}

\subsection{\label{subsec-case1-2}Some examples of Case~1 spacetime}

\begin{table}[tb]
  \centering
  \caption{Summary of examples of Case~1 spacetimes that realize Type~1 and Type~2.} 
  \scalebox{0.8}{
  \begin{tabular}{|c|c|} \hline
    Type~1 &Type~2\\\hline\hline
    $x_+=\frac{a_+}{v^n} \quad \left(n>1\right)$
    &$x_+=\frac{a_+}{v^n} \quad \left(0<n<1\right)$\\\hline
    $x_+=\frac{a_+}{v}\quad \left(a_+<\frac{2}{h_c}\right)$&
    $x_+=\frac{a_+}{v}\quad \left(a_+>\frac{2}{h_c}\right)$\\
    \hline
  \end{tabular} 
  }
  \label{table-Case1}
\end{table}

 By applying the propositions proved above, we can identify the 
 functional forms of $x_+(v)$ that lead to Type~1 and Type~2, respectively.
 Focusing on the power-law functions $x_+(v) = a_+/v^n$, we have obtained the results as 
 summarized in Table~\ref{table-Case1}.
 The proofs are presented in Appendix~\ref{Examples:case1}.
 The general tendency is that Type~1 and Type~2 spacetimes are realized when
 the decay of $x_+(v)$ is relatively fast (i.e., $0<n<1$) and slow ($n>1$), respectively.
 In the case of $n=1$, both types are possible depending on the value of the coefficient, $a_+$.

 \newpage

%
%======================================%
%<<<<<<<<<<<< SECTION IV  >>>>>>>>>>>>>>%
%======================================%
%
\section{\label{sec-case2}Case~2 : Spacelike inner AH}
 In this section, we study the global structures of Case~2 spacetimes (i.e., the case where the inner horizon is spacelike) using the same procedure as that for Case~1.

\begin{lemma}
  \label{lemma-case2-1}
  Outgoing null geodesics whose initial position is on or inside the inner AH never intersect the inner AH at any later time, and hence, stay inside the inner AH. They asymptotically approach $x=0$ in the limit $v\to\infty$.
\end{lemma}
 \begin{proof}
 Since $\frac{dx}{dv}= 0$ and $\frac{dx_-}{dv}>0$ are satisfied on the inner AH, at every point on the inner AH, the outgoing null geodesics cross the horizon from the outside region to the inside region.
Then, if an outgoing null geodesic whose initial position is on or inside the inner AH arrives at some point on the inner AH,
it immediately contradicts the uniqueness of solutions of first-order differential equations. 
Therefore, if an outgoing null geodesic is inside or on the inner AH at a given time, they will not intersect it at any later finite advanced time. 
      
      Next, we shall examine the asymptotic value of $x(v)$ of the outgoing null geodesics under consideration in the limit $v\rightarrow \infty$. Since we have $x(v)<x_-(v)$ for all $v>v_0$ and $\lim_{v\rightarrow\infty}x_-(v)=0$, we have $\lim_{v\rightarrow\infty}x(v)\leq 0$. 
      For the sake of contradiction, we assume $\lim_{v\rightarrow\infty}x(v)=x_{\mathrm{sup}} (< 0)$.
      The values of $dx/dv$ of these geodesics satisfy
      \begin{align}
        \lim_{v\rightarrow\infty}\frac{dx}{dv}=\frac{1}{2}h(\infty,r_{\mathrm{sup}})x_{\mathrm{sup}}^2>0.
      \end{align}
  This contradicts the assumption that $x_{\mathrm{sup}}$ be the limiting value of $v\rightarrow \infty$. 
  This leads to $x_{\mathrm{sup}}=0$, i.e., these outgoing null geodesics satisfy $\lim_{v\rightarrow \infty} x(v)=0$.
 \end{proof}  
It is also worth presenting the following lemma:
 \begin{lemma}
  \label{lemma-case2-2}
  Outgoing null geodesics in the region $x_-<x<0$ eventually cross the inner AH at later times. 
\end{lemma}
 \begin{proof}
 This is obvious because $dx/dv<0$ in this region, while $x_-$ asymptotes to zero. 
\end{proof}
From Lemmas~\ref{lemma-case2-1} and \ref{lemma-case2-2}, the black hole is present in each Case~2 spacetime,
and the event horizon is located in the region $x>0$.
 
\begin{figure}[t]
  \centering
  \includegraphics[width=0.5\linewidth]{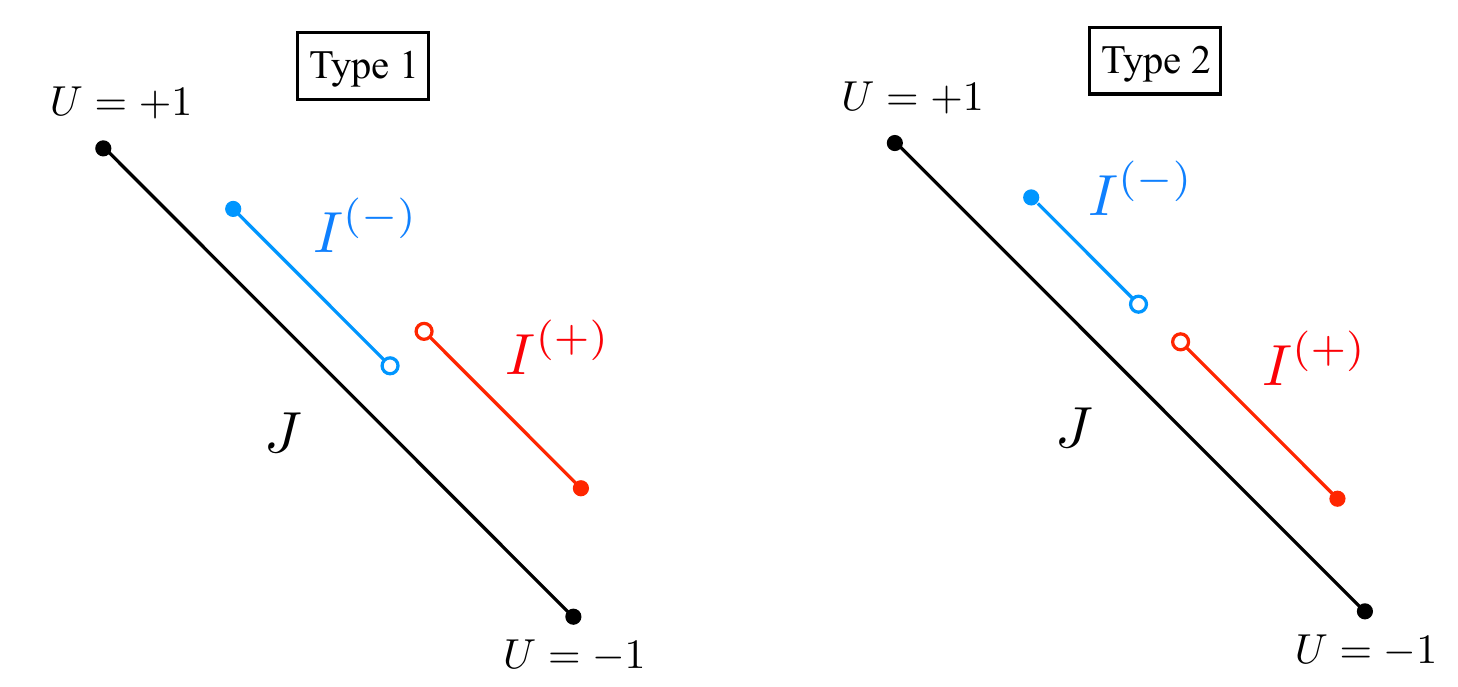}
  \caption{The same as Fig.~\ref{Rod-structure-case1} but for  Case~2. }
  \label{Rod-structure-case2}
\end{figure}

From Lemma \ref{lemma-1} and Lemma \ref{lemma-case2-1}, outgoing null geodesics crossing the inner AH 
stay in the inner AH at any later time, and outgoing null geodesics crossing the outer AH 
stay in the outer region at any later time. We now discuss the possible types of Penrose diagrams. 
Let us recall the three sets, $J$ and $I^{(\pm)}$ introduced in Sec.~\ref{Subsec:compactified-retarded}.
In Case~2, we have shown $I^{(-)}=\left(U^{(-)}_\infty,\,U^{(-)}_0 \right]$ and  $I^{(+)}=\left[U^{(+)}_0,\, U^{(+)}_\infty\right)$.
The above argument indicates that there must not be an overlap between $I^{(-)}$ and $I^{(+)}$.
Therefore, the inequality
\begin{equation}
-1<U^{(+)}_0<U^{(+)}_\infty\le U^{(-)}_\infty<U^{(-)}_0<+1
\end{equation}
must hold. Then, we define Type~1 and Type~2 of the spacetimes in Case~1 as follows:
\begin{dfn}
The spacetime of Case~2 is defined to be of Type~1 if  $U^{(+)}_\infty= U^{(-)}_\infty$ holds,
while the spacetime of Case~2 is defined to belong to Type~2 if  $U^{(+)}_\infty< U^{(-)}_\infty$ holds.
\label{Def:Case2-Type1-Type2}
\end{dfn}
The intervals $J$ and $I^{(\pm)}$ in Type~1 and Type~2 are depicted
in Fig.~\ref{Rod-structure-case2}. 
Physically, between the inner and outer AHs of a Type~1 spacetime,
there is a unique outgoing null geodesic that crosses 
neither the outer AH nor the inner AH. 
This is the outgoing null geodesics consisting of the event horizon.
On the other hand,  in a Type~2 spacetime, between the inner and outer AHs, there is a set of outgoing null geodesics with nonzero measure that do not cross
the two horizons, and the outermost one is the event horizon.
Since there is no possibility other than Type~1 and Type~2,
we have the following theorem:
\begin{thm}
  \label{thm-case2}
  The classification of Definition~\ref{Def:Case2-Type1-Type2} is complete
  in the sense that any spacetime of Case~2 belongs to either Type~1 or Type~2.
  \end{thm}

 From the above analyses, we know how all the outgoing null geodesics behave qualitatively
 once its Type is specified. Therefore, we can construct the Penrose diagram as shown in Figure \ref{pic-case2-diagrams}.\footnote{
 The Penrose diagram of a Type~1 spacetime in Case~2 can be found in existing literatures, e.g., Fig.~7 of \cite{Chen:2014jwq}.}

 \begin{figure}[t]
  \centering
  \includegraphics[width=0.5\linewidth]{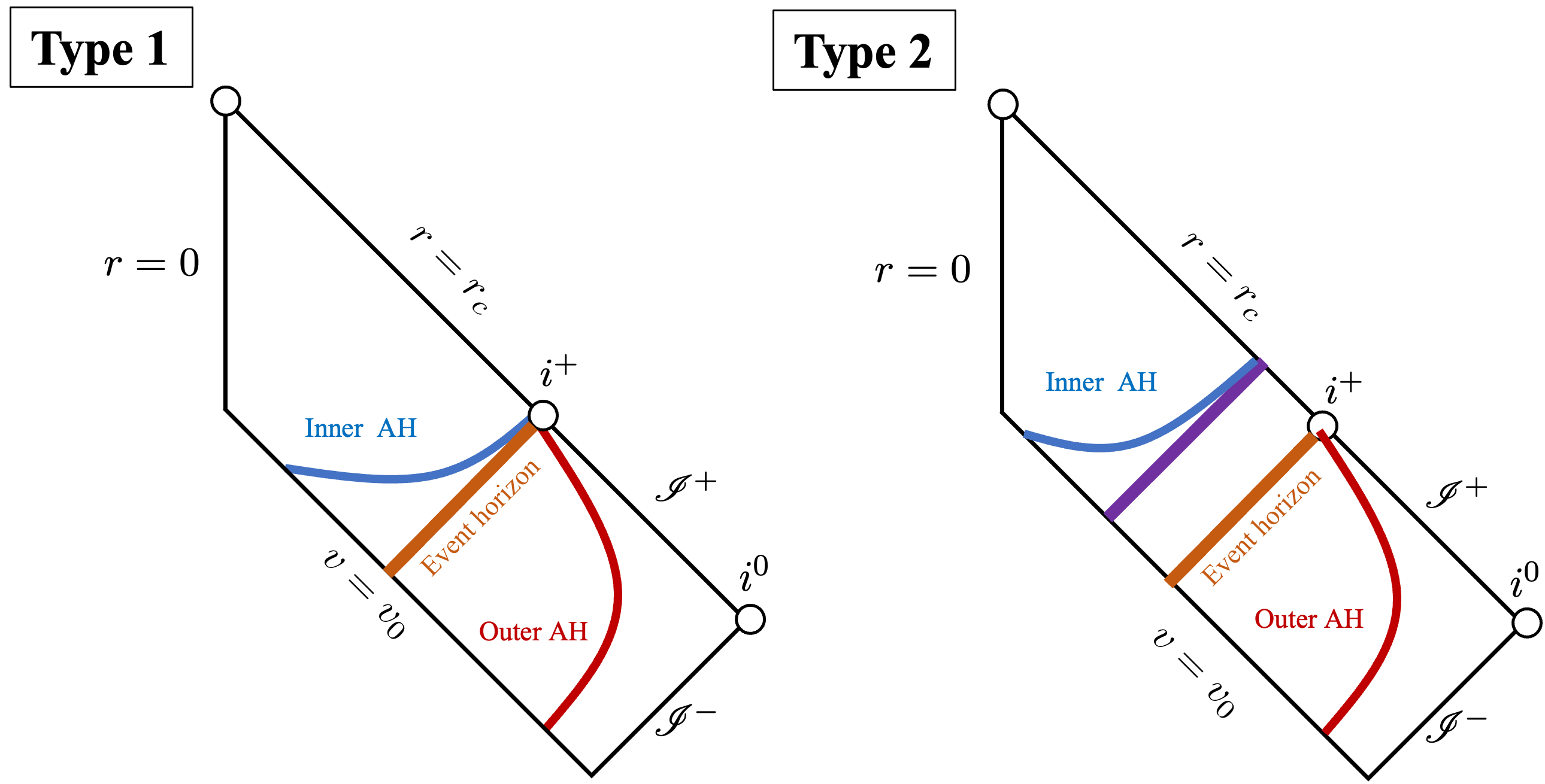}
  \caption{The Penrose diagrams of Case~2 spacetimes of Type~1 (left panel) and of Type 2 (right panel). 
  Similarly to Figure~\ref{pic-case1-diagrams}, the red, blue and orange curves represent the outer AH, the inner AH, and the event horizon, respectively. The purple line of Type~2 is an outgoing null geodesic that corresponds to $U=U^{(-)}_\infty$.}
  \label{pic-case2-diagrams}
\end{figure}

 \subsection{Sufficient conditions to classify as Type~1 or Type~2.}
 Here, we derive several sufficient conditions to determine which Penrose diagram corresponds to a given metric.  The following two propositions give the sufficient conditions for the spacetime to be of Type~1:

\begin{prop}
  \label{prop-case2-Type1-2}
  If there exist positive constants $\varepsilon_\pm$ for which
  \begin{align}
    \label{eq-case2-Type1-2-1}
    \lim_{v\rightarrow\infty}x_+\exp{\left[\frac{h_c}{2}\left\{(1+\varepsilon_+)X_+(v)+(1-\varepsilon_-)X_-(v)\right\}\right]}\,=\,0
  \end{align}
  is satisfied,
   then the spacetime is classified as Type~1.
\end{prop}
\begin{proof}
 For the sake of contradiction, we assume the existence of outgoing null geodesics with nonzero measure
 that stay in the region $x_-(v)<x(v)<x_+(v)$.  
 We denote these outgoing null geodesics as $x(v)$, and the innermost one as $x_{\rm A}(v)$ (that corresponds to $U=U^{(-)}_\infty$). 
 As they do not intersect the inner AH, they satisfy the condition $x(v)>0$ and $x_{\rm A}(v)>0$
 from Lemma~\ref{lemma-case2-2}. 
 The function $x(v)$ satisfies Eq.~\eqref{eq-null-x}, and $x_{\mathrm{A}} (v)$ satisfies 
 \begin{align}
  \frac{dx_{\mathrm{A}}}{dv}=\frac{h(v,x_{\mathrm{A}})}{2}(x_{\mathrm{A}}-x_+)(x_{\mathrm{A}}-x_-).
 \end{align}
   Introducing $\zeta$ by $\zeta(v)\coloneqq x(v)-x_{\mathrm{A}}$, the equation for $\zeta$ is
  \begin{align}
    \frac{d\zeta}{dv}
    &\,=\,\frac{h(v,x)}{2}[x_{\mathrm{A}}+x-(x_++x_-)]\zeta+\frac{h(v,x)-h(v,x_{\mathrm{A}})}{2}(x_{\rm A}-x_+)(x_{\rm A}-x_-).
  \end{align}
  Here, from Rolle's theorem, there exists $x_{\rm I}$ satisfying $x_{\rm A}<x_{\rm I}<x$ such that 
  \begin{align}
    \partial_x h(v,x_{\rm I})=\frac{h(v,x)-h(v,x_{\rm A})}{\zeta}
  \end{align}
  holds. Then, the equation can be rewritten as
  \begin{equation}
  \frac{d\zeta}{dv}\, = \, 
  \frac{\zeta}{2}
  \left[h(v,x)x+\left(h(v,x)-h_{,x}(v,x_{\rm I})x_{\rm A}\right)(x_{\rm A}-x_+-x_-)
  -h_{,x}(v,x_{\rm I})x_+x_-\right],
  \end{equation}
  Since $x$, $x_{\rm I}$, and $x_{\rm A}$ converge to zero while $h_{,x}(v,x)$ remains finite in the limit $v\to\infty$, we have $\lim_{v\to\infty}h(v,x)-h_{,x}(v,x_{\rm I})x_{\rm A} = h_c$,
  and thus, the positivity of $\left(h(v,x)-h_{,x}(v,x_{\rm I})x_{\rm A}\right)x_{\rm A}$ is guaranteed for $v>v_0$ with a sufficiently large $v_0$.
  Together with $h(v,x)x>0$, we have the inequality, 
  \begin{equation}
  \frac{d\zeta}{dv}\, > \, 
  -\frac{\zeta}{2}\left\{
  \left[h(v,x)-h_{,x}(v,x_{\rm I})(x_{\rm A}-x_-)\right]x_+
  +\left[h(v,x)-h_{,x}(v,x_{\rm I})x_{\rm A}\right]x_-
  \right\}.
  \end{equation}
  Since the functions inside the square brackets of the right-hand side both converge to $h_c$ in the limit $v\to 0$, 
  we have 
  \begin{subequations}
  \begin{eqnarray}
h(v,x)-h_{,x}(v,x_{\rm I})(x_{\rm A}-x_-)&\,<\,&h_c(1+\varepsilon_+),\\ 
  h(v,x)-h_{,x}(v,x_{\rm I})x_{\rm A}&\,>\,&h_c(1-\varepsilon_-),
  \end{eqnarray}
  \end{subequations}
  and hence,
  \begin{equation}
  \frac{d\zeta}{dv}\, > \, 
  -\frac{h_c}{2}
  \left[(1+\varepsilon_+)x_+
  +(1-\varepsilon_-)x_-
  \right]\zeta.
  \end{equation}
Dividing both sides by $\zeta$ and 
  integrating the resultant inequality, we have
  \begin{align}
    \zeta\,>\,\zeta_0\exp{\left[-\frac{h_c}{2}\left\{(1+\varepsilon_+)\left(X_+(v)-X_+(v_0)\right)+(1-\varepsilon_-)\left(X_-(v)-X_-(v_0)\right)\right\}\right]},
  \end{align}
  where $\zeta_0$ is defined as $\zeta_0\coloneqq \zeta(v_0)$. 
  Since $x_+>\zeta$ holds, we have
  \begin{align}
    x_+\exp{\left[\frac{h_c}{2}\left\{(1+\varepsilon_+)X_+(v)+(1-\varepsilon_-)X_-(v)\right\}\right]}
        \,>\,
        \zeta_0\exp{\left[\frac{h_c}{2}\left\{(1+\varepsilon_+)X_+(v_0)+(1-\varepsilon_-)X_-(v_0)\right\}\right]}.
  \end{align}
  However, the left-hand side converges to zero due to Eq.~\eqref{eq-case2-Type1-2-1}.
  This is a contradiction. Since this contradiction is caused by the assumption of the existence of
  outgoing null geodesics with nonzero measure in the range $x_-<x<x_+$, there is unique 
  outgoing null geodesic in that region, i.e., the spacetime is of Type~1.
\end{proof}
Proposition \ref{prop-case2-Type1-2} leads to the following two corollaries.
\begin{corollary}
  \label{corollary-case2-Type1-1}
  If there exist two positive constants $\varepsilon_{\pm}$  for which
  \begin{align}
    \label{eq-corollary-case2-Type1-1}
    \lim_{v\rightarrow\infty}\left[(1+\varepsilon_+)X_++(1-\varepsilon_-)X_-\right]\,\neq\,\infty
  \end{align}
  holds, then the spacetime is of Type~1.
\end{corollary}
 \begin{proof}
 When Eq.~\eqref{eq-corollary-case2-Type1-1} holds, the argument of the exponential function in Eq. \eqref{eq-case2-Type1-2-1} remains finite or diverges to $-\infty$. 
  Then, the exponential function remains finite. 
  Since $x_+$ asymptotes to zero for $v\to\infty$, the condition of Eq.~\eqref{eq-case2-Type1-2-1} is obviously satisfied, i.e., this spacetime is of Type~1.
 \end{proof}
\begin{corollary}
  \label{corollary-case2-Type1-2}
  If 
  \begin{align}
    \label{eq-case2-R}
    \lim_{v\rightarrow\infty}\frac{x_+}{|x_-|}\eqqcolon R\,<\,1
  \end{align}
  holds, then the spacetime is of Type~1.
\end{corollary}
\begin{proof}
  Choosing sufficiently small $\varepsilon_\pm$, it is possible to make the inequality 
  $R<(1-\varepsilon_-)/(1+\varepsilon_+)$ satisfied. Then, the inequality $x_+/|x_-| < (1-\varepsilon_-)/(1+\varepsilon_+)$ holds for $v>v_0$ with a sufficiently large $v_0$ from Eq.~\eqref{eq-case2-R}. 
  This inequality is rewritten as
  \begin{align}
    (1+\varepsilon_+)x_+-(1-\varepsilon_-)|x_-|\,<\,0.
  \end{align}
  Recalling the definitions of $X_\pm$ given in Eq.~\eqref{Primitive-function-xpm},
  $(1+\varepsilon_+)X_++(1-\varepsilon_-)X_-$ never diverges to $+\infty$. 
  Then, the condition of Corollary~\ref{corollary-case2-Type1-1} is satisfied,
   indicating that this spacetime is of Type~1.
\end{proof}
\begin{prop}
  \label{prop-Case2-Type1-3}
   If $x_-(v)$ satisfies
  \begin{align}
    \label{eq-Case2-Type1-3}
    \lim_{v\rightarrow\infty}X_-(v)\,=\,-\infty,
  \end{align}
   then, the spacetime is of Type~1.
\end{prop}
\begin{proof}
  For the sake of contradiction, we assume the existence of outgoing null geodesics with non-zero measure that stay in the region $x_-<x<x_+$. Suppose $x=x_{\rm EH}(v)$ to be the position of the event horizon, 
  which is the outermost null geodesics staying in this region satisfying
 \begin{align}
  \frac{dx_{\mathrm{EH}}}{dv}\,=\,\frac{h(v,x_{\mathrm{EH}})}{2}(x_{\mathrm{EH}}-x_+)(x_{\mathrm{EH}}-x_-).
 \end{align}
 The outgoing null geodesics between the two horizons obey Eq.~\eqref{eq-null-x},
 and they must stay in the region $0<x<x_{\rm EH}$ from Lemma~\ref{lemma-case2-2}. 
We define
  \begin{equation}
  \xi \ = \ 1-\frac{x}{x_{\rm EH}},
  \end{equation}
  which takes the value $0<\xi<1$. After some algebra, the quantity $\xi$ is found to satisfy the equation
  \begin{align}
    \frac{d\xi}{dv}
    \,=\,
    \frac{1}{2x_{\rm EH}^2}
    \left[
     \left(h(v,x_{\rm EH})-h(v,x)\right)x(x_{\rm EH}-x_+)(x_{\rm EH}-x_-)
     +h(v,x)(xx_{\rm EH}-x_+x_-)x_{\rm EH}\xi
    \right].
  \end{align}
  From Rolle's theorem, there is a value $x_{\rm I}$ in the range
 $x\le x_{\rm I}\le x_{\rm EH}$ satisfying $h(v,x_{\rm EH})-h(v,x)=h_{,x}(v,x_{\rm I})x_{\rm EH}\xi$.
 Substituting this formula, we obtain
 \begin{equation}
 \frac{d\xi}{dv}\,=\,
 \frac{\xi}{2}\left[h(v,x)-h_{,x}(v,x_{\rm I})(x_+-x_{\rm EH})\right]x
 -\frac{\xi}{2}\frac{x_+x_-}{x_{\rm EH}}\left[h(v,x)-h_{,x}(v,x_{\rm I})x\left(1-\frac{x_{\rm EH}}{x_+}\right)\right].
 \end{equation}
 Since $x$, $x_+$, $x_{\rm EH}$ and $x_{\rm I}$ converge to zero (keeping $0<x_{\rm EH}/x_+<1$) 
 while $h_{,x}(v,x)$ remains finite in the limit $v\to\infty$,
 both of the formulas in the square brackets of the first and second terms on the right-hand side
 converge to $h_c$ in the limit $v\to\infty$. Therefore,  
 the positivity of the first term on the right-hand side is guaranteed, 
 and for some small positive constant $\varepsilon$,
  \begin{align}
   h(v,x)-h_{,x}(v,x_{\rm I})x\left(1-\frac{x_{\rm EH}}{x_+}\right)\,>\, (1-\varepsilon)h_c
  \end{align}
  holds in the range $v>v_0$ for a sufficiently large $v_0$. As a result, we obtain
  \begin{align}
    \frac{d\xi}{dv}&
    \,>\,\frac{h_c(1-\varepsilon)}{2}\left(\frac{x_+|x_-|}{x_{\mathrm{EH}}}\right)\xi
    \,>\,\frac{h_c(1-\varepsilon)}{2}|x_-|\xi,\nonumber
  \end{align}
  where we used $x_+>x_{\rm EH}>0$ in the second inequality. 
  Dividing both sides by $\xi$ and integrating the resultant inequality, we have
  \begin{align}
    \label{eq-xi-integrate}
    \xi&\,>\,\xi_0\exp{\left[-\frac{h_c(1-\varepsilon)}{2}\left(X_-(v)-X_-(v_0)\right)\right]},
  \end{align}
  where $\xi _0$ is defined as $\xi_0\coloneqq \xi(v_0)$. Since  $\xi$ satisfies $\xi< 1$,
  the following inequality holds:
  \begin{align}
    \exp{\left[\frac{h_c(1-\varepsilon)}{2}X_-(v)\right]}\,>\,\xi_0\exp{\left[\frac{h_c(1-\varepsilon)}{2}X_-(v_0)\right]}.
  \end{align}
  However, the left-hand side converges to zero from  Eq.~\eqref{eq-Case2-Type1-3}. 
  This is a contradiction. Since this contradiction is caused by the assumption of the existence of
  outgoing null geodesics between the two horizons with nonzero measure,
  there can exist a unique outgoing geodesic in that region. 
  Consequently, the spacetime is of Type~1.
\end{proof}

Next, we prove a proposition that gives a sufficient condition for a given spacetime to be of Type~2.

\begin{prop}
  \label{prop-Case2-Type2-1}
  Suppose  $x_\pm$  satisfy
  \begin{align}
    \lim_{v\rightarrow\infty} \frac{x_-(v)}{x_+(v)}=0.
    \label{limit-of-x-dbx+}
  \end{align}
  If there exist two positive parameters  $\alpha,\beta$ satisfying $0<\alpha<1 $ and $0<\beta$ and two positive parameters $\varepsilon_{\alpha},\varepsilon_{\beta}$ for which the two conditions
  \begin{subequations} 
  \begin{align}
    \label{eq-Case2-Type2-1-1}
    \frac{dx_-/dv}{(x_++\beta x_-)x_-}&\,<\, -\frac{h_c}{2}(1+\varepsilon_{\beta})\frac{1+\beta}{\beta},
    \\
        \label{eq-Case2-Type2-1-2}
    \frac{dx_+/dv}{(\alpha x_+-x_-)x_+}&\,>\,-\frac{h_c}{2} (1-\varepsilon_{\alpha})\frac{1-\alpha}{\alpha}
  \end{align}
  \end{subequations}
  hold for $v>v_0$ for some $v_0$, then the spacetime is of Type~2.
\end{prop}
\begin{proof}
  Since
  \begin{align}
    \lim_{v\rightarrow\infty}\frac{\beta x_-(v)}{\alpha x_+(v)}=0
  \end{align}
  holds from the assumption, $\alpha x_+(v)>\beta |x_-(v)|$ holds for $v>v_0$ for a sufficiently large $v_0$. 
  Denoting the curves $x(v)=x_\alpha(v):=\alpha x_+(v)$ and $x(v)=x_\beta(v):=-\beta x_-(v)$ as $C_{\alpha}$ and $C_{\beta}$, respectively, we prove that outgoing null geodesics in the region between $C_{\alpha}$ and $C_{\beta}$ remain in this region by showing that the outgoing null geodesics cross $C_\beta$ just from the inside region to the outside region, while they cross $C_\beta$ just from the outside region to the inside region. 
  At each point on $C_\beta$, the derivative $dx/dv$ of an outgoing null geodesic that cross $C_\beta$ at that point is
  \begin{align}
    \left.\frac{dx}{dv}\right|_{C_{\beta}}=\frac{h(v,-\beta x_-(v))}{2}(1+\beta)(x_++\beta x_-)x_-.
  \end{align}
  Since $\lim_{v\rightarrow\infty}h(v,-\beta x_-(v))=h_c$ holds, for $v>v_0$, the inequality $h(v,-\beta x_-)<(1+\varepsilon_{\beta})h_c$ is satisfied by taking a sufficiently large $v_0$. Then, we obtain
  \begin{align}
    \label{eq-derivative-on-C_beta}
    \left.\frac{dx}{dv}\right|_{C_{\beta}}\,>\,\frac{h_c}{2}(1+\varepsilon_\beta)(1+\beta)(x_++\beta x_-)x_-
    \,>\,-\beta \frac{dx_-}{dv},
  \end{align}
  where we used the inequality of Eq.~\eqref{eq-Case2-Type2-1-1} in the second inequality.  
  Since the derivative of the curve $C_\beta$ is given by $dx_\beta/dv = -\beta dx_-/dv$, we have
    \begin{align}
      \frac{dx_\beta}{dv}&\,<\,\left.\frac{dx}{dv}\right|_{C_{\beta}}.
  \end{align}
  Then outgoing null geodesics intersect $C_{\beta}$ just from the inside region to the outside region, i.e., outgoing null geodesics in the region between $C_\alpha$ and $C_\beta$ never intersect $C_\beta$. 
  Similarly, at each point on  $C_\alpha$, an outgoing null geodesic crossing $C_\alpha$ at that point satisfies
  \begin{align}
    \left.\frac{dx}{dv}\right|_{C_{\alpha}}=-\frac{h(v,\alpha x_+(v))}{2}(1-\alpha)(\alpha x_+-x_-)x_+.
  \end{align}
  Since $\lim_{v\rightarrow \infty} h(v,\alpha x_+(v))=h_c$ holds, for $v>v_0$, the inequality $h(v,\alpha x_+)>(1-\varepsilon_{\alpha})h_c$ is satisfied by taking a sufficiently large $v_0$.  Then we obtain 
  \begin{align}
    \label{eq-derivative-on-C_alpha}
    \left.\frac{dx}{dv}\right|_{C_{\alpha}}\,<\,-\frac{h_c}{2}(1-\varepsilon_\alpha)(1-\alpha)(\alpha x_+-x_-)x_+
    \,<\, \alpha \frac{dx_+}{dv},
  \end{align}
  where we used  the inequality of Eq.~\eqref{eq-Case2-Type2-1-2} in the second inequality.
  Since the derivative of the curve $C_\alpha$ is given by $dx_\alpha/dv = \alpha dx_+/dv$, we have
  \begin{align}
    \frac{dx_\alpha}{dv}&>\left.\frac{dx}{dv}\right|_{C_\alpha}.
  \end{align}
  Then outgoing null geodesics intersect $C_{\alpha}$ just from the outside region to the inside region, i.e., outgoing null geodesics inside $C_{\alpha}$ never intersect it. Therefore, outgoing null geodesics in the region between $C_{\alpha}$ and $C_{\beta}$ remain in this region. This means the existence of outgoing null geodesics
  staying between the two AHs with nonzero measure, i.e., 
   the spacetime is of Type~2.
 
\end{proof}

This proposition leads to the following corollary:
\begin{corollary}
  \label{corollary-case2-Type2-1}
  We assume that $x_\pm$ satisfy Eq.~\eqref{limit-of-x-dbx+}. 
  If there exist two positive parameters $\alpha$ and $\beta$ satisfying $0<\alpha<1$ and $0<\beta$ for which 
  the two conditions
  \begin{subequations}
  \begin{align}
     \label{eq-case2-Type2-2-1}
    \lim_{v\rightarrow \infty}\frac{dx_+/dv}{x_+^2}&\,>\,-\frac{h_c}{2}(1-\alpha),\\
      \label{eq-case2-Type2-2-2}
  \lim_{v\rightarrow\infty}\frac{dx_-/dv}{x_+x_-}&\,<\,-\frac{h_c}{2}\frac{1+\beta}{\beta}
    \end{align}
    \end{subequations}
  hold, then the spacetime is of Type~2. Here, the limit of the left-hand side of Eq.~\eqref{eq-case2-Type2-2-1} is supposed to have a zero or finite negative value, while the limit of the left-hand side of Eq.~\eqref{eq-case2-Type2-2-2} is supposed to be $-\infty$ or to have a finite negative value.
\end{corollary}
 
\begin{proof}
  Since $\lim_{v\rightarrow\infty}{x_-}/{x_+}=0$, the inequality of Eq.~\eqref{eq-case2-Type2-2-1} can be rewritten as 
  \begin{align}
    \lim_{v\rightarrow\infty}\frac{dx_+/dv}{(\alpha x_+-x_-)x_+}\eqqcolon A_{\alpha}\,>\,-\frac{h_c}{2}\frac{1-\alpha}{\alpha}.
    \label{Limit_Aalpha}
  \end{align}
  From this condition, for an arbitrarily small positive parameter $\varepsilon_\alpha^\prime$,
  \begin{align}
    \label{eq-case2-Type2-varepsilon_1}
    A_{\alpha}(1+\varepsilon_{\alpha}')\,<\,\frac{dx_+/dv}{(\alpha x_+-x_-)x_+}
  \end{align}
  holds for sufficiently large $v$. On the other hand, from the inequality of Eq.~\eqref{Limit_Aalpha}, it is possible to choose sufficiently small two parameters $\varepsilon_{\alpha}$ and $\varepsilon_{\alpha}'$ which satisfy
  \begin{align}
    \label{eq-cases2-Type2-varepsilon_2}
    -\frac{h_c}{2}(1-\varepsilon_{\alpha})\frac{1-\alpha}{\alpha}<A_\alpha(1+\varepsilon_\alpha').
  \end{align}
 Combining the two inequalities of Eqs.~\eqref{eq-case2-Type2-varepsilon_1} and \eqref{eq-cases2-Type2-varepsilon_2},
 we find the inequality which is exactly identical to Eq.~\eqref{eq-Case2-Type2-1-2}. 

  Similarly, the inequality of Eq.~\eqref{eq-case2-Type2-2-2} can be rewritten as 
  \begin{align}
    \lim_{v\rightarrow\infty}\frac{dx_-/dv}{(x_++\beta x_-)x_-}\eqqcolon A_{\beta}\,<\,-\frac{h_c}{2}\frac{1+\beta}{\beta}.
    \label{Limit_Abeta}
  \end{align}
  In the case of $A_{\beta}=-\infty$, the inequality of Eq.~\eqref{eq-Case2-Type2-1-1} holds trivially. Hereafter, we only consider the situation where $A_{\beta}$ is finite. From this condition, for an arbitrarily small positive parameter $\varepsilon_{\beta}'$,
  \begin{align}
    \label{eq-case2-Type2-varepsilon-3}
    \frac{dx_-/dv}{(x_++\beta x_-)x_-}\,<\,A_\beta (1-\varepsilon_\beta')
  \end{align}
  holds for sufficiently large $v$. On the other hand, from the inequality of Eq.~\eqref{Limit_Abeta}, it is possible to choose sufficiently small two parameters $\varepsilon_{\beta}$ and $\varepsilon_{\beta}'$ which satisfy
  \begin{align}
    \label{eq-case2-Type2-varepsilon-4}
    A_{\beta}(1-\varepsilon_{\beta}')<-\frac{h_c}{2}(1+\varepsilon_\beta)\frac{1+\beta}{\beta}.
  \end{align}
  Combining the two inequalities of Eqs.~\eqref{eq-case2-Type2-varepsilon-3} and \eqref{eq-case2-Type2-varepsilon-4},
  we find the inequality which is exactly identical to Eq.~\eqref{eq-Case2-Type2-1-1}. Consequently, the conditions of Proposition~\ref{prop-Case2-Type2-1} are satisfied.
\end{proof}

\subsection{Some examples of Case~2 spacetimes}

\begin{table}[tbh]
  \centering
  \caption{Summary of examples of Case~2 spacetimes that realize Type~1 and Type~2.} 
  \scalebox{0.8}{
  \begin{tabular}{|c|c|} \hline
    Type~1 &Type~2\\\hline\hline
         $x_+=\frac{a_+}{v^n},x_-=\frac{a_-}{v^k}\quad \left(0<k<n\right)$
    \ & 
    $x_+=\frac{a_+}{v^n}, x_-=a_-\exp{\left(-\frac{v}{L}\right)}\quad (0<n<1)$ \\ \hline
      $x_+=\frac{a_+}{v}, x_-=\frac{a_-}{v^k}\quad \left(a_+h_c<2<2k\right)$
      &
      $x_+=\frac{a_+}{v},x_-=\frac{a_-}{v^k}\quad \left(2<a_+h_c<2k\right)$\\
      \hline
      $x_+=\frac{a_+}{v^n},x_-=\frac{a_-}{v^k} \quad(n>1,k>1)$
      &\ \\\hline
      $x_+=(\text{arbitrary}), x_-=\frac{a_-}{v^k}\quad (0<k\leq 1)$&\ \\
    \hline
  \end{tabular} 
  }
  \label{table-Case2}
\end{table}

 By applying the propositions \ref{prop-case2-Type1-2}, \ref{prop-Case2-Type1-3},  \ref{prop-Case2-Type2-1}
 and their corollaries proved above, we can identify the 
 functional forms of $x_\pm(v)$ that lead to Type~1 and Type~2, respectively.
 Focusing mainly on the power-law functions $x_+(v) = a_+/v^n$ and $x_-(v)=a_-/v^k$
 (with the exception of the exponential function of the first line of Type~2), we have obtained the results as 
 summarized in Table~\ref{table-Case2}.
 The proofs are presented in Appendix~\ref{Examples:case2}.
 From this table, it is found that the Type~2 spacetimes are realized only when
 $x_+(v)$ decays relatively slowly and $x_-(v)$ decays rapidly.
 This would be a natural result as understood by comparing with the result of Case~1: 
 Since the Type~2 spacetimes in Case~2
 have a similar structure to the Type~2 spacetimes in Case~1,
 the Type~2 spacetimes in Case~2 are expected to be realized
 only when $x_+(v)$ decays slowly and the inner AH
 is sufficiently close to a null surface.
 The spacetimes that do not satisfy these conditions are expected to belong to Type~1. 
 Unfortunately, we have missed to include the case
 of $x_+=a_+/v^n$ and $x_-=a_-/v^k$ with $0<n<1<k$, and the case of 
 $x_+=a_+/v$ and $x_-=a_-/v^k$ with $2<2k<a_+h_c$ in this table,
 because Propositions \ref{prop-case2-Type1-2} and \ref{prop-Case2-Type1-3}
 cannot cover these situations. 
 Note that these situations do not satisfy the sufficient conditions to be
 of Type~2 given in Proposition \ref{prop-Case2-Type2-1}.
 In other words, they satisfy a necessary condition to be of Type~1.
 We conjecture that these situations would belong to Type~1, 
 although the explicit proof is left as a future issue.

\newpage

%
%======================================%
%<<<<<<<<<<<< SECTION V  >>>>>>>>>>>>>>%
%======================================%
%
\section{\label{sec-case3}Case~3: Timelike inner AH}
In this section, we study the global structure of Case~3 spacetimes (the spacetime whose inner horizon is timelike) in a similar way to Cases 1 and 2. 
Case~3 spacetimes generally satisfy the following proposition:
\begin{lemma}
  \label{lemma-case3}
  %Outgoing null geodesics whose initial position is between the inner AH and the outer AH never intersect the inner AH.
  Outgoing null geodesics initially located between the inner and outer AHs or on the inner AH never intersect the inner AH at any later time.
\end{lemma}
\begin{proof}
  The proof is omitted here since the procedure of the proof is the same as that presented in 
  the former part of the proof of Lemma \ref{lemma-1}.
\end{proof}
It is also worth presenting the following lemma:
 \begin{lemma}
  \label{lemma-case3-2}
  Outgoing null geodesics in the region $0<x<x_-$ eventually cross the inner AH at later times. 
\end{lemma}
 \begin{proof}
 This is obvious because $dx/dv>0$ in this region, while $x_-$  asymptotes to zero. 
\end{proof}
 
Let us recall the three sets, $J$ and $I^{(\pm)}$ introduced in Sec.~\ref{Subsec:compactified-retarded}
in order to discuss the properties of Penrose diagrams in Case~3.
In Case~3, we have shown $I^{(-)}=\left[U^{(-)}_0,U^{(-)}_\infty \right)$ and  $I^{(+)}=\left[U^{(+)}_0,\, U^{(+)}_\infty\right)$.
Unlikely to Lemma~\ref{lemma-case1} in Case~1 and Lemma~\ref{lemma-case2-1} in Case~2,
Lemma~\ref{lemma-case3} does not impose any restriction on the positions of the intervals $I^{(\pm)}$,
because the outgoing null geodesics that cross the inner AH have the possibility of
both reaching and not reaching the outer AH.
One of the restrictions on the quantities $U^{(\pm)}_0$ and $U^{(\pm)}_\infty$ 
is 
\begin{equation}
U^{(+)}_\infty \,\le\, U^{(-)}_\infty \,\le\, +1,
\end{equation}
which comes from Eq.~\eqref{Upm-constraint}. The other constraints are trivial ones,
$U^{(+)}_0\,<\, U^{(+)}_\infty$, $U^{(-)}_0\,<\, U^{(-)}_\infty$, and
$-1\,<\,U^{(+)}_0\,<U^{(-)}_0$ which also comes from  Eq.~\eqref{Upm-constraint}.
In order to effectively classify the spacetimes under these constraints, first we define Class~1 and Class~2 of the spacetimes in Case~3 as follows:
\begin{dfn}
The spacetime of Case~3 is defined to be of Class~1 if  $U^{(+)}_\infty= U^{(-)}_\infty$ holds,
while the spacetime of Case~3 is defined to belong to Class~2 if  $U^{(+)}_\infty< U^{(-)}_\infty$ holds.
\label{Def:Case3-Class1-Class2}
\end{dfn}
Physically, in Class~1 spacetimes, all outgoing null geodesics emitted from the inner AH
reach the outer AH, while in Class~2 spacetimes, some of them do not arrive at
the outer AH.
Next, we define Class~A and Class~B for the spacetimes in Case~3 as follows:
\begin{dfn}
The spacetime of Case~3 is defined to be of Class~A if  $U^{(-)}_\infty=+1$ holds,
while the spacetime of Case~3 is defined to belong to Class~B if  $U^{(-)}_\infty<+1$ holds.
\label{Def:Case3-ClassA-ClassB}
\end{dfn}
Physically, in Class~A spacetimes, all outgoing null geodesics in the inside region of the inner AH
reach the inner AH, while in Class~B spacetimes, some of them do not arrive at
the inner AH. Finally, we define the Type~1-A, Type~1-B, Type~2-A, and Type~2-B of the
Case~3 spacetimes as follows:
\begin{dfn}
The spacetime of Case~3 is defined to be of Type~1-A if it belongs to
Class~1 and Class~A at the same time. The Type~1-B, Type~2-A, and Type~2-B
are defined in the same way.
\label{Def:Case3-Types}
\end{dfn}
To summarize, the Type~1-A, Type~1-B, Type~2-A, and Type~2-B spacetimes
satisfy the following conditions:
\begin{equation}
\begin{cases}
\textrm{Type~1-A:} & -1\,<\,U^{(+)}_0\,<\,U^{(-)}_0\,<\,U^{(+)}_\infty \, = \, U^{(-)}_\infty\,= \, +1;\\
\textrm{Type~1-B:} & -1\,<\,U^{(+)}_0\,<\,U^{(-)}_0\,<\,U^{(+)}_\infty \, = \, U^{(-)}_\infty\,< \, +1;\\
\textrm{Type~2-A:} & -1\,<\,U^{(+)}_0\,<\,U^{(+)}_\infty \, < \, U^{(-)}_\infty\,= \, +1;\\
\textrm{Type~2-B:} & -1\,<\,U^{(+)}_0\,<\,U^{(+)}_\infty \, < \, U^{(-)}_\infty\,< \, +1.
\end{cases}
\label{Four-Types-Inequalities}
\end{equation}
It is worth noting that because of Corollary~\ref{corollary-Sec2}, 
the black hole is absent in the spacetime of Type~1-A, while each spacetime of the other types
possesses the event horizon.
The intervals $J$ and $I^{(\pm)}$ of all types are depicted in Fig.~\ref{Rod-structure-case3}.
As we will show below, all of these types of spacetimes exist.

In the inequality for Type~2-A in Eq.~\eqref{Four-Types-Inequalities}, 
the quantity $U^{(-)}_0$ is not included. If we include this quantity, it may be possible to 
consider the subtypes as
\begin{equation}
\begin{cases}
\textrm{Type~2-A(i):} & -1\,<\,U^{(+)}_0\,<\,U^{(-)}_0\,<\,U^{(+)}_\infty \, < \, U^{(-)}_\infty\,= \, +1;\\
\textrm{Type~2-A(ii):} & -1\,<\,U^{(+)}_0\,<\,U^{(-)}_0\,=\,U^{(+)}_\infty \, < \, U^{(-)}_\infty\,= \, +1;\\
\textrm{Type~2-A(iii):} & -1\,<\,U^{(+)}_0\,<\,U^{(+)}_\infty \,<\,U^{(-)}_0\,< \, U^{(-)}_\infty\,= \, +1.
\end{cases}
\label{Subtypes-Type2A-Inequalities}
\end{equation}
However, we do not consider that such subtypes are very meaningful for the following reason. 
Suppose we have a spacetime of Type~2-A(iii). Then, let us consider what happens if the value
of $v_0$ is changed. Since the value of $U^{(-)}_0$ approaches $U^{(-)}_\infty$ as $v_0$ is increased while
the value of $U^{(+)}_\infty$ is unchanged,
the hypersurface $v=v_0$ that realizes Type~2-A(i)
can be taken by adopting a sufficiently large $v_0$. For this reason, we do not consider that there are fundamental differences
in these subtypes, and thus, we treat these three subtypes altogether. The similar statement
holds for Type~2-B as well.

\begin{figure}[t]
  \centering
  \includegraphics[width=0.6\linewidth]{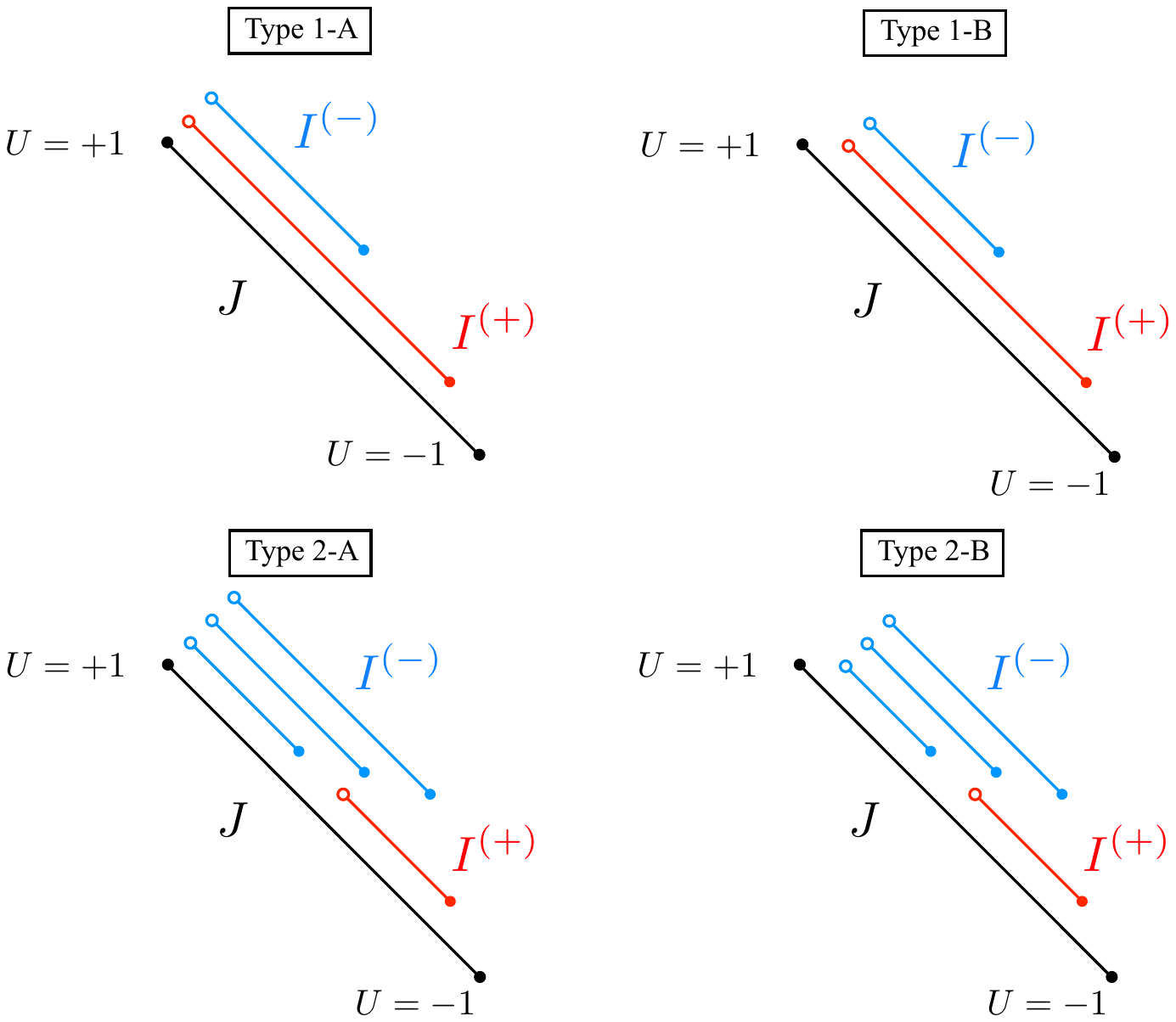}
  \caption{The intervals of $J$ and $I^{(\pm)}$ for Tyep 1-A (top-left),
  Type~1-B (top-right), Type~2-A (bottom-left), and Type~2-B (bottom-right) in Case~3. }
  \label{Rod-structure-case3}
\end{figure}

To summarize, we have developed the method for classifying Case~3 spacetimes, and as a result,
four types of the Case~3 spacetimes have been defined. 
Since there is no possibility other than these four types in Case~3 spacetimes,
we have the following theorem;
\begin{thm}
  \label{thm-case3}
  The classification by Definitions~\ref{Def:Case3-Class1-Class2}, \ref{Def:Case3-ClassA-ClassB}, 
  and \ref{Def:Case3-Types} is complete
  in the sense that any spacetime of Case~3 belongs to either Type~1-A, Type~1-B, Type~2-A, or Type~2-B.
  \end{thm}
Lemmas \ref{lemma-1} and \ref{lemma-case3} and Fig.~\ref{Rod-structure-case3} give us the sufficient information for the behavior of outgoing null geodesics in order to construct the Penrose diagram in each of the four types. 
These diagrams are presented in Fig.~\ref{pic-case3_diagrams}.

 \begin{figure}[t]
   \centering
   \includegraphics[width=0.6\linewidth]{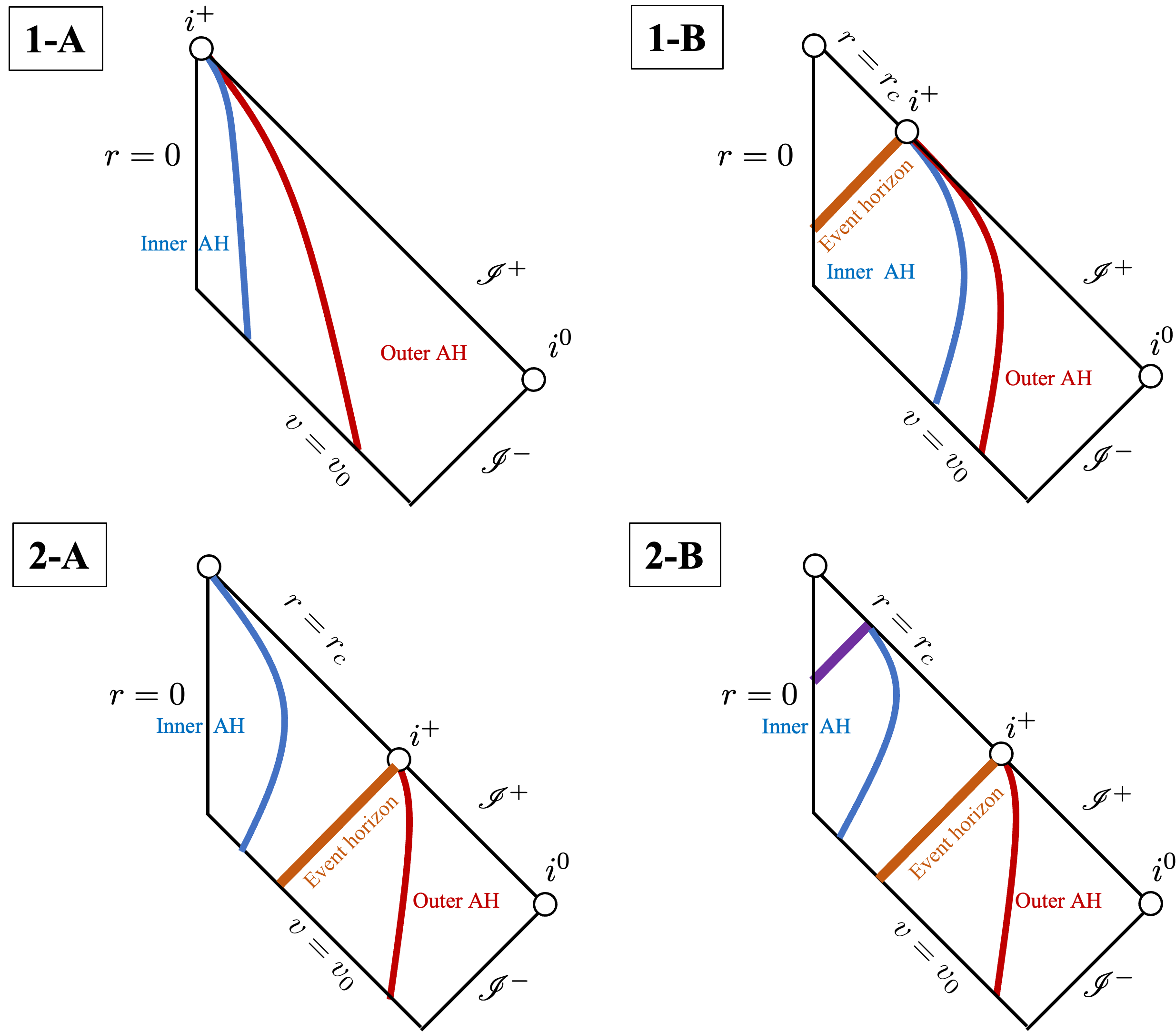}
   \caption{The Penrose diagrams of Case~3 spacetimes for Type~1-A (top-left), Type~1-B (top-right), Type~2-A (bottom-left),
   and Type~2-B (bottom-right). The red, blue and orange lines represent the outer AH, the inner AH and the event horizon respectively. }
     \label{pic-case3_diagrams}
 \end{figure}

 \subsection{Sufficient conditions to classify into one of the four classes.}
 
Here, we derive several sufficient conditions to determine which Penrose diagram corresponds to a
given metric. The contents of the propositions presented here would require some explanation.
Each proposition gives a sufficient condition for a given spacetime to belong to
one of the classes in Definitions~\ref{Def:Case3-Class1-Class2} and \ref{Def:Case3-ClassA-ClassB}. Namely,
Propositions~\ref{prop-Case3-Type1} and \ref{prop-Case3-Type2} give sufficient conditions for a given spacetime
to belong to Class~1 and Class~2, respectively. 
Similarly, Propositions~\ref{prop-Case3-TypeA} and \ref{prop-Case3-TypeB} give sufficient conditions for a given spacetime
to belong to Class~A and Class~B, respectively. 
For this reason, if one would like to show that the given spacetime belongs to 
Type~1-A, one must show that the sufficient conditions of Propositions~\ref{prop-Case3-Type1}
and \ref{prop-Case3-TypeA} are satisfied.

The following is the proposition to give the sufficient conditions for a given spacetime to be of Class~1:
 \begin{prop}
  \label{prop-Case3-Type1}
   If there exists a positive constant $\varepsilon$ for which 
  \begin{align}
    \label{eq-Case3-Type1-1}
    \lim_{v\rightarrow \infty}(x_+(v)-x_-(v))\exp{\left[\frac{1+\varepsilon}{2}h_c (X_+(v)-X_-(v))\right]}=0
  \end{align}
  is satisfied, then the spacetime is of Class~1.
 \end{prop}
 \begin{proof}
 As remarked after Definition~\ref{Def:Case3-Class1-Class2}, the condition for Class~1 spacetimes is that
 all outgoing null geodesics emitted from the inner AH reach the outer AH. We prove this using Eq.~\eqref{eq-Case3-Type1-1}. 
 For the sake of contradiction, we suppose the existence of a geodesic that stays in the region $x_-(v)<x(v)<x_+(v)$
 among such outgoing null geodesics.  
 For that geodesic, we define $\eta(v)$ as $\eta(v)\coloneqq x(v)-x_-(v)$, whose positivity in the range $v>v_0$ is guaranteed by Lemma \ref{lemma-case3}. Then, it satisfies
  \begin{align}
    \frac{d\eta}{dv}\,=\,\frac{h(v,x)}{2}(\eta-(x_+-x_-))\eta-\frac{dx_-}{dv}   \,>\, -\frac{h(v,x)}{2}(x_+-x_-)\eta.
  \end{align}
  where we used $h(v,x)\eta^2>0$ and $dx_-/dv<0$ in the inequality.
  Since $h(v,x)<h_c(1+\varepsilon)$ holds for $v>v_0$ with a sufficiently large $v_0$ between the inner and outer AHs, we obtain
  \begin{align}
    \frac{d\eta}{dv}&>-\frac{1+\varepsilon}{2}h_c(x_+-x_-)\eta\nonumber.
  \end{align}
  Dividing both sides with $\eta>0$ and integrating this inequality, we have
  \begin{align}
      \eta(v)&>\eta_0\exp{\left[-\frac{1+\varepsilon}{2}h_c[X_+(v)-X_-(v)-(X_+(v_0)-X_-(v_0))]\right]},
  \end{align}
  where $\eta_0$ is defined as $\eta_0\coloneqq \eta(v_0)$.
  Since $x_+(v)-x_-(v)>\eta(v)$ holds from the assumption,
  we have  
  \begin{align}
     \left(x_+(v)-x_-(v)\right)\exp{\left[\frac{1+\varepsilon}{2}h_c\left(X_+(v)-X_-(v)\right)\right]}
     \,>\, 
     \eta_0\exp{\left[\frac{1+\varepsilon}{2}h_c\left(X_+(v_0)-X_-(v_0)\right)\right]}.
  \end{align}
  However, from the assumption of Eq.~\eqref{eq-Case3-Type1-1}, the left-hand side converges to zero 
  as $v\to\infty$. This is a contradiction. 
  Since this contradiction is caused by the assumption of the existence of a geodesic that stays in the region $x_-(v)<x(v)<x_+(v)$,
  we have shown that all outgoing null geodesics between the inner AH and the outer AH intersect the outer AH, i.e., the spacetime is of Class~1.
 \end{proof}
This proposition leads to the following corollary:
\begin{corollary}
  \label{corollary-Case3-Type1}
  If there exists a positive constant $\varepsilon$ for which 
  \begin{align}
    \label{eq-Case3-Type1-3}
    \lim_{v\rightarrow\infty}x_+(v)\exp{\left[\frac{1+\varepsilon}{2}h_c X_+(v)\right]}=0
  \end{align}
  is satisfied, then the spacetime is of Class~1.
\end{corollary}
\begin{proof}
  Because $x_+(v)-x_-(v)<x_+(v)$ holds, 
   $X_+(v)-X_-(v)<X_+(v)-X_-(v_0)$ is satisfied.
   This leads to the inequality
  \begin{align}
    0&\,<\,(x_+(v)-x_-(v))\exp{\left[\frac{1+\varepsilon}{2}h_c (X_+(v)-X_-(v))\right]}
    \nonumber\\
    &\,<\,x_+(v)\exp{\left[\frac{1+\varepsilon}{2}h_c (X_+(v)-X_-(v_0))\right]}.
  \end{align}
  Therefore, the condition of Eq.~\eqref{eq-Case3-Type1-3} leads to that of Eq.~\eqref{eq-Case3-Type1-1}, indicating that 
  the spacetime is of Class~1.
\end{proof}

The following is the proposition to give the sufficient conditions for a given spacetime to be of Class~2:
\begin{prop}
  \label{prop-Case3-Type2}
   If there exist a constant $\alpha$ satisfying $0<\alpha<1$ and a positive constant $\varepsilon$ for which
  \begin{align}
    \label{eq-Case3-Type2-1}
    \alpha \frac{dx_+}{dv}+(1-\alpha)\frac{dx_-}{dv}>-\frac{h_c}{2}\alpha(1-\alpha)(1-\varepsilon)(x_+-x_-)^2
  \end{align}
  is satisfied for $v>v_0$ with a sufficiently large $v_0$, then the spacetime is of Class~2.
\end{prop}
\begin{proof}
As remarked after Definition~\ref{Def:Case3-Class1-Class2}, the condition for Class~2 spacetimes is that
 some of outgoing null geodesics emitted from the inner AH do not reach the outer AH. 
 It is sufficient to prove the existence of outgoing null geodesics 
 between the inner and outer AHs that cannot intersect the outer AH. 
 Using a parameter $\alpha $ satisfying $0<\alpha<1$, we introduce a curve $C$ as
  \begin{align}
    x_C(v)\coloneqq \alpha x_+(v)+(1-\alpha) x_-(v).
  \end{align}
  At each point on $C$, the value of $dx/dv$ of an outgoing null geodesic crossing $C$ at that point is
  \begin{align}
    \left.\frac{dx}{dv}\right|_{C}&=\frac{h(v,x_C)}{2}(x_C-x_+)(x_C-x_-)\nonumber\\
    &=-\frac{h(v,x_C)}{2}\alpha(1-\alpha)(x_+-x_-)^2.
      \end{align}
  Since $x_C$ satisfies $\lim_{v\rightarrow\infty}x_C=0$, $h(v,x_C)>(1-\varepsilon)h_c$ holds for $v>v_0$ with sufficiently large $v_0$, and thus, we derive
  \begin{align}
    \left.\frac{dx}{dv}\right|_{C}<-\frac{h_c}{2}\alpha(1-\alpha)(1-\varepsilon)(x_+-x_-)^2.
    \label{Derivative_dxdv_C}
  \end{align}
  On the other hand, from the assumption of Eq.~\eqref{eq-Case3-Type2-1}, we have
  \begin{align}
    \frac{dx_C}{dv}&\,=\,\alpha \frac{dx_+}{dv}+(1-\alpha)\frac{dx_-}{dv}
    \,>\,-\frac{h_c}{2}\alpha(1-\alpha)(1-\varepsilon)(x_+-x_-)^2.
%    &>\left.\frac{dx}{dv}\right|_{C}.
\label{Derivative_dxCdv}
  \end{align}
  Combining the inequalities of Eqs.~\eqref{Derivative_dxdv_C} and \eqref{Derivative_dxCdv}, we have
  \begin{equation}
  \frac{dx_C}{dv}\,>\,\left.\frac{dx}{dv}\right|_{C}.
  \end{equation}
  This means that on the curve $C$, outgoing null geodesics cross from the outside region to the inside region. 
  Consequently, outgoing null geodesics inside the curve $C$ do not intersect this curve, and thus, never
  arrive at the outer AH, i.e., the spacetime is of Class~2.
\end{proof}
This proposition leads to the following corollary.
\begin{corollary}
  \label{corollary-Case3-Type2}
  If there exists a constant $\alpha$ in the range $0<\alpha<1$ for which
  \begin{align}
    \label{eq-Case3-Type2-2}
    \lim_{v\rightarrow\infty}\frac{\alpha dx_+/dv+(1-\alpha)dx_-/dv}{(x_+-x_-)^2}\eqqcolon A\,>\, -\frac{h_c}{2}\alpha(1-\alpha)
  \end{align}
  is satisfied, then the spacetime is of Class~2.
\end{corollary}
\begin{proof}
    Since $\alpha\frac{dx_+}{dv}+(1-\alpha)\frac{dx_-}{dv}$ is negative, $A$ must be negative or zero. Since $A>-\frac{h_c}{2}\alpha (1-\alpha)$, there exists a sufficiently small positive number $\varepsilon$ satisfying $A>-\frac{h_c}{2}\alpha(1-\alpha)(1-\varepsilon)$.
    Then, for $v>v_0$ with sufficiently large $v_0$, we have
  \begin{align}
    \left|\frac{\alpha dx_+/dv+(1-\alpha)dx_-/dv}{(x_+-x_-)^2}-A\right|\,<\,A+\frac{h_c}{2}\alpha(1-\alpha)(1-\varepsilon),
  \end{align}
   and this leads to
  \begin{align}
    \alpha \frac{dx_+}{dv}+(1-\alpha)\frac{dx_-}{dv}\,>\,-\frac{h_c}{2}\alpha(1-\alpha)(1-\varepsilon)(x_+-x_-)^2.
  \end{align}
  This inequality is equivalent to that in Eq.~\eqref{eq-Case3-Type2-1}, i.e., the spacetime is of Class~2.
\end{proof}

The following is the proposition to give the sufficient conditions for a given spacetime to be of Class~A:
\begin{prop}
  \label{prop-Case3-TypeA}
   If there exists a positive value $\varepsilon$ satisfying 
  \begin{align}
    \label{eq-Case3-TypeA-1}
    \lim_{v\rightarrow\infty}x_-(v)\exp{\left[\frac{h_c}{2}(1-\varepsilon)X_+(v)\right]}\,=\,\infty,
  \end{align}
  the spacetime is of Class~A.
\end{prop}

\begin{proof}
As remarked after Definition~\ref{Def:Case3-ClassA-ClassB}, the condition for Class~A spacetimes is that
each outgoing null geodesics inside the inner AH necessarily intersect the inner AH at a later time. 
  Since the outgoing null geodesics in $0\leq x<x_-(v)$ intersect the inner AH in finite advanced time from Lemma~\ref{lemma-case3-2}, it is sufficient to show that any outgoing null geodesic starting in the region $x<0$ reaches $x=0$ in a finite advanced time. 
  For the sake of contradiction, suppose the existence of an outgoing null geodesic coming from the region $x < 0$ that does not reach $x = 0$. From this assumption, $x(v)$ of the outgoing null geodesic must
  converge to a nonpositive value at $v\to \infty$ because it is a monotonically increasing function. Then, 
  it must satisfy $\lim_{v\to\infty}x(v)=0$ because convergence to a negative value contradicts 
  the positivity of $dx/dv$ from the equation. Therefore, we have $\lim_{v\to\infty}h(v,x(v))=h_c$, and hence,
  the condition $h(v,x(v))>(1-\varepsilon)h_c$ can be made satisfied
  in the region $v>v_0$  by adopting a sufficiently large value of $v_0$.
  We now introduce $\tilde{\eta}\coloneqq x_-(v)-x(v)$ which satisfies $\tilde{\eta}>x_-(v)>0$. 
  Since $dx_-/dv<0$ and $-\tilde{\eta}^2<-x_-\tilde{\eta}$ are satisfied, we have
  \begin{equation}
    \frac{d\tilde{\eta}}{dv}\,=\,-\frac{h}{2}\left[\tilde{\eta}+(x_+-x_-)\right]\tilde{\eta}+\frac{dx_-}{dv}
    \,<\, -\frac{h(v,x)}{2}x_+\tilde{\eta}
    \,<\, -\frac{h_c}{2}(1-\varepsilon)x_+\tilde{\eta}.
  \end{equation}
  Dividing both sides by $\tilde{\eta}$ and integrating this inequality, we have
  \begin{align}
    \tilde{\eta}&\,<\,\tilde{\eta}_0\exp{\left[-\frac{h_c}{2}(1-\varepsilon)\left(X_+(v)-X_+(v_0)\right)\right]},
  \end{align}
   where $\tilde{\eta}_0$ is defined as $\tilde{\eta}_0\coloneqq \tilde{\eta}(v_0)$. From $x_-(v)<\tilde{\eta}$,
  \begin{align}
    x_-(v)\exp{\left[\frac{h_c}{2}(1-\varepsilon)X_+(v)\right]}
    \,<\,\tilde{\eta}_0\exp{\left[\frac{h_c}{2}(1-\varepsilon)X_+(v_0)\right]}
  \end{align}
  holds. However, from Eq.~\eqref{eq-Case3-TypeA-1}, the left-hand side diverges for $v\rightarrow\infty$, and therefore, this inequality is violated for a sufficiently large $v$. 
  This is a contradiction. Since this contradiction is caused by the assumption of the existence of an outgoing null geodesic that does not reach $x=0$, all of the outgoing null geodesics from $x<0$ must arrive at $x=0$ in a finite advanced time. 
  Therefore, the spacetime is of Class~A.
\end{proof}

The following is the proposition to give the sufficient conditions for a given spacetime to be of Class~B:
\begin{prop}
  \label{prop-Case3-TypeB}
   If there exist positive constants $\alpha$ and $\varepsilon$ for which 
  \begin{align}
    \label{eq-Case3-TypeB-1}
    \frac{dx_-/dv}{(x_++\alpha x_-)x_-}\,<\,-\frac{1+\alpha}{2\alpha}(1+\varepsilon)h_c
  \end{align}
  is satisfied for $v>v_0$ for sufficiently large $v_0$, then the spacetime is of Class~B.
\end{prop}
\begin{proof}
As remarked after Definition~\ref{Def:Case3-ClassA-ClassB}, the condition for Class~B spacetimes is the
  existence of an outgoing null geodesic which remains inside the inner apparent horizon, $x<x_-(v)$. 
    Using the constant $\alpha$ given in the proposition, we introduce a curve $C$ by $x_C(v)=-\alpha x_-(v)$. 
  The derivative of a curve $C$ is 
  \begin{equation}
  \label{eq-case3-typeb-1}
   \frac{dx_C}{dv} \, = \, -\alpha \frac{dx_-}{dv}. 
  \end{equation}
  On the other hand, at each point on $C$, the value of $dx/dv$ of an outgoing null geodesic crossing $C$ at that point is  
  \begin{align}
      \left.\frac{dx}{dv}\right|_{C}=\frac{h(v,-\alpha x_-)}{2}(1+\alpha)(x_++\alpha x_-)x_-.
  \end{align}
  Since $\lim_{v\rightarrow\infty} h(v,-\alpha x_-)=h_c$ holds, we have $(1+\varepsilon)h_c>h(v,-\alpha x_-)$ for the positive number $\varepsilon$ given in the proposition in the range $v>v_0$ if $v_0$ is sufficiently large. Then, we have
  \begin{align}
    \label{eq-case3-typeb-2}
    \left.\frac{dx}{dv}\right|_{C}\,<\,\frac{h_c}{2}(1+\alpha)(1+\varepsilon)(x_++\alpha x_-)x_-
    \,<\,-\alpha\frac{dx_-}{dv},
  \end{align}
  where we used Eq.~\eqref{eq-Case3-TypeB-1} in the second inequality.
Combining Eqs.~\eqref{eq-case3-typeb-1} and \eqref{eq-case3-typeb-2} yields
  \begin{align}
    \label{eq-Case3-TypeB-2}
    \left.\frac{dx}{dv}\right|_{C}\,<\,  \frac{dx_C}{dv}.
  \end{align}
  Therefore, at each point on the curve $C$ in the region $v>v_0$,
  the outgoing null geodesic crosses that point from the outside region to the inside region.
  This means that no outgoing null geodesic crosses the curve $C$ from the inside region
  to the outside region, implying the existence of null geodesics staying in the region $x<x_-(v)$.
  Thus, any Case~3 spacetime satisfying the inequality of Eq.~\eqref{eq-Case3-TypeB-1}
  is of Class~B.  
\end{proof}

This proposition trivially leads to the following corollary.
\begin{corollary}
  \label{corollary-Case3-TypeB}
  If
  \begin{align}
    \label{eq-Case3-TypeB-3}
    \lim_{v\rightarrow \infty}\frac{dx_-/dv}{(x_++\alpha x_-)x_-}\,<\,-\frac{1+\alpha}{2\alpha}h_c
  \end{align}
  holds for some positive constant $\alpha$, then the spacetime is of Class~B.
\end{corollary}
Note that the condition of Proposition \ref{prop-Case3-TypeB} is broader than that of Corollary \ref{corollary-Case3-TypeB} in the sense that Proposition~\ref{prop-Case3-TypeB} holds in the situation where the left-hand side of Eq.~\eqref{eq-Case3-TypeB-1} does not have a limit for $v\to \infty$.

\subsection{Some example of Case~3 spacetimes}

\begin{table}[tbh]
  \centering
  \caption{Summary of examples of Case~3 spacetimes to realize Types 1-A, 1-B, 2-A, and 2-B.} 
  \scalebox{0.8}{
  \begin{tabular}{c|c|c|c} \hline\hline
    Type~1-A & Type~1-B & Type~2-A & Type~2-B\\ \hline\hline
    \ & 
    \begin{tabular}{c}
      $x_+=\frac{a_+}{v^n},\, x_-=\frac{a_-}{v^k}$\\
      $(1<n\leq k)$
    \end{tabular}&
    \
    \begin{tabular}{c}
      $x_+=\frac{a_+}{v},\,x_-=\frac{a_-}{v^k}$\\
      $\left(2<2k<a_+h_{c}\right)$
    \end{tabular}
    &
    \begin{tabular}{c}
      $x_+=\frac{a_+}{v},\, x_-=\frac{a_-}{v^k}$\\
      $\left(2<{a_+h_c}<2k\right)$
    \end{tabular}
     \\ \hline
    \begin{tabular}{c}
      $x_+=\frac{a_+}{v^n},\, x_-=x_+-\frac{b}{v^k}$\\
      $(0<n<1,\, b>0,\, k>1)$
    \end{tabular}&
    \ &
    \begin{tabular}{c}
      $x_+=\frac{a_+}{v^n},\, x_-=\frac{a_-}{v^k}$\\
      $(0<n<1,\, n\leq k)$
    \end{tabular}&
    \begin{tabular}{c}
      $x_+=\frac{a_+}{v^n},\, x_-=a_-e^{-{v}/{L}}$\\
      $(0<n<1)$
    \end{tabular}\\ \hline
    \begin{tabular}{c}
      $x_\pm=\frac{a_\pm}{v}$\\
      $\left(\frac{2}{h_{c}}<a_+<a_-+\frac{2}{h_c}\right)$
    \end{tabular}&
    \begin{tabular}{c}
      $x_\pm=\frac{a_\pm}{v}$\\
      $\left(\sqrt{a_+}+\sqrt{a_-}<\sqrt{\frac{2}{h_c}}\right)$
    \end{tabular}&
    \begin{tabular}{c}
      $x_\pm=\frac{a_\pm}{v}$\\
      $\left(\sqrt{\frac{2}{h_c}}<\sqrt{a_+}-\sqrt{a_-}\right)$
    \end{tabular}&
    \ \\ \hline
  \end{tabular} 
  }

  \label{table-Case3}
\end{table}

 By applying Propositions~\ref{prop-Case3-Type1}, \ref{prop-Case3-Type2},   \ref{prop-Case3-TypeA},
 and \ref{prop-Case3-TypeB}
 and their corollaries proved above, we can identify the 
 functional forms of $x_\pm(v)$ that lead to Types 1-A, 1-B, 2-A and 2-B, respectively.
 Focusing mainly on the power-law functions $x_+(v) = a_+/v^n$ and $x_-(v)=a_-/v^k$
 (with the exception of the exponential function of the second line of Type~2-B), we have obtained the results as 
 summarized in Table~\ref{table-Case3}.
 Note that for the form $x_+=a_+/v^n$ and $x_-=a_-/v^k$, the condition $n\le k$ must be satisfied since $x_+(v)>x_-(v)>0$
 (when $n=k$, an additional condition $a_+>a_-$ must be required). 
 The proofs of this table are presented in Appendix~\ref{Examples:case3}.
 Similarly to the Type~2 of Case~2, the Type~2-B spacetimes are realized only when
 $x_+(v)$ decays slowly and $x_-(v)$ decays rapidly.
 This is a natural result because the Type~2-B spacetimes in Case~3
 have similar structures as those of Type~2 spacetimes in Case~1,
 and thus, the Type~2-B spacetimes in Case~3 is expected to be realized
 only when $x_+(v)$ decays slowly and the inner AH
 is sufficiently close to a null surface.
 On the other hand, the Type~2-A spacetimes occur when 
 the decay rate of $x_+(v)$ is slow while the decay rate of $x_-(v)$ is not fast enough to realize the Type~2-B spacetimes.
 The Type~1-B spacetimes are realized when the decay rates of both $x_\pm(v)$ are fast, and 
 the Type~1-A spacetimes occur when the decay of $x_\pm(v)$ is slow and
 the inner and outer AHs are located at sufficiently close positions.
 This result of Type~1-A (i.e., the type without the event horizon) would be natural, 
 because the effects of the trapped region are expected to become weak in such situations.
 We do not expect that the condition $2/h_c<a_+<a_-+2/h_c$ in the case of $x_\pm=a_\pm/v$ in Type~1-A would be
 optimal, because as we see in the next section, an approximate study 
 with $h(v,x)=h_c$ gives a more widely applicable condition,
 \begin{equation}
 \sqrt{a_+}-\sqrt{a_-} < \sqrt{\frac{2}{h_c}} < \sqrt{a_+}+\sqrt{a_-}.
 \label{condition-Type-1A-xpm-eq-apm-db-v}
 \end{equation}
 For this reason, the sufficient conditions derived here would include a room for improvement.

\pagebreak

%
%======================================%
%<<<<<<<<<<<< SECTION VI  >>>>>>>>>>>>>>%
%======================================%
%
\section{Photon behavior in solvable models}
\label{Sec:Solvable-models}

In the previous sections, we have derived the general conditions to specify which Penrose diagram is realized for a given metric.
In those analyses, it is difficult to understand the concrete behavior of outgoing null geodesics
for a given situation. In this section, we consider models for which 
solutions of the outgoing null geodesics are given in terms of elementary functions.
For this purpose, we assume the constancy of $h(v,x)$ (i.e., $h(v,x)=h_c$) in this section,
and solve the equation
\begin{equation}
\frac{dx}{dv} \, = \, \frac{h_c}{2}\left(x-x_+(v)\right)\left(x-x_-(v)\right)
\label{Eq:outgoing-null-geodesics-constant-hc}
\end{equation}
for particular choices of $x_\pm (v)$. 
Because $h(v,x)$ must asymptote to zero for $x\to \infty$,
this assumption is physically unrealistic in the sense that it must break down
at large distances. Hence, the analysis here must be
regarded as an approximate one that can only be applicable 
to the neighborhood of $x=0$. However, the analyses here provide us with a lot of physical indications
to understand what is actually going on in the spacetimes under consideration.
In particular, we consider the structure of outgoing null geodesics in the phase space $(v,\,x)$
adopting the concepts in the area of dynamical systems.

\subsection{Preparation}
\label{Sec:Solvable-preparation}

As a preparation, we introduce three differential equations for the function $z=z(T)$.
All equations of outgoing null geodesics considered in this section 
are reduced to one of these three equations after suitable transformations.

The first equation is
\begin{equation}
\frac{dz}{dT} \ = \ z^2-1.
\label{Equation-z-T-1}
\end{equation}
This equation is closely related to the famous {\it logistic equation}.
There are two constant solutions, $z=\pm 1$, which are called the {\it fixed points}.
It is convenient to divide the phase space $(T,z)$ into the three regions,
$z>1$ (region I), $-1<z<1$ (region II), and $z<-1$ (region III).
The solutions in these regions are
\begin{subequations}
\begin{eqnarray}
z & = & -\coth(T-T_d)  \qquad  (T<T_d),\\
z & = & -\tanh(T-T_c)  \qquad  (-\infty<T<\infty),\\
z & = & -\coth(T-T_e)  \qquad  (T_e<T),
\end{eqnarray}
\end{subequations}
respectively. In the region I, the solution $z(T)$ monotonically increases from unity
and diverges at $T=T_d$. In the region II, the solution $z(T)$ asymptotes to
$\pm 1$ for $T\to \mp \infty$.
In the region III, the solution $z(T)$ emerges at $T=T_e$ from $z=-\infty$ and asymptotes to
$z\to -1$ for $T\to \infty$.  
From these solutions, it is found that the fixed point $z=1$ is a {\it repeller},
while the point $z=-1$ is an {\it attractor}. 

The second equation is
\begin{equation}
\frac{dz}{dT} \ = \ z^2.
\label{Equation-z-T-2}
\end{equation}
There is one constant solution, $z=0$. This fixed point can be regarded as
the special case where the attractor and repeller become degenerate.
It is convenient to divide the phase space $(T,z)$ into the two regions
$z>0$ (region I) and  $z<0$ (region III), where the region II does not appear
as a result of the degeneracy. The solutions in these two regions are
\begin{subequations}
\begin{eqnarray}
z & = & \frac{1}{T_d-T}  \qquad  (T<T_d),\\
z & = & \frac{1}{T_e-T}  \qquad  (T_e<T).
\end{eqnarray}
\end{subequations}
In the region I, the solution $z(T)$ monotonically increases from zero
and diverges at $T=T_d$, while in the region III, 
the solution $z(T)$ emerges at $T=T_e$ from $z=-\infty$ and asymptotes to
$z\to 0$ for $T\to \infty$.  
From this behavior, the fixed point $z=0$ is called the {\it saddle point}.

The third equation is
\begin{equation}
\frac{dz}{dT} \ = \ z^2+1.
\label{Equation-z-T-3}
\end{equation}
This equation does not have any constant solution, and hence, no fixed point is present.
The solution is
\begin{equation}
z \, = \, \tan (T-T_c) \qquad \left(T_c-\frac{\pi}{2}<T<T_c+\frac{\pi}{2}\right).
\end{equation}
This solution monotonically increases from $-\infty$ to $+\infty$, and $z(T)$ crosses zero at
$T=T_c$.

\subsection{Models of $x_\pm(v)=O(1/v^n)$ with $0<n<1$}
\label{Sec:Solvable-0<n<1}

We begin with the models where the decay of $x_\pm (v)$ is relatively slow,
i.e., $x_\pm(v)=O(1/v^n)$ with $0<n<1$. We choose 
\begin{equation}
x_\pm(v) \, = \,
\frac{1}{v^n}
\left[\frac{a_++a_-}{2}+\frac{n}{h_cv^{1-n}}
\pm \sqrt{\left(\frac{a_++a_-}{2}+\frac{n}{h_cv^{1-n}}\right)^2-a_+a_-}
\right].
\end{equation}
For large $v$, the behavior of $x_\pm(v)$ is approximately given as
\begin{equation}
x_\pm(v) \, = \,
\frac{a_\pm}{v^n}
\left[
1\pm \frac{2n}{h_c(a_+-a_-)v^{1-n}}+O\left(\frac{1}{v^{2-2n}}\right)
\right].
\end{equation} 
By the transformation
\begin{equation}
x(v) \, = \, \frac{1}{2v^n}\left[(a_+-a_-)z(v)+(a_++a_-)\right]
\end{equation}
and 
\begin{equation}
T \ = \ \frac{h_c(a_+-a_-)}{4(1-n)}v^{1-n},
\end{equation}
Eq.~\eqref{Eq:outgoing-null-geodesics-constant-hc} is reduced to the first equation
of Sec.~\ref{Sec:Solvable-preparation}, i.e.,  Eq.~\eqref{Equation-z-T-1}.

The solution that corresponds to the repeller $z=+1$ is given by $x_{\rm EH}(v)=a_+/v^{n}$,
and this is the event horizon. Since the repulsive property is kept
in the phase space $(v,x)$ as well, it would be appropriate to also call it the {\it repulsive separatrix}.
The solution that corresponds to the attractor $z=-1$ is given by $x_{\rm A}(v)=a_-/v^{n}$.
Because of the attractive property of this solution
in the phase space $(v,x)$, we call it the {\it attractive separatrix}.
The solutions that correspond to the regions I, II, and III are
\begin{subequations}
\begin{eqnarray}
x & = & \frac{1}{2v^n}\left\{-(a_+-a_-)\coth\left[\frac{h_c(a_+-a_-)}{4(1-n)}\left(v^{1-n}-v_d^{1-n}\right)\right] +(a_++a_-)\right\}  \quad  (v<v_d),\\
x & = & \frac{1}{2v^n}\left\{-(a_+-a_-)\tanh\left[\frac{h_c(a_+-a_-)}{4(1-n)}\left(v^{1-n}-v_c^{1-n}\right)\right] +(a_++a_-)\right\}  \quad  (0<v<\infty),\\
x & = & \frac{1}{2v^n}\left\{-(a_+-a_-)\coth\left[\frac{h_c(a_+-a_-)}{4(1-n)}\left(v^{1-n}-v_e^{1-n}\right)\right] +(a_++a_-)\right\}  \quad  (v_e<v).
\end{eqnarray}
\end{subequations}

 \begin{figure}[t]
   \centering
   \includegraphics[width=0.4\linewidth]{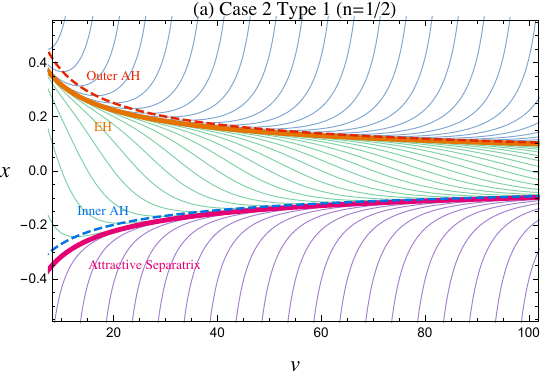}
   \includegraphics[width=0.4\linewidth]{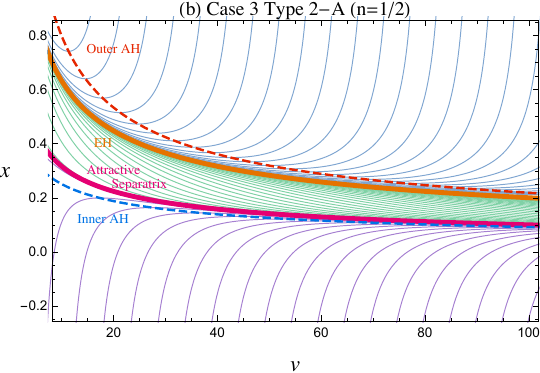}
   \caption{The behavior of null geodesic congruence 
   in the models with $n=1/2$ in the phase space $(v,x)$. The event horizon and the attractive separatrix
   are shown with thick curves, and the outer and inner AHs are shown by dashed curves.
   (a) The model with $a_\pm = \pm 1$ and $h_c=1$. This corresponds to Type~1 spacetime of Case~2.
   (b) The model with $a_+ = 2$, $a_-=1$, and $h_c=1$. This corresponds to Type~2-A spacetime of Case~3.}
     \label{n-onehalf-Case2-Case3}
 \end{figure}

 Figure~\ref{n-onehalf-Case2-Case3} shows the behavior of outgoing null geodesics
 in the phase space $(v,\,x)$  for the choice
 $n=1/2$, $a_\pm = \pm 1$, and $h_c=1$ (left panel), 
 and $n=1/2$, $a_+ = 2$, $a_-=1$, and $h_c=1$ (right panel).\footnote{Strictly speaking,
 in these examples we adopt $|a_-|^{2/3}$ as the unit of the length, and normalized
 $a_+$ and $h_c$ with $|a_-|^{2/3}$ appropriately.}
 The event horizon and the attractive separatrix are shown with thick curves,
 and the outer and inner AHs are indicated with dashed curves.
 Since $a_-$ is negative and positive in the left and right panels, the spacetimes in the cases 2 and 3 are considered, respectively.
 In both cases, the outer AH is located slightly outside the event horizon, and their locations
 approximately coincide 
 in the sense that $\lim_{v\to\infty}x_+/x_{\rm EH}=1$ holds.
 In the cases 2 and 3, the attractive separatrices are located slightly inside and outside the inner AH,
 respectively, and in both cases, their locations approximately coincide
 in the sense that $\lim_{v\to\infty}x_-/x_{\rm A}=1$ holds.
 Since the attractive separatrix strongly attracts other outgoing null geodesics,
 in the left panel, all outgoing null geodesics between the event horizon and the inner AH
 eventually cross the inner AH as they approach the attractive separatrix.
 This indicates that the Penrose diagram is Type~1 of Case~2.
 Similarly, in the right panel,
 all outgoing null geodesics inside the inner AH
 eventually cross the inner AH as they approach the attractive separatrix.
 This indicates that the Penrose diagram is Type~2-A of Case~3.

 To summarize, when the decay of $x_\pm(v)$ is relatively slow,
 there appears an important attractive separatrix solution, $x=x_{\rm A}(v)$,
 in the neighborhood of the inner AH, 
 in addition to the event horizon $x=x_{\rm EH}(v)$ in the neighborhood of the outer AH. 
 The strong attractive property of $x=x_{\rm A}(v)$ is one of the important factors 
 to determine the type of the Penrose diagrams.

\subsection{Models of $x_\pm(v)=O(1/v^n)$ with $n=1$}
\label{Sec:Solvable-n=1}

We now study the models of $x_\pm(v)=O(1/v^n)$ with $n=1$. These models
are particularly interesting in that they show a variety of possible types in Case~3.
As the position of the AHs, we choose
\begin{equation}
x_\pm = \frac{a_\pm}{v}.
\end{equation}
Introducing the new variable $y$ with $x(v)=y(v)/v$, we have
the equation
\begin{equation}
\frac{2}{h_c}v\frac{dy}{dv} \ = \ (y-b)^2 - \frac{D}{4},
\label{Eq:null-geodesic-y}
\end{equation}
where
\begin{equation}
b = \frac{a_++a_-}{2}-\frac{1}{h_c}\qquad \textrm{and}\qquad D \, = \, 4(b^2-a_+a_-).
\label{Def:b-and-D}
\end{equation}
For later convenience, 
we define 
\begin{equation}
d \, = \, \frac{\sqrt{|D|}}{2}.
\label{Def:d}
\end{equation}
We discuss the properties of the spacetime
for the cases $D>0$, $D=0$, and $D<0$, one by one.

\subsubsection{Models with $D>0$}
\label{Sec:Solvable-D>0}

If $D$ is positive, the right-hand side of Eq.~\eqref{Eq:null-geodesic-y}
can be written as $[y-(b+d)][y-(b-d)]$. Then, introducing $z(v)$ with
\begin{equation}
y(v)\, = \, b + d\cdot z(v)
\label{y-z-relation}
\end{equation}
and $T$ with
\begin{equation}
\frac{v}{v_0} \, = \, \exp\left[\frac{2}{h_c d}(T-T_0)\right]
\label{v-T-relation}
\end{equation}
for some constants $v_0$ and $T_0$, 
Eq.~\eqref{Eq:null-geodesic-y} is reduced to the first equation
of Sec.~\ref{Sec:Solvable-preparation}, i.e.,  Eq.~\eqref{Equation-z-T-1}.

The solution that corresponds to the repeller $z=+1$ is given by $x_{\rm EH}(v)=(b+d)/v$,
which is the event horizon. 
The solution that corresponds to the attractor $z=-1$ is given by $x=x_{\rm A}(v)=(b-d)/v$,
which gives the attractive separatrix.
The solutions that correspond to the regions I, II, and III are
\begin{subequations}
\begin{eqnarray}
x & = &  \frac{1}{v}\left[b-d\frac{(v/v_d)^{h_cd}+1}{(v/v_d)^{h_cd}-1}\right]  \quad  (v<v_d),\\
x & = &  \frac{1}{v}\left[b-d\frac{(v/v_c)^{h_cd}-1}{(v/v_c)^{h_cd}+1}\right]   \quad  (0<v<\infty),\\
x & = &  \frac{1}{v}\left[b-d\frac{(v/v_e)^{h_cd}+1}{(v/v_e)^{h_cd}-1}\right]   \quad  (v_e<v).
\end{eqnarray}
\end{subequations}
In contrast to the models with $0<n<1$, the null geodesics approach
the attractive separatrix with a power-law behavior, $x/x_{\rm A}-1\propto 1/v^{h_cd}$.
Depending on the value of $h_cd$, the strength of the attractive property of the separatrix 
$x=x_{\rm A}(v)$ can be very weak.

In the case of $a_-<0$ (i.e., in Case~2), the positivity of $D$ is guaranteed and
the resultant Penrose diagram is Type~1.
We have to be careful in the case of $a_->0$ (i.e., in Case~3), 
since there are two possibilities to realize $D>0$:
\begin{equation}
\sqrt{a_+}-\sqrt{a_-}>\sqrt{\frac{2}{h_c}}, \qquad \textrm{and} \qquad
\sqrt{a_+}+\sqrt{a_-}<\sqrt{\frac{2}{h_c}}.
\label{condition-positiveD}
\end{equation}
In these two cases, the property of the spacetime changes a lot.
In the former case, we have $b>\sqrt{a_+a_-}>a_-$, and hence,
the event horizon exists outside the inner AH 
since $x_{\rm EH} = (b+d)/v > a_-/v = x_-(v)$ holds.
On the other hand, in the latter case, we have $b<-\sqrt{a_+a_-}<0$,
and hence, $x_{\rm EH} = (b+d)/v <0$ since $|b|>d>0$.
Therefore, the event horizon exists inside the inner AH.

 \begin{figure}[t]
   \centering
   \includegraphics[width=0.4\linewidth]{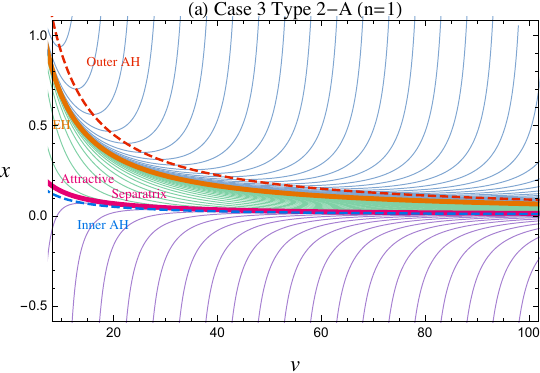}
   \includegraphics[width=0.4\linewidth]{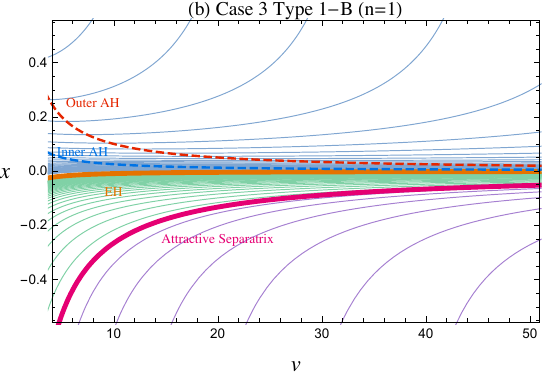}
   \includegraphics[width=0.4\linewidth]{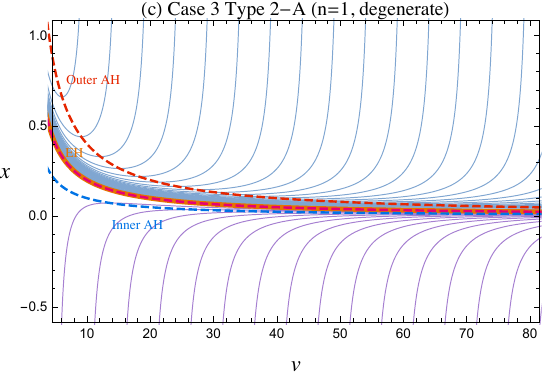}
   \includegraphics[width=0.4\linewidth]{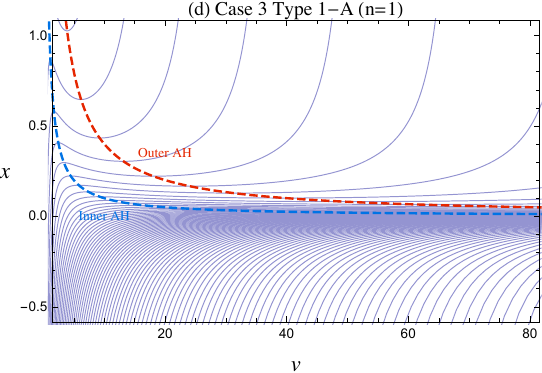}
   \caption{The behavior of outgoing radial null geodesics
   in the models with $n=1$ in the phase space $(v,x)$. 
   (a) The model with $a_+ = 9$, $a_-= 1$ and $h_c=2$. This corresponds to Type~2-A spacetime of Case~3.
   (b) The model with $a_+ = 1$, $a_-=1/4$, and $h_c=1/2$. This corresponds to Type~1-B spacetime of Case~3.
   (c) The model with $a_+ = 4$, $a_-= 1$ and $h_c=2$ where the event horizon becomes a degenerate separatrix. This corresponds to Type~2-A spacetime of Case~3.
   (d) The model with $a_+ = 4$, $a_-= 1$ and $h_c=1/2$. This corresponds to Type~1-A spacetime of Case~3.
   }
     \label{n-one-Case3}
 \end{figure}

 Figure~\ref{n-one-Case3} shows the behavior of outgoing null geodesics
 in the phase space $(v,\,x)$  for the choice
 $n=1$, $a_+ = 9$, $a_-=1$, and $h_c=2$ (top-left panel), 
 and $n=1$, $a_+ = 1$, $a_-=1/4$, and $h_c=1/2$ (top-right panel).
 These top-left and top-right panels correspond to the former and latter cases of Eq.~\eqref{condition-positiveD}.
 In contrast to the case $0<n<1$, 
 the locations of the outer AH and the event horizon do not approximately
 coincide in the sense that $\lim_{v\to\infty}x_+/x_{\rm EH}$ is not unity.
 Similarly, the locations of the inner AH and the attractive separatrix do not approximately coincide, either.
 In the top-left panel, all outgoing null geodesics inside the inner AH
 eventually cross the inner AH, and the event horizon exists
 outside the inner AH. This indicates that the spacetime is of Type~2-A of Case~3.
 In the top-right panel, the event horizon is confirmed to exist inside the inner AH,
 and hence, the Penrose diagram is of Type~1-B of Case~3.

\subsubsection{Models with $D=0$}
\label{Sec:Solvable-Deq0}

We now consider the models with $D=0$, which is possible only in Case~3.
Introducing $z(v)$ with
\begin{equation}
y(v)\, = \, z(v)+b
\end{equation}
and $T$ with
\begin{equation}
\frac{v}{v_0} \, = \, \exp\left[\frac{2}{h_c}(T-T_0)\right]
\end{equation}
for some constants $v_0$ and $T_0$, 
Eq.~\eqref{Eq:null-geodesic-y} is reduced to the second equation
of Sec.~\ref{Sec:Solvable-preparation}, i.e.,  Eq.~\eqref{Equation-z-T-2}.

Equation~\eqref{Equation-z-T-2} has the saddle point solution, $z=0$, and the
corresponding solution is given by $x_{\rm EH}(v)=b/v$.
Here, the event horizon becomes the degenerate separatrix:
The outgoing null geodesics outside the event horizon escape to infinity,
while those inside the event horizon approach it.
In fact, the solutions that correspond to the regions I and III are
\begin{subequations}
\begin{eqnarray}
x & = &  \frac{1}{v}\left[b-\frac{2/h_c}{\log(v/v_d)}\right]  \quad  (v<v_d),\\
x & = &  \frac{1}{v}\left[b-\frac{2/h_c}{\log(v/v_e)}\right]   \quad  (v_e<v).\label{Solution-D=0-insideBH}
\end{eqnarray}
\end{subequations}
Since the null geodesics inside the event horizon approach
the event horizon with a logarithmic behavior, $x/x_{\rm EH}-1\propto 1/\log v$,
the strength of the attractive property of the degenerate separatrix 
$x=x_{\rm A}(v)$ is extremely weak.

Similarly to the models with $D>0$,  
there are two possibilities to realize $D=0$:
\begin{equation}
\sqrt{a_+}\mp \sqrt{a_-}=\sqrt{\frac{2}{h_c}}.
\label{condition-zeroD}
\end{equation}
In the cases of upper and lower signs, the values of $b$ are positive and negative,
respectively. The type of the Penrose diagram is not changed from the models with $D>0$, and
the cases of upper and lower signs correspond to the Types 2-A and 1-B, respectively. 
The bottom-left panel of Fig.~\ref{n-one-Case3}
shows the case of $a_+=4$, $a_-=1$, and $h_c=2$, which corresponds to the upper sign 
of Eq.~\eqref{condition-zeroD}.

\subsubsection{Models with $D<0$}

We now focus on the models with $D<0$, which is possible only in Case~3.
The condition $D<0$ is equivalent to the one presented in Eq.~\eqref{condition-Type-1A-xpm-eq-apm-db-v}.
With the same transformation as Eqs.~\eqref{y-z-relation} and $\eqref{v-T-relation}$, 
Eq.~\eqref{Eq:null-geodesic-y} is reduced to the third equation
of Sec.~\ref{Sec:Solvable-preparation}, i.e.,  
Eq.~\eqref{Equation-z-T-3}.
Since Eq.~\eqref{Equation-z-T-3} has no fixed point solution, 
there is no event horizon in these models.
In fact, the solutions are
\begin{subequations}
\begin{eqnarray}
x & = &  \frac{1}{v}\left\{b+d\tan\left[\frac{h_cd}{2}\log\left(\frac{v}{v_c}\right)\right]\right\}  \quad 
\left(e^{-\pi/h_cd}<\frac{v}{v_c}<e^{\pi/h_cd}\right),
\end{eqnarray}
\end{subequations}
each of which evolves from $-\infty$ to $\infty$. 
The bottom-right panel of Fig.~\ref{n-one-Case3}
shows the case of $a_+=4$, $a_-=1$, and $h_c=1/2$.
The worldline $x=x(v)$ of each photon has a local maximum and a local minimum, and
they correspond to the positions of the inner and outer AHs, respectively.

\subsection{Models of $x_\pm(v)=O(1/v^n)$ with $n>1$}
\label{Sec:Solvable-n>1}

Next, we consider the models where the decay of $x_\pm (v)$ is relatively fast,
i.e., $x_\pm(v)=O(1/v^n)$ with $n>1$. We choose 
\begin{equation}
x_\pm(v) \, = \,
\pm\frac{a}{v^n}
+\frac{h_ca^2}{2(2n-1)v^{2n-1}}.
\end{equation}
Since $x_-(v)<0$, this spacetime is of Case~2.
By the transformation
\begin{equation}
x(v) \, = \, z(v)+\frac{h_ca^2}{2(2n-1)v^{2n-1}}
\end{equation}
and 
\begin{equation}
T \ = \ \frac{h_c}{2}v,
\end{equation}
Eq.~\eqref{Eq:outgoing-null-geodesics-constant-hc} is reduced to the second equation
of Sec.~\ref{Sec:Solvable-preparation}, i.e.,  Eq.~\eqref{Equation-z-T-2}.
When $x_\pm (v)$ decay rapidly, the effects of $x_\pm(v)$
in Eq.~\eqref{Eq:outgoing-null-geodesics-constant-hc} are weak, and
this would be the reason why the equation is reduced to the second equation
of Sec.~\ref{Sec:Solvable-preparation}.

 \begin{figure}[t]
   \centering
   \includegraphics[width=0.5\linewidth]{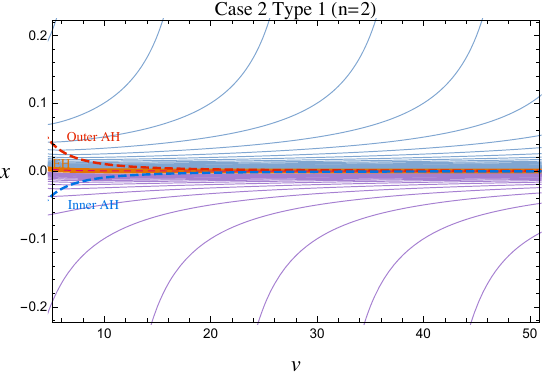}
   \caption{The behavior of null geodesic congruence 
   in the models with $n=2$, $a=1$, $h_c=2$ in the phase space $(v,x)$. 
   This corresponds to Type~1 spacetime of Case~2.}
     \label{n-two-Case2-Type1-degenerate}
 \end{figure}

Equation~\eqref{Equation-z-T-2} has the saddle point solution, $z=0$, and the corresponding solution is given by
\begin{equation}
x_{\rm EH}\,=\,\frac{h_ca^2}{2(2n-1)v^{2n-1}}.
\end{equation}
The event horizon is located between the outer and inner AHs,
and $x_{\rm EH}$ decays much faster than $x_\pm (v)$. 
Similarly to the case of $D=0$ and $n=1$, the event horizon becomes a degenerate separatrix.
The solutions that correspond to the regions I and III are
\begin{subequations}
\begin{eqnarray}
x & = &  \frac{2/h_c}{v_d-v}+\frac{h_ca^2}{2(2n-1)v^{2n-1}}  \quad  (v<v_d),\\
x & = &  \frac{2/h_c}{v_e-v}+\frac{h_ca^2}{2(2n-1)v^{2n-1}}   \quad  (v_e<v).
\end{eqnarray}
\end{subequations}
The outgoing null geodesics inside the event horizon
approach the event horizon as $x_{\rm EH}-x \approx 2/h_cv$.
The decay rate of $x(v)$ is much slower than those of $x_{\rm EH}(v)$ or $x_\pm(v)$. 
This means that the value of $x(v)-x_-(v)$ of any outgoing null geodesic
between the event and the inner AH 
 eventually becomes negative as $v$ is increased, and thus, it crosses the inner AH.
Therefore, the spacetime belongs to Type~1 of Case~2.
Figure~\ref{n-two-Case2-Type1-degenerate} shows the behavior of outgoing radial null geodesics
in the models with $n=2$, $a=1$, $h_c=2$ in the phase space $(v,x)$.

\subsection{Models realizing Type~2 of Case~2 and Type~2-B of Case~3}
\label{Sec:Solvable-Type2Case2-Type2BCase3}

Our last model is the one that realizes Type~2 of Case~2 and Type~2-B of Case~3.
We choose
\begin{equation}
x_\pm(v) \ = \ \frac{a}{2v}\left[1-\frac{4b}{a(k-2)v^{k-1}}\pm \sqrt{1-\frac{4kb}{a(k-2)v^{k-1}}}\right],
\end{equation}
where $k>2$ is assumed. 
For large $v$, the behavior of $x_\pm(v)$ are approximately given as
\begin{equation}
x_+(v) \, = \,
\frac{1}{v}
\left[
a - \frac{(k+2)b}{(k-2)v^{k-1}}+O\left(\frac{1}{v^{2k-2}}\right)
\right],
\end{equation} 
\begin{equation}
x_-(v) \, = \,
\frac{1}{v}
\left[
\frac{b}{v^{k-1}} +O\left(\frac{1}{v^{2k-2}}\right)
\right].
\end{equation} 
The formula for $x_-(v)$ indicates that if $b$ is negative and positive, the spacetime is of Cases 2 and 3, respectively.
We also require $h_c$ to be given by
\begin{equation}
h_c=\frac{4}{a}.
\end{equation}
By the transformation
\begin{equation}
x(v) \, = \, \frac{a}{4v}\left[z(v)+1-\frac{8b}{a(k-2)v^{k-1}}\right]
\end{equation}
and 
\begin{equation}
\frac{v}{v_0} \, = \, \exp[2(T-T_0)]
\end{equation}
for some constants $v_0$ and $T_0$, 
Eq.~\eqref{Eq:outgoing-null-geodesics-constant-hc} is reduced to the first equation
of Sec.~\ref{Sec:Solvable-preparation}, i.e.,  Eq.~\eqref{Equation-z-T-1}.

 \begin{figure}[t]
   \centering
   \includegraphics[width=0.4\linewidth]{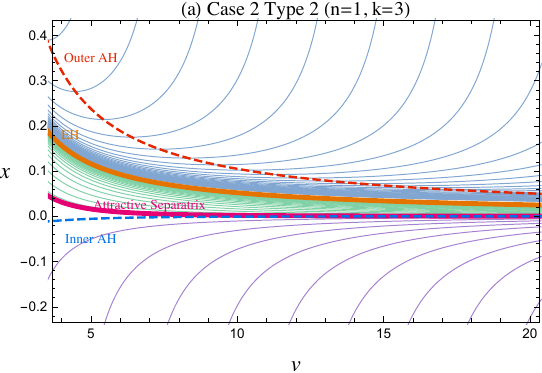}
   \includegraphics[width=0.4\linewidth]{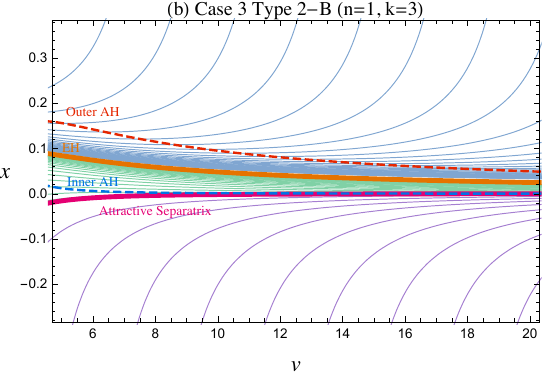}
   \caption{The behavior of null geodesic congruence 
   in the models that realize the spacetimes of Type~2 of Case~2 [left panel (a)]
   and of Type~2-B of Case~3 [right panel (b)].
   The parameter choices are: (a) $a=1$, $b=-1$, and $k=3$; and (b) $a=1$, $b=1$, and $k=3$.}
     \label{n-two-Case2}
 \end{figure}

The solution that corresponds to the repeller $z=+1$ is given by 
\begin{equation}
x_{\rm EH}(v)\,=\,\frac{1}{2v}\left[a-\frac{4b}{(k-2)v^{k-1}}\right],
\end{equation}
which is the event horizon. 
The solution that corresponds to the attractor $z=-1$ is given by 
\begin{equation}
x_{\rm A}(v)\,=\, -\frac{2b}{(k-2)v^k},
\end{equation}
which gives the attractive separatrix.
The solutions that correspond to the regions I, II, and III are
\begin{subequations}
\begin{eqnarray}
x & = &  \frac{1}{2v}\left[\frac{av_d}{v_d-v}-\frac{4b}{(k-2)v^{k-1}}\right]  \quad  (v<v_d),\\
x & = &  \frac{1}{2v}\left[\frac{av_c}{v_c+v}-\frac{4b}{(k-2)v^{k-1}}\right]    \quad  (0<v<\infty),\\
x & = &  \frac{1}{2v}\left[\frac{av_e}{v_e-v}-\frac{4b}{(k-2)v^{k-1}}\right]   \quad  (v_e<v).
\end{eqnarray}
\end{subequations}
In the case that $b$ is negative, $x_{\rm A}(v)$ is positive while $x_-(v)$ is negative.
Therefore, all outgoing null geodesics between the event horizon
and the attractive separatrix never cross the inner AH.
This indicates that the spacetime belongs to Type~2 of Case~2.
Interestingly, the attractive separatrix coincides with the null geodesic $U=U^{(-)}_\infty$
defined in Sec.~\ref{Subsec:compactified-retarded}. 
Similarly, if $b$ is positive, $x_{\rm A}(v)$ is negative while $x_-(v)$ is positive, 
and hence, all outgoing null geodesics inside the attractive separatrix
never arrive at the inner AH. 
Since the inner AH is located inside the event horizon, the spacetime belongs to Type~2-B of Case~3.
Again, the attractive separatrix coincides with the null geodesic $U=U^{(-)}_\infty$. 
The left and right panels of Fig.~\ref{n-two-Case2} show the examples of the phase space $(v,x)$
for $a=1$, $b=-1$, and $k=3$ (i.e., Type~2 of Case~2)
and for $a=1$, $b=1$, and $k=3$ (i.e., Type~2-B of Case~3).
As shown in these examples, the relative locations of the attractive separatrix and the inner AH are also an important factor to determine the type of the Penrose diagram.

\subsection{On the structural stability of these models}

From the models presented here, 
we have seen the various important factors to determine the types of Penrose diagrams:
(i) the existence, nonexistence, and degeneracy of the event horizon (the repulsive separatrix)
and the attractive separatrix; 
(ii) the relative positions of the separatrices and the inner AH;
and (iii) the strength of the attractive property of the attractive or degenerate separatrix.
Since the discussed models are based on the particular choices of $x_\pm(v)$ and 
the constancy of $h(v,x)$, 
whether the obtained view here is not changed when these assumptions are violated
must be questioned.
Although we do not expect that modifications in the subleading orders in $x_\pm(v)$ and $h(v,x)$
would change these structures completely,
it is not clear whether all equations for $x(v)$ can be reduced to the three equations
introduced in Sec.~\ref{Sec:Solvable-preparation} particularly
in the cases where $h(v,x)$ depends both on $v$ and $x$.
Without this transformation, the definition of the attractive separatrix
may become vague, although there would be a clear definition
of the repulsive separatrix, i.e., the event horizon.
Therefore, there are many things to be discussed concerning the generality
of the obtained view in this section, but here, we postpone these issues as a future work.

 \newpage
 \section{\label{sec-extendibility}Extendibility of spacetimes}

In this section, we briefly discuss the extendibility of spacetimes using the
solvable models in the previous section. 
 As we have shown in Secs.~\ref{sec-case1}, \ref{sec-case2} and \ref{sec-case3}, spacetimes of all types except Type~1-A 
 in Case 3  have the event horizons, and 
 there exist the future boundaries $\mathcal{H}_{\rm I}$ given by $v=\infty$ and $r=r_c$ which 
 all outgoing null geodesics in the black hole region asymptotically approach as $v\rightarrow\infty$.  
 Here, we discuss whether $\mathcal{H}_{\rm I}$ is extendible or not, which 
 is equivalent to whether $\mathcal{H}_{\rm I}$ is a Cauchy horizon or future null infinity. 
 For this purpose, we calculate the affine parameters $\lambda$ of outgoing null geodesics.
 The spacetime is extendible if $\lambda$ remains finite in the limit $v\to\infty$, 
 while the spacetime is not extendible if $\lambda$ diverges.
 
 Supposing that the outgoing null geodesics are parametrized as $(v,\,x)=(v(\lambda),\,x(\lambda))$,
 we derive the formula for calculating the affine parameter $\lambda$ of an outgoing null geodesic given by $x=x(v)$. 
 In the metric of Eq.~\eqref{eq-metric}, the $v$-component of the geodesic equations is
 \begin{equation}
 \ddot{v}\,=\,-\left(\frac{A_{,v}}{A}+fA_{,x}+\frac12f_{,x}A\right)\dot{v}^2,
 \label{Eq:geodesic-v-component}
 \end{equation}
 where dot denotes the derivative with respect to $\lambda$. 
 Because the equation for the outgoing null geodesics is given by Eq.~\eqref{eq-outgoing}, we have
 \begin{equation}
 \frac{dA}{dv}\, = \, A_{,v}+A_{,x}\frac{dx}{dv}\, = \, A_{,v}+\frac12 AA_{,x}f,
 \end{equation}
 and thus, Eq.~\eqref{Eq:geodesic-v-component} is rewritten as
 \begin{equation}
 \frac{\ddot{v}}{\dot{v}}\, = \, -\frac{\dot{A}}{A} -\frac12H_{,x}\dot{v},
 \end{equation}
 with
 \begin{equation}
 H(v,x)\, :=\, A(v,x)f(v,x) \, = \, h(v,x)\left(x-x_+(v)\right)\left(x-x_-(v)\right).
 \label{Def:H}
 \end{equation}
 Integrating this equation, we have
 \begin{equation}
 \dot{v} \, = \, \dot{v}(\lambda_0) \frac{A(\lambda_0)}{A(\lambda)} \exp\left[-\frac12\int_{v_0}^{v} H_{,x}(v^{\prime\prime},x(v^{\prime\prime}))dv^{\prime\prime}\right],
 \end{equation}
 where $A(\lambda)$ abbreviates $A(v(\lambda),x(\lambda))$, the initial position is $(v,\,x)=(v_0,\,x_0)$,
 and the affine parameter at the initial position is $\lambda_0$. 
 Taking the inverse of both sides and integrating again, we obtain
 \begin{equation}
 \lambda-\lambda_0 \, = \, \frac{1}{\dot{v}(\lambda_0)A(\lambda_0)}
 \int_{v_0}^{v}A(\lambda) \exp\left[\frac12\int_{v_0}^{v^\prime} H_{,x}(v^{\prime\prime},x(v^{\prime\prime}))dv^{\prime\prime}\right]dv^\prime.
 \label{lambda-formula}
 \end{equation}
 Since all outgoing null geodesics converge as $x\to 0$ in the black hole, the factor $A(\lambda)$
 also converges to a finite value $A_\infty = \lim_{v\to\infty}A(v,x(v))$, 
 and hence, is bounded as $(1-\varepsilon)A_\infty<A(\lambda)<(1+\varepsilon)A_\infty$ for an arbitrary positive $\varepsilon$ 
 in the range $v\ge v_0$ if we adopt a sufficiently large $v_0$. This means that $A(\lambda)$ 
 does not affect whether $\lambda$ diverges or not.
 In order for the affine parameter $\lambda$ to remain finite in the limit $v\to\infty$, 
 the integral in the argument of the exponential function must diverge to $-\infty$ sufficiently rapidly.
 
Below, we derive a sufficient condition for the finiteness of $\lambda$.
Since the increase in $\lambda$ is bounded as
 \begin{equation}
 \lambda-\lambda_0 \, < \, \frac{(1+\varepsilon)A_\infty}{\dot{v}(\lambda_0)A(\lambda_0)}
 \int_{v_0}^{v} \exp\left[\frac12\int_{v_0}^{v^\prime} H_{,x}(v^{\prime\prime},x(v^{\prime\prime}))dv^{\prime\prime}\right]dv^\prime,
 \label{lambda-bound}
 \end{equation}
 we discuss the condition for the integral in the right-hand side to be finite. 
From Eq.~\eqref{Def:H}, $H_{,x}$ is expressed as
\begin{multline}
H_{,x}  \,=\, h_c(2x-x_+-x_-)
+
2\left[h+\frac12h_{,x}x-h_c\right]x
\\
-\left[h+h_{,x}x-h_c\right]x_+
-\left[h+h_{,x}(x-x_+)-h_c\right]x_-.
\end{multline}
Since all three functions in the square brackets converge to zero, it is possible to satisfy the inequality
\begin{eqnarray}
H_{,x}& \,<\, & h_c(2x-x_+-x_-)
+h_c\left(2\varepsilon|x|+\varepsilon_+x_++\varepsilon_-|x_-|\right)
\nonumber\\
&\,=\,&
h_c\left[2(1+\sigma \varepsilon) x -\left(1-\varepsilon_+\right) x_+ -\left(1-\sigma_-\varepsilon_-\right)x_-\right],
\end{eqnarray}
for arbitrary positive constants $\varepsilon$ and
$\varepsilon_\pm$ in the range $v>v_0$ by adopting a sufficiently large $v_0$, 
where we introduced $\sigma=\mathrm{sgn}(x)$ and $\sigma_-=\mathrm{sgn}(x_-)$ in the second line. 
Integrating this inequality, we find
\begin{equation}
\exp\left[\frac12\int_{v_0}^{v^\prime} 
H_{,x}dv^{\prime\prime}\right]
\, < \,
\exp\left[S(v^\prime)-S(v_0)\right],
\end{equation}
where
\begin{equation}
S(v)\, := \,h_c\left[(1+\sigma \varepsilon) X(v) -\frac{1-\varepsilon_+}{2} X_+(v) -\frac{1-\sigma_-\varepsilon_-}{2}X_-(v)\right].
\label{Def:S(v)}
\end{equation}
Here, we introduced the primitive function of $x(v)$,
\begin{equation}
X(v)\,:=\, \int x(v)dv,
\end{equation}
in addition to $X_\pm(v)$ defined in Eq.~\eqref{Primitive-function-xpm}.
From this inequality,
the affine parameter remains finite if $\int\exp[S(v)]dv$ remains finite. This leads to the following Lemma:
\begin{lemma}
\label{Lemma:finiteness-of-lambda}
The affine parameter $\lambda$ is finite in the limit $v\to\infty$ if
there exist positive constants $\varepsilon$ and $\varepsilon_\pm$ for which 
$S(v)$ defined in Eq.~\eqref{Def:S(v)} satisfies
\begin{equation}
S(v) \,\approx\, \begin{cases}
-Av^\alpha & (A>0,\, \alpha>0),\\
-B\log v & (B>1).
\end{cases}
\end{equation}
Here, the symbol $\approx$ indicates that the leading order is equivalent. In other words,
$S(v)$ can have subleading order corrections since they do not affect whether $\lambda$ diverges or not.
\end{lemma}

We now apply this lemma to the solvable models studied in Sec.~\ref{Sec:Solvable-models}
(but note that a separate treatment is required in the model of Sec.~\ref{Sec:Solvable-Deq0}).
As we present the detailed calculations in Appendix~\ref{Finiteness-of-affine-parameters}, 
the affine parameters of outgoing null geodesics in the black hole region turn out to remain finite in the limit $v\to\infty$  in all models. 
Our calculations show that if the attractive separatrix is present, the finiteness condition is easily satisfied.
On the other hand, if there is a degenerate separatrix,
the finiteness of the affine parameter is not so obvious because
outgoing null geodesics approach the event horizon whose affine parameter diverges in the limit $v\to\infty$.
However, in this case, outgoing null geodesics 
approach the event horizon very slowly, and this behavior realizes the finiteness of the affine parameters. 
To summarize, $\mathcal{H}_{\rm I}$ is a Cauchy horizon in all models of Sec.~\ref{Sec:Solvable-models}, 
and we have not discovered any example of inextendible spacetimes up to now.

The shortcoming of Lemma~\ref{Lemma:finiteness-of-lambda} is that $S(v)$ includes $X(v)$, and hence, 
the solution of $x(v)$ must be obtained in advance. 
It would be worth challenging to derive sufficient conditions for the finiteness of $\lambda$
only in terms of $x_\pm(v)$. In particular, we have derived several sufficient conditions
for realizing various types in Cases 1, 2, and 3, and it is of interest whether those sufficient conditions
guarantee the extendibility of spacetimes beyond $\mathcal{H}_{\rm I}$. 
The study in such a direction is ongoing and we plan to discuss this issue 
in a separate paper. 
It is also interesting to look for the configuration where the spacetime is not extendible,
i.e., the affine parameter diverges in the limit $v\to\infty$.

 \newpage

%
%======================================%
%<<<<<<<<<<<< SECTION VII  >>>>>>>>>>>>>>%
%======================================%
%
 \section{\label{sec-conclusions}Conclusions}

In this paper, we have studied the global structures of spherically symmetric regular black holes that evaporate over infinite periods of time, supposing that the regular black holes asymptote to the extremal states. 
Each spacetime has the inner and outer AHs, and we have not imposed strong restrictions on their radii,
and hence, have carried out fairly general analyses. 
The radius of the outer AH is supposed to be a decreasing function of the ingoing null coordinate $v$, and we have considered the three cases, Cases 1, 2, and 3, where the radius of the inner AH is constant, increases, or decreases
with respect to $v$, respectively. 
By studying the possible behavior of outgoing null geodesics in each case, we have established the classification of resultant Penrose diagrams without loss of generality. 
Here, we have taken into account not only of the existence/nonexistence of the event horizon,
but also of the relative positions of the outer and inner AHs.
Then, we have derived some sufficient conditions to judge the type of a given spacetime. 
Our analyses have revealed that the velocity of the outer and inner AHs in approaching $r=r_c$ is a significant factor in determining the spacetime structures. 

 In Case~1 spacetimes, two distinct types, Type~1 and Type~2, have been defined in Definition~\ref{Def:Case1-Type1-Type2} depending on whether there is an outgoing null geodesic that does not intersect the outer AH between the two AHs 
  other than the event horizon. 
 The Penrose diagrams of the two types are shown in Fig.~\ref{pic-case1-diagrams}, and 
 some examples of Case~1 spacetimes are summarized in Table~\ref{table-Case1}.

 In Case~2 spacetimes, similarly to Case 1, two distinct types, Type~1 and Type~2, have been defined in Definition~\ref{Def:Case2-Type1-Type2} depending on whether there is an outgoing null geodesic that does not intersect one of the two AHs other than the event horizon. 
 The Penrose diagrams of the two types are shown in Fig.~\ref{pic-case2-diagrams}, and 
 some examples of Case~2 spacetimes are summarized in Table~\ref{table-Case2}.

 In Case~3 spacetimes, four distinct types have been defined in Definitions~\ref{Def:Case3-Class1-Class2},
 \ref{Def:Case3-ClassA-ClassB} and \ref{Def:Case3-Types}. 
 In Definition~\ref{Def:Case3-Class1-Class2}, Class 1 and Class 2 spacetimes have been defined depending on whether there is an outgoing null geodesic emitted from the inner AH that does not reach the outer AH.
 In Definition~\ref{Def:Case3-ClassA-ClassB}, Class A and Class B spacetimes have been defined depending on whether there is an outgoing null geodesic inside the inner AH that does not reach the inner AH.
 Then, the four types, Types 1-A, 1-B, 2-A, and 2-B have been defined depending on which two Classes the spacetime belongs to.
 The Penrose diagrams of the four types are shown in Fig.~\ref{pic-case3_diagrams}, and 
  the examples are summarized in Table \ref{table-Case3}. 

In addition, in order to clarify the concrete behavior of outgoing null geodesics, we have studied some models
for which the equation for outgoing null geodesics is solvable.
We have found that there appears an attractive separatrix, in addition to the repulsive separatrix which coincides with the event horizon.
The existence/nonexistence and relative positions of the two separatrices, whether the two separatrices are degenerate or not,
and the strength of the attractive property of the attractive or degenerate separatrix
have been found to be the important factors to determine the spacetime structure.

Using these solvable models, we have also examined the extendibility of the future edge inside the black hole region given by
$v=\infty$ and $r=r_c$ (denoted as $\mathcal{H}_{\rm I}$). 
The result is that in all models, the null geodesics are extendible at $\mathcal{H}_{\rm I}$, indicating the formation of a Cauchy horizon. The Hayward spacetime studied in Appendix~\ref{sec-Hayward} is also extendible. 
However, the analysis in this paper is restricted to the specific solvable models.
The study on the sufficient conditions for the extendibility only in terms of the behavior of the outer and inner AH is ongoing, and we hope to report this work in a separate paper.

 Our research is expected to make an important contribution to studies on the information loss problem and the stability related to the formation of a Cauchy horizon, -- the areas where the global structures of black hole spacetimes play a crucial role.  In particular, the condition for the nonexistence of the event horizon (Type 1-A of Case 3) clarified in this paper would play an important role because it suggests the possibility that information may not be trapped in the black hole.
 Furthermore, our study has possibility to be related to the area of primordial black holes, since if black holes in our universe are regular black holes, then the property of the evaporation must be changed a lot at its final stage.

\bigskip
\bigskip
\bigskip

\begin{flushleft}
{\bf Acknowlegments}
\end{flushleft}

We thank Ken-ichi Nakao for helpful comments.
K.S. is in part supported by JST, the establishment of university fellowships towards the creation of science technology innovation, Grant Number JPMJFS2138. H. Y. is in part supported by JSPS KAKENHI Grant Numbers JP22H01220 and
JP21H05189,
and is partly supported by MEXT Promotion of Distinctive Joint Research Center Program  JPMXP0723833165.

 \newpage
 \appendix
 \def\thesection{\Alph{section}}

%
%======================================%
%<<<<<<<<<<<< APPENDIX A  >>>>>>>>>>>>>>%
%======================================%
%
 \section{\label{sec-Hayward}Geometry and properties of the Hayward black hole} 
 
 In this Appendix, we examine the spacetime structure of an evaporating regular black hole in a specific model, i.e., the Hayward black hole. 
 First, we shall review the geometry and properties of the Hayward black hole. 
 The static Hayward black hole is described by the metric of Eq.~\eqref{eq-metric} with 
 \begin{align}
   f(r)\,=\,1-\frac{2mr^2}{r^3+2m\ell ^2}, \qquad A(v,r)\,=\,1,
 \end{align}
 where $m$ is the mass of the black hole and $\ell $ is a length parameter introduced to smear out the curvature singularity.
 The equation $f(r)=0$ gives three solutions. The largest one is the radius of the event horizon, $r_+$,
 and the second largest one is the radius of the Cauchy horizon, $r_-$.
 These two solutions exist in the extremal or sub-extremal state.
 Other than these two solutions, there always exists a negative solution, which we denote as $r=-q\ell$ where $q$ is a positive number. 
 We normalize the physical quantities with $\ell$ as $\mu=m/\ell$, $\rho=r/\ell$, and $\rho_\pm=r_\pm/\ell $.
 Then, $\mu$ is parametrically given as
 \begin{align}
  \label{eq-Hayward-mass}
   \mu\,=\,\frac{q^3}{2(1-q^2)},
 \end{align}
 and $f(\rho)$ is expressed as $f(\rho)=h(\rho)(\rho-\rho_+)(\rho-\rho_-)$ with
 \begin{equation}
 \rho_\pm \, = \, \frac{q}{2(1-q^2)}\left[1\pm\sqrt{4q^2-3}\right],
 \end{equation}
 and 
 \begin{equation}
 h(\rho)\,=\, \frac{(1-q^2)(\rho+q)}{(1-q^2)\rho^3+q^3}.
 \end{equation}
 The Hayward black hole becomes extremal when $q=\sqrt{3}/2$, and correspondingly, the mass becomes
 $\mu\,=\,\mu_c:=\frac{3\sqrt{3}}{4}$ and the radius of the extremal horizon is $\rho_c=\sqrt{3}$.
 In the extremal case, the function $h(\rho)$ takes the value $h_c=1/3$  on the extremal horizon.
 On the other hand, $q\to 1$ gives the large-mass limit.

 Here, we introduce a more useful parameter $p$ compared to $q$ via
 \begin{equation}
 p\,=\, \sqrt{4q^2-3},
 \end{equation}
 for which the extremal limit is given as $p\to 0$ while the large-mass limit is $p\to 1$. 
 We introduce the shifted radial coordinate $x=\rho-\rho_c$
 and the shifted radii of the two AHs, $x_\pm = \rho_\pm-\rho_c$ in the same manner as 
 Eqs.~\eqref{shifted-radial-coordinate} and \eqref{horizons-shifted-radial-coordinate}.
 Rewriting $f(\rho(x))$ and $h(\rho(x))$ as $f(x)$ and $h(x)$, respectively, we obtain
 \begin{equation}
 f(x)\, = \, h(x)(x-x_+)(x-x_-)
 \end{equation}
 with
 \begin{equation}
 x_\pm \, = \, \frac{\pm2p(3\mp p)}{(1\mp p)\left[\sqrt{3+p^2}+\sqrt{3}(1\mp p)\right]},
 \end{equation}
 and 
 \begin{equation}
 h(x)\, = \, \frac{(1-p^2)(2\sqrt{3}+2x+\sqrt{3+p^2})}{2(1-p^2)(\sqrt{3}+x)^3+(3+p^2)^{3/2}}.
 \end{equation}
 The mass parameter is rewritten as
 \begin{equation}
 \mu \, = \, \frac{(3+p^2)^{3/2}}{4(1-p^2)}.
 \end{equation}
 The merit of introducing the parameter $p$ is that the approximate formulas near the extremal state
 become apparent, e.g., $x_\pm\approx \pm \sqrt{3}p$.

 We now consider the evaporation of a Hayward black hole due to Hawking radiation. 
  To describe black hole evaporation, we extend the Hayward black hole metric to a dynamical one by requiring the mass to be dependent on the ingoing null coordinate as $m= m(v)$ in the same way as the Hayward's original paper~\cite{Hayward:2005gi}. 
 This is equivalent to that $p$ is promoted from a constant to a dynamical variable, $p(v)$.
 We impose the following two conditions to describe the black hole evaporation:
 \begin{enumerate}
     \item The energy flux emitted by evaporating black holes obeys the Stefan-Boltzmann law throughout the evolution;
     \item The rate of decrease in the black hole mass with respect to the ingoing null coordinate $\frac{dm}{dv}$ is identical to the energy flux by the Stefan-Boltzmann law.
 \end{enumerate}
 Based on these two assumptions, the evolution equation of $m(v)$ is determined by the equation,
 \begin{align}
   c^2\frac{dm}{dv}&\,=\,-\sigma T_{\rm H}^4 A_{\rm H},
   \label{mass-evolution-Hayward}
 \end{align}
 where $A_{\rm H}=4\pi r_+^2$ is the area of the outer AH and $\sigma=\pi^2k^4/60\hbar^3c^2$ is the Stefan-Boltzmann constant. Here, we explicitly show the speed of light $c$, the reduced Planck constant $\hbar$,
 and Boltzmann constant $k$ for clarity. 
 $T_{\rm H}$ is the Hawking temperature given in terms of the surface gravity $\kappa$ as $kT_{\rm H}=(\kappa/2\pi)\hbar c$. 
 In terms of quantities normalized with $\ell$, Eq.~\eqref{mass-evolution-Hayward} is rewritten as
 \begin{equation}
   \label{eq-Hayward-mass-decay}
 \frac{d\mu}{dV}\, = \, -\frac{1}{240\pi}\left(\frac{\ell_{\rm P}}{\ell}\right)^2\hat{\kappa}^4\rho_+^2,
 \end{equation}
 where  $V=v/\ell$, $\hat{\kappa}=\ell\kappa$, and $\ell_{\rm P}$ is the Planck length.
 The normalized surface gravity $\hat{\kappa}$ is calculated as
 \begin{equation}
 \hat{\kappa} \, = \, \frac12f^\prime (x_+)\, = \,\frac{p(3-p)(1-p)}{(3+p^2)^{3/2}}.
 \end{equation}
 Rewriting Eq.~\eqref{eq-Hayward-mass-decay} as the equation for $p(V)$, we have
  \begin{equation}
 \frac{dp}{dV}\, = \, -\frac{3}{2a^2}\cdot
 \frac{p^3(1+p)^2(1-p)^4(1-p/3)^3}{(1+p^2/3)^{11/2}(1+p/3)},
 \label{Equation-for-p}
 \end{equation}
 where $a$ is determined in terms of the Planck length $\ell_{\rm P}=\sqrt{G\hbar/c^3}$ as
 \begin{equation}
 a \, = \, 3^{11/4}\sqrt{{10\pi}}\left(\frac{\ell}{\ell_{\rm P}}\right).
 \end{equation}

 \begin{figure}[t]
   \centering
   \includegraphics*[width=0.5\linewidth]{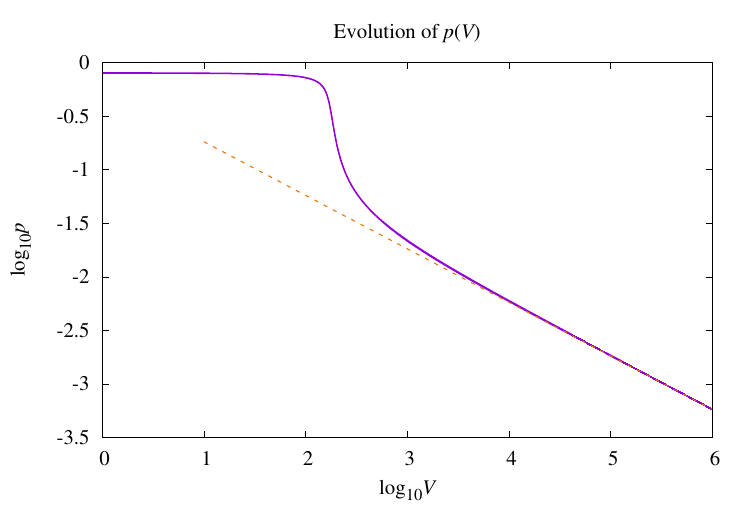}
   \caption{The evolution of $p(V)$ for the choice $a=1$ with the initial condition $p(0)=0.8$. The axes for $V$ and $p$ are shown in the logarithmic scale. The line that corresponds to $p=a/\sqrt{3V}$ is drawn by a dashed line for comparison.}
   \label{pic-Hayward-p}
 \end{figure}

 \begin{figure}[t]
   \centering
   \includegraphics*[width=0.45\linewidth]{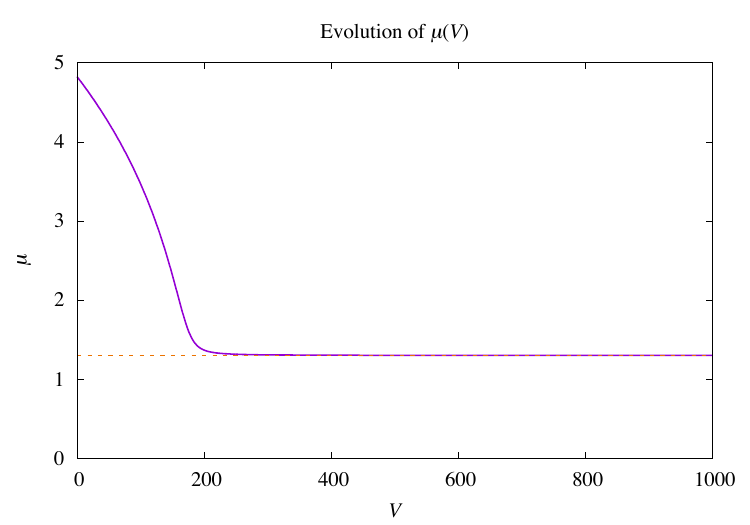}
   \includegraphics*[width=0.45\linewidth]{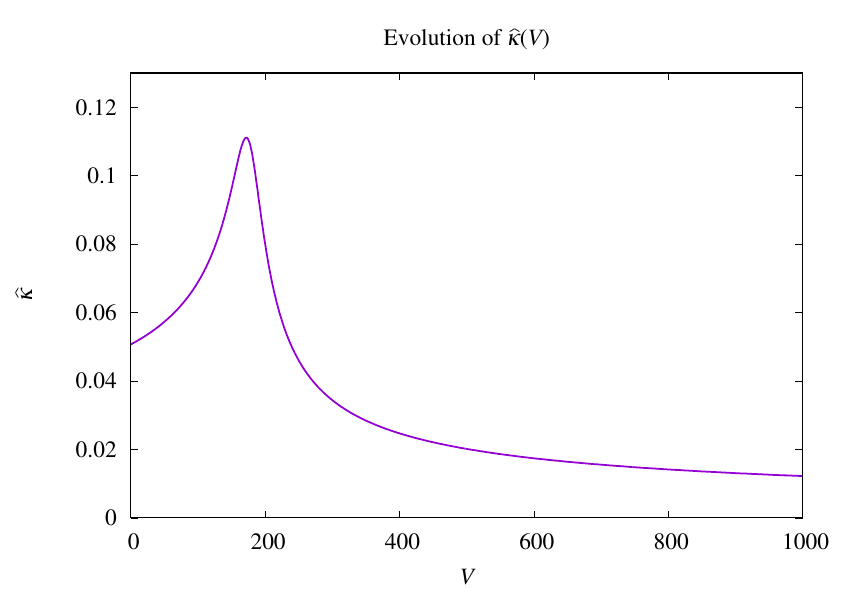}
   \caption{Time evolutions of the normalized mass of an evaporating Hayward black hole, $\mu(V)$ (left) and the normalized surface gravity $\hat{\kappa}(V)$ (right) with $a=1$ and $p(0)=0.8$. In the left panel, the value of $\mu=\mu_c=3\sqrt{3}/4$ is shown by a dashed line for comparison.}
   \label{pic-Hayward-some-variables}
 \end{figure}

 We study the asymptotic behavior of $p(V)$ in the range of large $v$. In this case, 
 $p(V)$ is expected to asymptote to zero and hence, is very small because the black hole approaches the extremal state.
 For this reason, Eq.~\eqref{Equation-for-p} is approximated as
 \begin{equation}
 \frac{dp}{dV}\,\approx\, -\frac{3}{2a^2}p^3.
 \end{equation}
 This is immediately solved as
 \begin{equation}
 p \,\approx\, \frac{a}{\sqrt{3(V-V_0)}} \,\approx\, \frac{a}{\sqrt{3V}},
 \label{p(V)-approximate}
 \end{equation}
 where $V_0$ is the integral constant, and we omitted $V_0$ in the second approximate equality since its effect is a subdominant order. 
 This means 
 \begin{equation}
 x_\pm \, \approx \,\pm\frac{a}{\sqrt{V}}.
 \end{equation}
 Then, the inner horizon is spacelike (i.e., Case 2) and Table \ref{table-Case2} shows that the evaporating Hayward black hole is of Type~1 in Case 2.

 The numerical calculation is also performed.
 Here, the parameters are chosen as $a=1$ and  $p(0)=0.8$, which corresponds to $\rho_+(0)\approx 9.539$.
 This choice is mainly due to the numerical convenience, and we have to keep in mind that
 the choice $a=1$ may not be very realistic because $\ell\sim 10^{-2} \ell_{\rm P}$ holds in this case. 
 The numerical solution to Eq.~\eqref{Equation-for-p} is presented in Fig.~\ref{pic-Hayward-p}. 
 The both axes are shown in the logarithmic scale, and the line that corresponds to the approximate formula
 of Eq.~\eqref{p(V)-approximate} is shown for comparison. 
 It is found that the behavior of $p(V)$ is well approximated by
 Eq.~\eqref{p(V)-approximate} for $V\gtrsim 10^4$.
 The time evolutions of the normalized mass $\mu$ and the normalized surface gravity $\hat{\kappa}$ 
are shown in the left and right panels of Fig.~\ref{pic-Hayward-some-variables}, respectively.
The mass decreases and asymptotes to the value $\mu_c=3\sqrt{3}/4$.
The surface gravity (and hence, the Hawking temperature) increases first, but at $V\approx 172$, starts to decrease
reflecting the fact that the black hole approaches the extremal state.

From the above calculations, the asymptotic behavior of the evaporating Hayward black hole has been understood.
We then analyze the behavior of the outgoing null geodesics in the region $x_-\lesssim x \lesssim x_+$, and assess the extendibility of this spacetime. As we have shown in Sec.~\ref{sec-case2}, this spacetime has an inner boundary $\mathcal{H}_{\rm I}$ which the outgoing null geodesics asymptotically approach in the limit $v\rightarrow\infty$. 
In order to examine the extendibility, it is necessary to check the finiteness of the affine parameter of outgoing null geodesics inside the black hole region using the formula of Eq.~\eqref{lambda-formula}. For this purpose, setting $T=\sqrt{V}$, 
analytic solutions to Eq.~\eqref{eq-null-x} for the outgoing null geodesics 
will be obtained under the approximation where the orders up to $O(1/T)$ are taken into account.
In this approximation, $x$ and $x_\pm\approx \pm a/T$ are considered to be $O(1/T)$,
and the leading-order of $h(v,x)=h_c=1/3$ is adopted for $h(v,x)$. Expressing
the approximate formula for $x(T)$ as
 \begin{align}
    x(T)=a\frac{\chi(T)}{T},
 \end{align}
 where $\chi(T)$ is $O(1)$ quantity, 
 the equation for $\chi(T)$ is written as
 \begin{equation}
 \frac{d\chi}{dT}\, = \, \frac{1}{3}a(\chi^2-1)+\frac{\chi}{T}.
 \end{equation}
 Here, the second term on the right-hand side must be ignored because it is a subdominant order.
 In this approximation, $\chi=\pm 1$ are the constant solutions, which correspond to the repulsive
 and attractive separatrices discussed in Sec.~\ref{Sec:Solvable-models} (the repulsive separatrix is identical to the event horizon).
Although these separatrices coincide with the outer and inner AHs in the current approximation,
 they must deviate from each other if the higher-order terms are taken into account.
 Other than the separatrices, the solutions are
 \begin{align}
     \chi(V)=
     \begin{cases}
         -{\coth{\left[\frac{a}{3}\left(\sqrt{V}-\sqrt{V_d}\right)\right]}} \quad (V<V_d),\\
         -{\tanh{\left[\frac{a}{3}\left(\sqrt{V}-\sqrt{V_c}\right)\right]}} \quad (V_0<V<\infty),\\
         -{\coth{\left[\frac{a}{3}\left(\sqrt{V}-\sqrt{V_e}\right)\right]}} \quad (V_e<V).
     \end{cases}
     \label{Hayward-outgoing-null-approximate-solution}
 \end{align}
 Note that, for the outgoing null geodesics outside the repulsive separatrix, this approximation is immediately broken because $\chi(V)$ grows rapidly, but the solution inside the black hole region can be used for the range where $\chi(V)$ approaches $-1$.
 The obtained solution has a close resemblance to the ones in Sec.~\ref{Sec:Solvable-0<n<1}.
 Similarly to that case, all of the outgoing null geodesics are attracted to the
 attractive separatrix at $x\approx-a/\sqrt{V}$ located near the inner AH.
 Then, it is possible to apply Lemma~\ref{Lemma:finiteness-of-lambda} in Sec.~\ref{sec-extendibility}, and by the similar calculation in Appendix~\ref{Finiteness-Solvable-0<n<1}, the finiteness of the affine parameters of outgoing null geodesics at $\mathcal{H}_{\rm I}$ is guaranteed.

 \begin{figure}[t]
 \centering
\includegraphics*[width=0.75\linewidth]{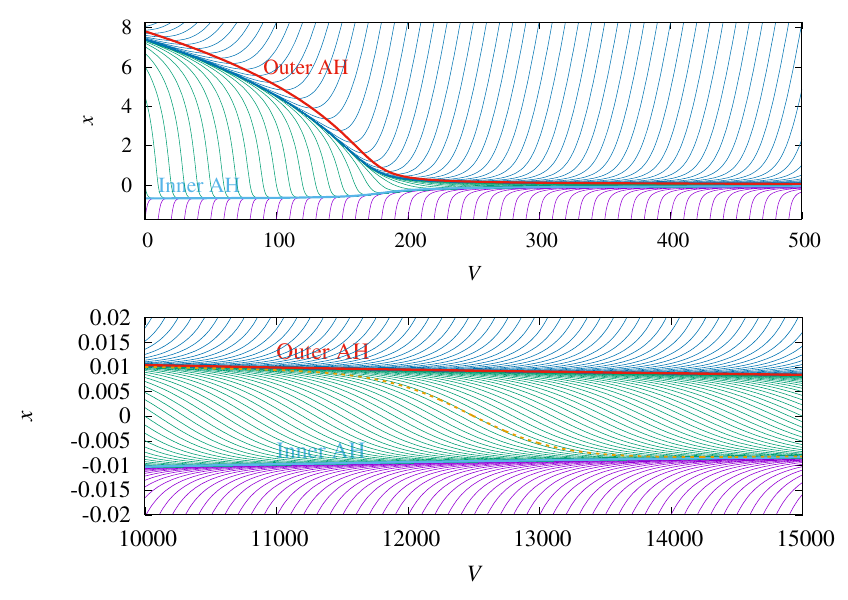}
         \caption{The behavior of the outgoing null geodesics (thin curves) and the positions of the outer and inner AHs (thick curves) in the $(V,x)$-plane for the range $0<V<500$ (top panel) and for $10000<V<15000$ (bottom panel). 
         In the bottom panel, an outgoing null geodesic given by the approximate formula of Eq.~\eqref{Hayward-outgoing-null-approximate-solution} is shown by a dotted curve for comparison.}
   \label{pic-Hayward-outgoing-null}
 \end{figure}

 As a check, we have performed a numerical analysis and obtained  
 the numerical solutions of Eq.\eqref{eq-outgoing} as shown in Figure \ref{pic-Hayward-outgoing-null}.
  The top panel shows the range $0<V<500$, i.e., the initial stage of the evaporation.
 The two AHs approach each other and the system asymptotes to the extremal state.
 Throughout the evolution, outgoing null geodesics inside the black hole are attracted to the neighborhood of the inner AH.
 The bottom panel highlights the behavior near the two AHs in a relatively late time period,
 $10000<V<15000$, where the approximation of Eq.~\eqref{Hayward-outgoing-null-approximate-solution}
 is expected to hold with the error of $O(\%)$.
 One of the approximate solutions of Eq.~\eqref{Hayward-outgoing-null-approximate-solution} is shown
 by a dotted curve, and it coincides with the numerical solution within the expected error.
 These numerical results support our conclusion on the extendibility of the spacetime
 based on the approximate analysis.

 \begin{figure}[t]
     \centering
     \includegraphics[width=0.25\linewidth]{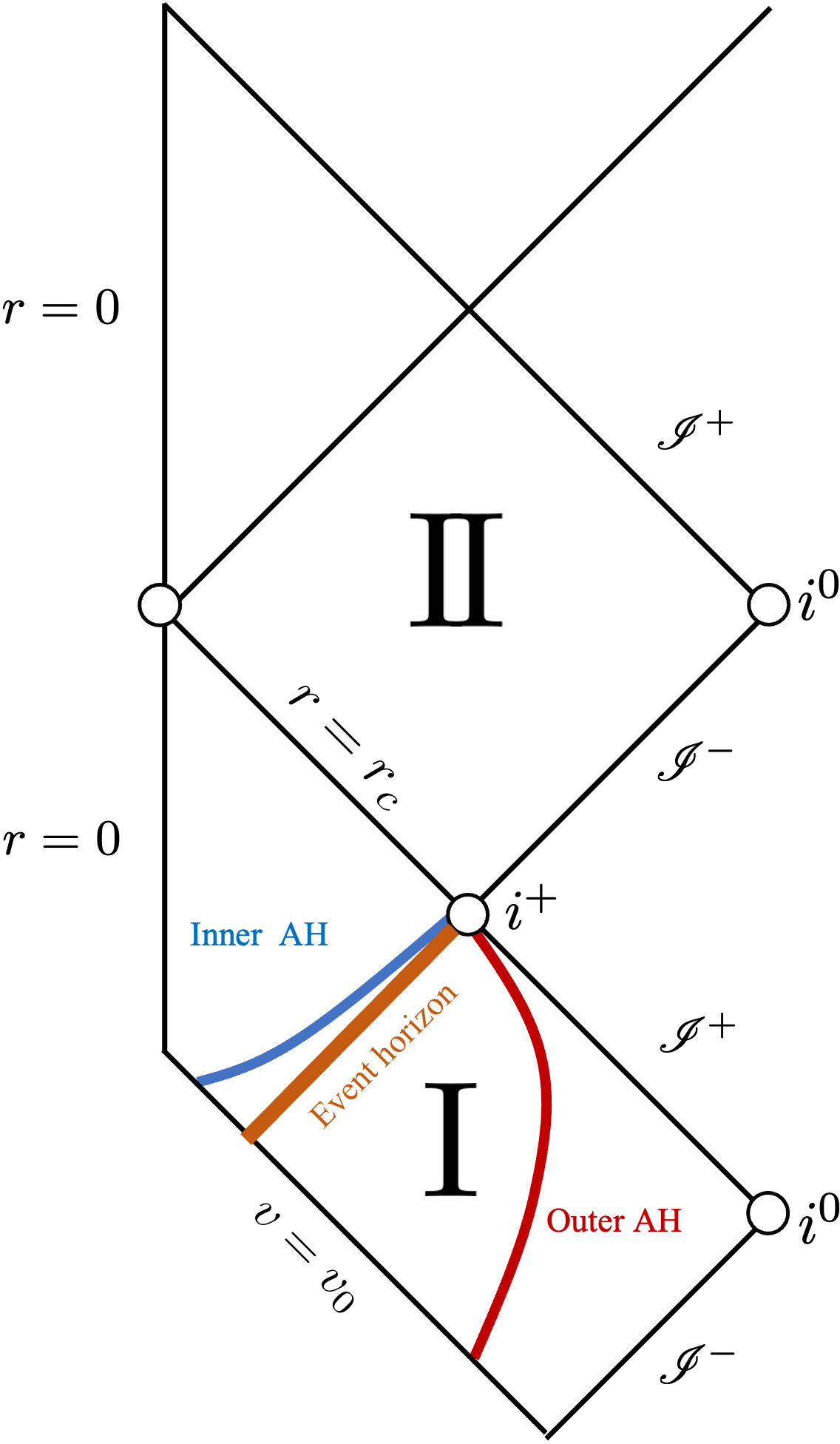}
     \caption{Maximally extended spacetime structure of an evaporating Hayward black hole.}
     \label{pic-full-diagram}
 \end{figure}

We now discuss the maximal extension of the evaporating Hayward black hole spacetime.
We denote the spacetime region represented by Eq.~\eqref{eq-metric} as the region I, and the region extended from the inner boundary $\mathcal{H}_{\rm I}$ as the region \two. 
The region \two\ is a spacetime region beyond $v=\infty$, and therefore, the region \two\ should be expressed with the extremal Hayward metric. 
The maximal Penrose diagram is shown in Figure \ref{pic-full-diagram}.

\newpage

%
%======================================%
%<<<<<<<<<<<< APPENDIX B  >>>>>>>>>>>>>>%
%======================================%
%
\section{Examples for various types of Penrose diagrams}
\label{Examples-proof}

In this appendix, we present the proofs for the examples that lead to various types of the Penrose diagrams
in Cases 1, 2, and 3 (which are summarized in Tables~\ref{table-Case1}, \ref{table-Case2}, and \ref{table-Case3}).

\subsection{Case~1 spacetimes}
\label{Examples:case1}

We begin with the Case~1 spacetimes, focusing attention on the power-law functions.

\subsubsection{Examples of Type~1 in Case~1}
Here, we present some examples of the functions of $x_+(v)$ that lead to Type~1.
\begin{example}
  Case~1 spacetimes with $x_+=\frac{a_+}{v^n}$ satisfying $n>1$ are of Type~1.
\end{example}
\begin{proof}
  Since these spacetimes satisfy
  \begin{align}
    \lim_{v\rightarrow\infty}x_+(v)\exp\left[\frac{1+\varepsilon}{2}h_c X_+(v)\right]=
    \lim_{v\rightarrow \infty}\frac{a_+}{v^n}\exp{\left[-\frac{(1+\varepsilon)a_+h_c}{2(n-1)v^{n-1}}\right]}=0,
  \end{align}
  the condition of Proposition \ref{prop-Case1-Type1} is satisfied.
\end{proof}
\begin{example}
  Case~1 spacetimes with $x_+=\frac{a_+}{v}$ satisfying $a_+<\frac{2}{h_c}$ are of Type~1.
\end{example}
\begin{proof}
  Since $\frac{a_+h_c}{2}<1$ holds, choosing a sufficiently small $\varepsilon$ leads to $\frac{(1+\varepsilon)a_+h_c}{2}<1$. Then, the function $x_+(v)$ satisfies
  \begin{align}
    \lim_{v\rightarrow\infty}x_+(v)\exp\left[\frac{1+\varepsilon}{2}h_c X_+(v)\right]=\lim_{v\rightarrow \infty}a_+v^{\frac{(1+\varepsilon)a_+h_c}{2}-1}=0,
  \end{align}
  and thus, the condition of Proposition \ref{prop-Case1-Type1} is satisfied.
\end{proof}

\subsubsection{Examples of Type~2 in Case~1}
Here, we present some examples of the functions of $x_+(v)$ that lead to Type~2.

\begin{example}
  \label{example-Case1-Type2-1}
  Case~1 spacetimes with $x_+=\frac{a_+}{v^n}$ satisfying $0<n<1$ are of Type~2.
\end{example}
\begin{proof}
  Since Eq.~\eqref{eq-Case1-sufficient-2-1} is evaluated as
  \begin{align}
     \lim_{v\rightarrow \infty}x_+(v)\exp{\left[\frac{1-\alpha}{2}h_c X_+(v)\right]}
     \,=\,\lim_{v\rightarrow \infty}\frac{a_+}{v^n}\exp{\left[\frac{(1-\alpha)a_+h_c}{2(1-n)}v^{1-n}\right]}\,=\,\infty,
  \end{align}
  the condition of Proposition \ref{prop-Case1-prop2} is satisfied.
\end{proof}

\begin{example}
  \label{example-Case1-Type2-2}
  Case~1 spacetimes with $x_+=\frac{a_+}{v}$ satisfying $a_+>\frac{2}{h_c}$ are of Type~2.
\end{example}
\begin{proof}
  Since $\frac{a_+h_c}{2}>1$ holds, choosing sufficiently small $\alpha$ leads to $\frac{(1-\alpha)a_+h_c}{2}>1$.
  Then,  Eq.~\eqref{eq-Case1-sufficient-2-1} is evaluated as
  \begin{align}
     \lim_{v\rightarrow \infty}x_+(v)\exp{\left[\frac{1-\alpha}{2}h_c X_+(v)\right]}
     \,=\,\lim_{v\rightarrow\infty}a_+v^{\frac{(1-\alpha)h_ca_+}{2}-1}
     \,=\,\infty,
  \end{align}
  and therefore, the condition of Proposition \ref{prop-Case1-prop2} is satisfied.
\end{proof}
It is also possible to confirm that Examples 3 and 4 are of Type~2 using Corollary~\ref{corollary-Case1-1}.

\subsection{Case~2 spacetimes}
\label{Examples:case2}

Next, we turn our attention to Case~2 spacetimes.

\subsubsection{Examples of Type~1 in Case~2}

We present some examples of the functions of $x_+(v)$ and $x_-(v)$ that lead to Type~1 spacetimes. 
Although the power-law functions are primarily considered, the exponential behavior is also taken into account in some cases. \begin{example}
  Case~2 spacetimes with $x_+=\frac{a_+}{v^n}$ and $x_-=\frac{a_-}{v^k}$ satisfying $0<k<n$ are of Type~1.
\end{example} 
\begin{proof}
  Since these spacetimes satisfy
  \begin{align}
    \lim_{v\rightarrow\infty}\frac{x_+}{|x_-|}=\frac{a_+}{|a_-|}\frac{1}{v^{n-k}}=0,
  \end{align}
   the condition of Corollary \ref{corollary-case2-Type1-2} is satisfied.
\end{proof}
\begin{example}
  Case~2 spacetimes with $x_+=\frac{a_+}{v^n}$ and $x_-=\frac{a_-}{v^n}$ satisfying $a_+<|a_-| $ are of Type~1.
\end{example}
\begin{proof}
  Since these spacetimes satisfy
  \begin{align}
    \lim_{v\rightarrow\infty} \frac{x_+}{|x_-|}=\frac{a_+}{|a_-|}<1,
  \end{align}
  the condition of Corollary \ref{corollary-case2-Type1-2} is satisfied.
\end{proof}
\begin{example}
  Case~2 spacetimes with $x_+=\frac{a_+}{v^n}$  and $x_-=\frac{a_-}{v^k}$ satisfying $n>1$ and $k>1$ are of Type~1.
\end{example}
\begin{proof}
  Since we have $X_+=-a_+/(n-1)v^{n-1}$ and $X_-=|a_-|/(k-1)v^{k-1}$, we have
  \begin{align}
    \lim_{v\rightarrow \infty}\left[(1+\varepsilon_+)X_++(1-\varepsilon_-)X_-\right]
    \,=\,\lim_{v\rightarrow \infty}\left[-\frac{(1+\varepsilon_+)a_+}{(n-1)v^{n-1}}+\frac{(1-\varepsilon_-)|a_-|}{(k-1)v^{k-1}}\right]
    \,=\, 0 \,\neq\, \infty.
  \end{align}
  Thus, the condition of Corollary \ref{corollary-case2-Type1-1} is satisfied.
\end{proof}
\begin{example}
  Case~2 spacetimes with $x_+=\frac{a_+}{v}$ and $x_-=\frac{a_-}{v^k}$ satisfying $a_+h_c<2<2k$ are of Type~1.
\end{example}
\begin{proof}
From the assumption, it is possible to adopt a sufficiently small positive constant $\varepsilon_+$ satisfying $h_ca_+(1+\varepsilon_+)<2$. 
Substituting $X_+=a_+\log v$ and $X_-=|a_-|/(k-1)v^{k-1}$, the left-hand side of Eq.~\eqref{eq-case2-Type1-2-1} is evaluated as
  \begin{align}
    &\lim_{v\rightarrow\infty}x_+\exp{\left[\frac{h_c}{2}\left\{(1+\varepsilon_+)X_++(1-\varepsilon_-)X_-\right\}\right]}\nonumber\\
    &\,=\,\lim_{v\rightarrow\infty}a_+v^{\frac12(1+\varepsilon)h_ca_+-1}\exp{\left[\frac{h_c|a_-|(1-\varepsilon_-)}{2(k-1)v^{k-1}}\right]}
    \,=\,0.
  \end{align}
  Then these spacetimes satisfy the condition of Proposition \ref{prop-case2-Type1-2}.
\end{proof}

\begin{example}
  Case~2 spacetimes with $x_-=\frac{a_-}{v^k}$ satisfying $0<k\le 1$ and an arbitrary decreasing function $x_+$ are of Type~1.
\end{example}
\begin{proof}
  Since the primitive function of $x_-$ diverges to $-\infty$ in the limit $v\rightarrow\infty$, Eq.~\eqref{eq-Case2-Type1-3} of Proposition \ref{prop-Case2-Type1-3} is trivially satisfied.
\end{proof}

\subsubsection{Examples of Type~2 in Case~2}

Next, we introduce some examples of Type~2 spacetimes in Case~2.
\begin{example}
  Case~2 spacetimes with $x_+=\frac{a_+}{v^n}$ and $x_-=a_-\exp{\left(-\frac{v}{L}\right)}$ satisfying $ 0<n<1$ are of Type~2.
\end{example}
\begin{proof}
  For these functions $x_+$ and $x_-$, the left-hand sides of Eqs.~\eqref{eq-case2-Type2-2-1} and \eqref{eq-case2-Type2-2-2} of Corollary \ref{corollary-case2-Type2-1} are evaluated as
  \begin{subequations}
  \begin{align}
      \lim_{v\rightarrow\infty}\frac{dx_+/dv}{x_+^2}&\,=\,\lim_{v\rightarrow\infty}-\frac{n}{a_+}v^{n-1}=0,\\
      \lim_{v\rightarrow\infty}\frac{dx_-/dv}{x_+x_-}&\,=\,\lim_{v\rightarrow\infty}-\frac{1}{a_+L}v^n=-\infty.
  \end{align}
  \end{subequations}
  Then these spacetimes satisfy the conditions of Corollary \ref{corollary-case2-Type2-1}.
\end{proof}
\begin{example}
  Case~2 spacetimes with $x_+=\frac{a_+}{v}$ and $x_-=\frac{a_-}{v^k} $ satisfying $ 2<a_+h_c<2k $ are of Type~2.
\end{example}
\begin{proof}
  From the assumption $2<a_+h_c$, it is possible to take a sufficiently small $\alpha$ that satisfies $2<a_+h_c(1-\alpha)$, and from the assumption $a_+h_c<2k$, it is possible to realize the condition $a_+h_c<\frac{2k\beta}{1+\beta}$ by taking a sufficiently large $\beta$. 
  Evaluating the left-hand sides of Eqs.~\eqref{eq-case2-Type2-2-1} and \eqref{eq-case2-Type2-2-2}, we have
  \begin{subequations}
  \begin{align}
      \lim_{v\rightarrow\infty}\frac{dx_+/dv}{x_+^2}&\,=\,-\frac{1}{a_+}>-\frac{h_c}{2}(1-\alpha),\\
      \lim_{v\rightarrow \infty}\frac{dx_-/dv}{x_+x_-}&\,=\,-\frac{k}{a_+}<-\frac{h_c}{2}\frac{1+\beta}{\beta}.
        \end{align}
        \end{subequations}
  Then the spacetimes satisfy the conditions of Corollary \ref{corollary-case2-Type2-1}.
\end{proof}

\subsection{Case~3 spacetimes}
\label{Examples:case3}

Finally, we present some examples of the functions of $x_+(v)$ and $x_-(v)$ that lead to Type~1-A, Type~1-B, Type~2-A and Type~2-B, respectively. 
Although the power-law functions are primarily considered, the exponential behavior is also taken into account in some cases.

\subsubsection{Examples of Type~1-A in Case~3}

We begin with the Type~1-A spacetimes in Case~3.
\begin{example}
  Case~3 spacetimes with $x_{\pm}=\frac{a_\pm}{v}$ satisfying
  \begin{align}
    \frac{2}{h_{c}}<a_+<a_-+\frac{2}{h_c}
  \end{align}
  are of Type~1-A.
\end{example}
\begin{proof}
  We evaluate Eq.~\eqref{eq-Case3-Type1-1} of Proposition \ref{prop-Case3-Type1}. 
  Since $(a_+-a_-)h_c<2$ holds, taking sufficiently small $\varepsilon$ leads to $(1+\varepsilon)(a_+-a_-)h_c<2$. Then,
  \begin{align}
    \lim_{v\rightarrow\infty}(x_+(v)-x_-(v))\exp{\left[\frac{1+\varepsilon}{2}h_c (X_+(v)-X_-(v))\right]}
    &=\lim_{v\rightarrow\infty}(a_+-a_-)v^{\frac{1+\varepsilon}{2}(a_+-a_-)h_c-1}\nonumber\\
    &=0
  \end{align}
  is satisfied. Therefore, these spacetimes are of Class~1. 
  Next, we evaluate Eq.~\eqref{eq-Case3-TypeA-1} of Proposition \ref{prop-Case3-TypeA}.
  Since $a_+h_c>2$ holds, the condition 
  $a_+h_c(1-\varepsilon)>2$ is also satisfied by
  adopting a sufficiently small $\varepsilon$. Then, 
  \begin{align}
    \lim_{v\rightarrow\infty}x_-(v)\exp{\left[\frac{h_{c}}{2}(1-\varepsilon)X_+(v)\right]} \,=\,\lim_{v\rightarrow\infty}a_-v^{\frac{a_+h_{c}}{2}(1-\varepsilon)-1}
    \,=\,\infty.
  \end{align}
  Therefore, these spacetimes are of Class~A. 
  Consequently, these spacetimes are of Type~1-A.
\end{proof}

\begin{example}
  Case~3 spacetimes with $x_+=\frac{a_+}{v^n}$ satisfying $0<n<1$ and $x_-=x_+-\frac{b}{v^k}$ satisfying $b>0$ and $k>1$ are of Type~1-A.
\end{example}
\begin{proof}
  We evaluate Eq.~\eqref{eq-Case3-Type1-1} of Proposition \ref{prop-Case3-Type1};
  \begin{align}
    \lim_{v\rightarrow\infty}(x_+(v)-x_-(v))\exp{\left[\frac{1+\varepsilon}{2}h_c (X_+(v)-X_-(v))\right]}&=\lim_{v\rightarrow\infty}\frac{b}{v^k}\exp{\left[-\frac{(1+\varepsilon)h_c}{2}\frac{b}{(k-1)v^{k-1}}\right]}\nonumber\\
    &=0.
  \end{align}
  Therefore, these spacetimes are of Class~1. Next, we evaluate Eq.~\eqref{eq-Case3-TypeA-1} of Proposition \ref{prop-Case3-TypeA};
  \begin{align}
    \lim_{v\rightarrow\infty}x_-(v)\exp{\left[\frac{h_{c}}{2}(1-\varepsilon)X_+(v)\right]} &=\lim_{v\rightarrow\infty}\left(\frac{a_+}{v^n}-\frac{b}{v^k}\right)\exp{\left[\frac{h_{c}a_+(1-\varepsilon)}{2(1-n)}v^{1-n}\right]}\nonumber\\
    &=\infty.
  \end{align}
  Therefore, these spacetimes are of Class~A. Consequently, these spacetimes are of Type~1-A.
\end{proof}

\subsubsection{Examples of Type~1-B in Case~3}

Next, we would show some examples of Type~1-B spacetimes in Case~3.
%\begin{example}
%  Case~3 spacetimes with $x_\pm= \frac{a_\pm}{v^n}$ satisfying $n>1$ are of Type~1-B.
%\end{example}
%\begin{proof}
%  We evaluate Eq.~\eqref{eq-Case3-Type1-3} of Corollary \ref{corollary-Case3-Type1};
%  \begin{align}
%    \lim_{v\rightarrow\infty}x_+(v)\exp{\left[\frac{(1+\varepsilon)h_c}{2}X_+(v)\right]}&=\frac{a_+}{v^n}\exp{\left[-\frac{(1+\varepsilon)h_c}{2}\frac{a_+}{(n-1)v^{n-1}}\right]}\nonumber\\
%    &=0.
%  \end{align}
%  Therefore, these spacetimes are of Class~1. Next, we evaluate Eq.~\eqref{eq-Case3-TypeB-3} of Corollary \ref{corollary-Case3-TypeB};
%  \begin{align}
%    \lim_{v\rightarrow \infty}\frac{dx_-/dv}{(x_++\alpha x_-)x_-}\,=\,\lim_{v\rightarrow\infty}\frac{-n}{a_++\alpha a_-}v^{n-1}
%    \,=\,-\infty.
%  \end{align}
%  Therefore, these spacetimes are of Class~B. Consequently, these spacetimes are of Type~1-B.
%\end{proof}
\begin{example}
  Case~3 spacetimes with $x_+=\frac{a_+}{v^n}$ and $x_-=\frac{a_-}{v^k}$ satisfying $1<n\le k$ are of Type~1-B.
  Here, $a_-<a_+$ holds in the case $n=k$.
\end{example}
\begin{proof}
  We evaluate Eq.~\eqref{eq-Case3-Type1-3} of Corollary~\ref{corollary-Case3-Type1}:
    \begin{equation}
\lim_{v\to\infty}x_+(v)
\exp\left[\frac{(1+\varepsilon)h_c}{2}X_+(v)\right]
\ = \
\frac{a_+}{v^n}\exp\left[-\frac{(1+\varepsilon)h_c}{2}\ \frac{a_+}{(n-1)v^{n-1}}\right]
\ = \ 0.
    \end{equation}
  Therefore, these spacetimes are of Class~1. Next, we evaluate Eq.~\eqref{eq-Case3-TypeB-3} of Corollary \ref{corollary-Case3-TypeB};
  \begin{align}
    \lim_{v\rightarrow \infty}\frac{dx_-/dv}{(x_++\alpha x_-)x_-}\,=\,\lim_{v\rightarrow\infty}\frac{-k}{a_++\alpha a_-/v^{k-n}}v^{n-1}
    \,=\,-\infty.
  \end{align}
  Therefore, these spacetimes are of Class~B. Consequently, these spacetimes are of Type~1-B.
\end{proof}
\begin{example}
  Case~3 spacetimes with $x_\pm=\frac{a_\pm}{v}$ satisfying $\sqrt{a_+}+\sqrt{a_-}<\sqrt{\frac{2}{h_c}}$
  are of Type~1-B.
\end{example}
\begin{proof}
  We evaluate Eq.~\eqref{eq-Case3-Type1-3} of Corollary \ref{corollary-Case3-Type1}. 
  Since $a_+<(\sqrt{a_+}+\sqrt{a_-})^2<\frac{2}{h_c}$ holds, taking a sufficiently small $\varepsilon$ leads to $\frac{1+\varepsilon}{2}h_ca_+<1$. Then,
  \begin{align}
    \lim_{v\rightarrow\infty}x_+(v)\exp{\left[\frac{(1+\varepsilon)h_c}{2}X_+(v)\right]}
    \,=\, \lim_{v\rightarrow\infty}a_+v^{\frac{1+\varepsilon}{2}h_c a_+-1}
    \,=\,0.
  \end{align}
  Therefore, these spacetimes are of Class~1. Next, we evaluate the inequality of Eq.~\eqref{eq-Case3-TypeB-3}
  of Corollary \ref{corollary-Case3-TypeB};
  \begin{align}
    \lim_{v\rightarrow\infty}\frac{dx_-/dv}{(x_++\alpha x_-)x_-} 
    \,=\,\lim_{v\rightarrow\infty}\frac{-a_-/v^2}{(a_++\alpha a_-)a_-/v^2}
    \,=\,-\frac{1}{a_++\alpha a_-}.
  \end{align}
Since choosing $\alpha=\sqrt{a_+a_-}$ leads to
  \begin{align}
    -\frac{1+\alpha}{2\alpha}h_c-\left(-\frac{1}{a_++\alpha a_-}\right)
    \,=\,\frac{h_c}{2\sqrt{a_+}(\sqrt{a_+}+\sqrt{a_-})}\left[\frac{2}{h_c}-(\sqrt{a_+}+\sqrt{a_-})^2\right]\,>\,0,
  \end{align}
  we have
  \begin{align}
    \lim_{v\rightarrow\infty}\frac{dx_-/dv}{(x_++\alpha x_-)x_-}\,<\,-\frac{1+\alpha}{2\alpha}h_c.
  \end{align}
  Therefore, these spacetimes are of Class~B. Consequently, these spacetimes are of Type~1-B.
\end{proof}

\subsubsection{Examples of Type~2-A in Case~3}

Next, we show some examples of Type~2-A spacetimes.
%\begin{example}
%    Case~3 spacetimes with $x_\pm=\frac{a_\pm}{v^n}$ satisfying $0<n<1$ are of Type~2-A.
%\end{example}
%\begin{proof}
%  We evaluate Eq.~\eqref{eq-Case3-Type2-2} of Corollary \ref{corollary-Case3-Type2};
%  \begin{align}
%    \lim_{v\rightarrow\infty}\frac{\alpha dx_+/dv+(1-\alpha)dx_-/dv}{(x_+-x_-)^2}
%    \,=\,
%    \lim_{v\rightarrow\infty}\frac{-\alpha a_++(1-\alpha)a_-}{(a_+-a_-)^2}n v^{n-1}
%    \,=\,0.
%  \end{align}
%  Therefore, these spacetimes are of Class~2. Next, we evaluate Eq.~\eqref{eq-Case3-TypeA-1} of Corollary \ref{prop-Case3-TypeA};
%  \begin{align}
%    \lim_{v\rightarrow\infty}x_-(v)\exp{\left[\frac{h_{c}}{2}(1-\varepsilon)X_+(v)\right]}
%    \,=\,\lim_{v\rightarrow\infty}\frac{a_-}{v^n}\exp{\left[\frac{h_{c}a_+(1-\varepsilon)}{2(1-n)}v^{1-n}\right]}
%    \,=\,\infty.
%  \end{align}
%  Therefore, these spacetimes are of Class~A. Consequently, these spacetimes are of Type~2-A.
%\end{proof}
\begin{example}
  Case~3 spacetimes with $x_+=\frac{a_+}{v^n}$ and $x_-=\frac{a_-}{v^k}$ satisfying $ 0<n<1$ and $k\ge n$ are of Type~2-A.
  Here, $a_+>a_-$ holds in the case $k=n$.
\end{example}
\begin{proof}
  We evaluate Eq.~\eqref{eq-Case3-Type2-2} of Corollary \ref{corollary-Case3-Type2};
  \begin{align}
    \lim_{v\rightarrow\infty}\frac{\alpha dx_+/dv+(1-\alpha)dx_-/dv}{(x_+-x_-)^2}
    \,=\,\lim_{v\rightarrow\infty}\frac{-\alpha n a_+-(1-\alpha)ka_-/v^{k-n}}{(a_+-a_-/v^{k-n})^2} v^{n-1}
    \,=\,0.
  \end{align}
  Therefore, these spacetimes are of Class~2. 
  Next, we evaluate Eq.~\eqref{eq-Case3-TypeA-1} of Proposition \ref{prop-Case3-TypeA};
  \begin{align}
    \lim_{v\rightarrow\infty}x_-(v)\exp{\left[\frac{h_{c}}{2}(1-\varepsilon)X_+(v)\right]}
    \,=\,\lim_{v\rightarrow\infty}\frac{a_-}{v^k}\exp{\left[\frac{h_ca_+(1-\varepsilon)}{2(1-n)}v^{1-n}\right]}
    \,=\,\infty.
  \end{align}
  Therefore, these spacetimes are of Class~A. Consequently, these spacetimes are of Type~2-A.
\end{proof}
\begin{example}
  Case~3 spacetimes with $x_\pm=\frac{a_\pm}{v}$ satisfying $\sqrt{\frac{2}{h_c}}<\sqrt{a_+}-\sqrt{a_-}$
  are of Type~2-A.
\end{example}
\begin{proof}
  We evaluate Eq.~\eqref{eq-Case3-Type2-2} of Corollary \ref{corollary-Case3-Type2};
  \begin{align}
    \lim_{v\rightarrow\infty}\frac{\alpha dx_+/dv+(1-\alpha)dx_-/dv}{(x_+-x_-)^2}\,=\,-\frac{\alpha a_++(1-\alpha)a_-}{(a_+-a_-)^2}.
  \end{align}
  Since choosing $\alpha=\frac{\sqrt{a_-}}{\sqrt{a_+}+\sqrt{a_-}}$ leads to
  \begin{align}
    -\frac{\alpha a_++(1-\alpha)a_-}{(a_+-a_-)^2}-\left[-\frac{h_c}{2}\alpha(1-\alpha)\right]
    \,=\,\frac{\sqrt{a_+a_-}}{(\sqrt{a_+}+\sqrt{a_-})^2}\left[\frac{h_c}{2}-\frac{1}{(\sqrt{a_+}-\sqrt{a_-})^2}\right]
    \,>\,0,
  \end{align}
  we have
  \begin{align}
    \lim_{v\rightarrow\infty}\frac{\alpha dx_+/dv+(1-\alpha)dx_-/dv}{(x_+-x_-)^2}
        \,>\,-\frac{h_c}{2}\alpha (1-\alpha).
  \end{align}
  Therefore, these spacetimes are of Class~2. Next, we evaluate Eq.~\eqref{eq-Case3-TypeA-1} of Proposition \ref{prop-Case3-TypeA}.
  Since $\frac{2}{h_c}<(\sqrt{a_+}-\sqrt{a_-})^2<a_+$ holds, the condition $2<h_ca_+(1-\varepsilon)$
  also holds by choosing a sufficiently small $\varepsilon$. Then, we have
  \begin{align}
    \lim_{v\rightarrow\infty}x_-(v)\exp{\left[\frac{h_{c}}{2}(1-\varepsilon)X_+(v)\right]}
    \,=\,\lim_{v\rightarrow\infty}a_-v^{\frac{h_{c}a_+(1-\varepsilon)}{2}-1}\,=\,\infty.
  \end{align}
  Therefore, these spacetimes are of Class~A. Consequently, these spacetimes are of Type~2-A.
\end{proof}
\begin{example}
  Case~3 spacetimes with $x_+=\frac{a_+}{v}$ and $x_-=\frac{a_-}{v^k}$ satisfying $2<2k<{h_{c}}a_+$ are of Type~2-A.
\end{example}
\begin{proof}
  We evaluate Eq.~\eqref{eq-Case3-Type2-2} of  Corollary \ref{corollary-Case3-Type2}. 
  Since $2<a_+h_c$ holds, taking a sufficiently small $\alpha$ leads to $\frac{1}{a_+}<\frac{1-\alpha}{2}h_c$, and then we have
  \begin{align}
    \lim_{v\rightarrow\infty}\frac{\alpha dx_+/dv+(1-\alpha)dx_-/dv}{(x_+-x_-)^2}
    \,=\,-\frac{\alpha}{a_+}
    \,>\,-\frac{h_c}{2}\alpha(1-\alpha).
  \end{align}
  Therefore, these spacetimes are of Class~2. Next, we evaluate Eq.~\eqref{eq-Case3-TypeA-1} of Proposition \ref{prop-Case3-TypeA}. Since $2k<h_ca_+$ holds, choosing a sufficiently small $\varepsilon$ leads to 
  the condition $2k<h_ca_+(1-\varepsilon)$. Then, we have
  \begin{align}
    \lim_{v\rightarrow\infty}x_-(v)\exp{\left[\frac{h_{c}}{2}(1-\varepsilon)X_+(v)\right]}
    \,=\,\lim_{v\rightarrow\infty}a_-v^{\frac{h_{c}a_+(1-\varepsilon)}{2}-k}
    \,=\,\infty.
  \end{align}
Therefore, these spacetimes are of Class~A. Consequently, these spacetimes are of Type~2-A.
\end{proof}

\subsubsection{Examples of Type~2-B in Case~3}

Finally, we show some examples of Type~2-B spacetimes in Case~3.  
\begin{example}
  Case~3 spacetimes with $x_+=\frac{a_+}{v^n}$ and $x_-=a_-\exp{\left(-\frac{v}{L}\right)}$ satisfying $0<n<1$ and $L>0$ are of Type~2-B.
\end{example}
\begin{proof}
  We evaluate Eq.~\eqref{eq-Case3-Type2-2} of Corollary \ref{corollary-Case3-Type2};
  \begin{align}
    \lim_{v\rightarrow\infty}\frac{\alpha dx_+/dv+(1-\alpha)dx_-/dv}{(x_+-x_-)^2}
    \,=\,\lim_{v\rightarrow\infty}\frac{-\alpha n a_+-(1-\alpha)(a_-/L)v^{n+1}\exp{(-v/L)}}{(a_+-a_-v^n\exp{(-v/L)})^2}v^{n-1}
    \,=\,0.
  \end{align}
  Therefore, these spacetimes are of Class~2. Next, we evaluate Eq.~\eqref{eq-Case3-TypeB-3} of Corollary \ref{corollary-Case3-TypeB};
  \begin{align}
    \lim_{v\rightarrow\infty}\frac{dx_-/dv}{(x_++\alpha x_-)x_-}
    \,=\,\lim_{v\rightarrow\infty}\frac{-v^n/L}{a_++\alpha a_-v^n\exp{(-v/L)}}=-\infty.
  \end{align}
  Therefore, these spacetimes are of Class~B. Consequently, these spacetimes are of Type~2-B.
\end{proof}
\begin{example}
  Case~3 spacetimes with $x_+=\frac{a_+}{v}$ and $x_-=\frac{a_-}{v^k}$ satisfying $2<a_+h_c<2k$ are of Type~2-B.
\end{example}
\begin{proof}
  We evaluate Eq.~\eqref{eq-Case3-Type2-2} of Corollary \ref{corollary-Case3-Type2};
  \begin{align}
    \lim_{v\rightarrow\infty}\frac{\alpha dx_+/dv+(1-\alpha)dx_-/dv}{(x_+-x_-)^2}
    \,=\,\lim_{v\rightarrow\infty}\frac{-\alpha a_+-(1-\alpha)ka_-/v^{k-1}}{(a_+-a_-/v^{k-1})^2}
    \,=\,-\frac{\alpha}{a_+}.
  \end{align}
  Since $-\frac{1}{a_+}>-\frac{h_c}{2}$ holds, choosing a sufficiently small $\alpha$ leads to $-\frac{\alpha}{a_+}>-\frac{h_c}{2}\alpha(1-\alpha)$. Then for such an $\alpha$, the condition of Eq.~\eqref{eq-Case3-Type2-2} is satisfied, and
    therefore, these spacetimes are of Class~2. Next, we evaluate Eq.~\eqref{eq-Case3-TypeB-3}
  of Corollary \ref{corollary-Case3-TypeB};
  \begin{align}
    \lim_{v\rightarrow\infty}\frac{dx_-/dv}{(x_++\alpha x_-)x_-}\,=\,-\frac{k}{a_+}.
  \end{align}
  Since $-\frac{k}{a_+}<-\frac{h_c}{2}$ holds, choosing sufficiently large $\alpha$ leads to $-\frac{k}{a_+}<-\frac{1+\alpha}{2\alpha}h_c$. Then, for such $\alpha$, the condition of Eq.~\eqref{eq-Case3-TypeB-3} is satisfied,
  and
  therefore, these spacetimes are of Class~B. Consequently, these spacetimes are of Type~2-B.
\end{proof}

\pagebreak

\section{Finiteness of affine parameters in solvable models}
\label{Finiteness-of-affine-parameters}

In this appendix, we prove the finiteness of outgoing null geodesics inside the event horizon
in solvable models of Sec.~\ref{Sec:Solvable-models}. 
Our analyses will be done by applying Lemma~\ref{Lemma:finiteness-of-lambda},
except for the model of Sec.~\ref{Sec:Solvable-Deq0} for which a direct calculation is required.

\subsection{Models of Sec.~\ref{Sec:Solvable-0<n<1}}
\label{Finiteness-Solvable-0<n<1}

In Sec.~\ref{Sec:Solvable-0<n<1}, the models with $x_\pm \approx a_\pm/v^n$ with $0<n<1$ are considered.
Since all outgoing null geodesics inside the event horizon
approach the attractive separatrix located at a close position to the inner AH, $x(v)$ behaves as $x\approx a_-/v^n$.
Integration of $x_\pm$ and $x$ gives
\begin{equation}
X_+\,\approx\, \frac{a_+}{1-n}v^{1-n}, \qquad \textrm{and}\qquad X\,\approx\, X_-\,\approx\, \frac{a_-}{1-n}v^{1-n},
\end{equation}
and then, we have
\begin{equation}
S(v)\,\approx\, \left[-\frac{h_c(a_+-a_-)}{2(1-n)}+O(\varepsilon,\varepsilon_\pm)\right]v^{1-n}.
\end{equation}
Thus, the finiteness of the affine parameter is guaranteed from Lemma~\ref{Lemma:finiteness-of-lambda}.

\subsection{Models of Sec.~\ref{Sec:Solvable-n=1}}

In Sec.~\ref{Sec:Solvable-n=1}, the models with $x_\pm \approx a_\pm/v$ are considered. 
Since the solutions to this configuration are obtained in the cases $D>0$ and $D=0$ separately,
where $D$ is defined in Eq.~\eqref{Def:b-and-D}, 
we discuss the extendibility in these two cases one by one.

\subsubsection{Models of $D>0$}

Models of $D>0$ are discussed in Sec.~\ref{Sec:Solvable-D>0}, and 
there, it is found that there exists an attractive separatrix $x=(b-d)/v$ which all outgoing null 
geodesics in the black hole region are attracted, where $b$ and $d$ are defined in Eqs.~\eqref{Def:b-and-D} and \eqref{Def:d}. 
For this reason, we have $x\approx (b-d)/v$, since the deviation from the attractive separatrix is subdominant.
Then, we have $X_\pm \approx a_\pm\log v$ and $X\,\approx\, (b-d)\log v$, and hence,
\begin{equation}
S(v)\,\approx\, \left[-1-h_cd+O(\varepsilon,\varepsilon_\pm)\right]\log v.
\end{equation}
Thus, the finiteness of the affine parameter is guaranteed from Lemma~\ref{Lemma:finiteness-of-lambda}.

\subsubsection{Models of $D=0$}

Models of $D=0$ are discussed in Sec.~\ref{Sec:Solvable-Deq0}, and there,  it is found that there exists a degenerate separatrix $x=b/v$ (that is, the event horizon) which all outgoing null geodesics in the black hole region are attracted.
Since it turns out that Lemma~\ref{Lemma:finiteness-of-lambda} is not applicable to this situation,
we recalculate the bound on $\lambda$ using Eq.~\eqref{lambda-bound}. 
Since models in Sec.~\ref{Sec:Solvable-models} have assumed $h(v,x)=h_c$, 
Eq.~\eqref{lambda-bound} leads to the inequality,
\begin{equation}
\lambda-\lambda_0 \,<\, \frac{(1+\varepsilon)A_\infty\exp\left[-\tilde{S}(v_0)\right]}{\dot{v}(\lambda_0)A(\lambda_0)}
\int_{v_0}^v \exp\left[\tilde{S}(v^\prime)\right]dv^\prime 
\end{equation}
with 
\begin{equation}
\tilde{S}(v)\,:=\, h_c\left[X(v)-\frac{1}{2}(X_+(v)+X_-(v))\right].
\end{equation}
Using $x_\pm=a_\pm/v$ and the solution for $x(v)$ given in Eq.~\eqref{Solution-D=0-insideBH},
$\tilde{S}(v)$ is calculated as
\begin{equation}
\tilde{S}(v)\,=\, \log\left[v^{-1}\left(\log\frac{v}{v_e}\right)^{-2}\right].
\end{equation}
Then, integrating $\exp\left[\tilde{S}(v^\prime)\right]$, we have 
\begin{equation}
\int_{v_0}^v \exp\left[\tilde{S}(v^\prime)\right]dv^\prime
\, = \, \frac{1}{\log(v_0/v_e)}-\frac{1}{\log(v/v_e)}.
\end{equation}
This demonstrates the finiteness of $\lambda$ in the limit $v\to\infty$.

\subsection{Models of Sec.~\ref{Sec:Solvable-n>1}}

In Sec.~\ref{Sec:Solvable-n>1}, the models with $x_\pm \approx \pm a/v^n$ with $n>1$ are considered.
All outgoing null geodesics inside the event horizon
approach the degenerate separatrix that is the event horizon, and their behavior is given as $x\approx (2/h_c)/(v_e-v)$.
Therefore, $X_\pm$ give subdominant contribution to $S(v)$, and the leading-order behavior
is determined by $X(v)\approx\, -({2}/{h_c})\log v$. Then, we have
\begin{equation}
S(v)\,\approx\, \left[-2+O(\varepsilon)\right]\log v,
\end{equation}
and thus, the finiteness of the affine parameter is guaranteed from Lemma~\ref{Lemma:finiteness-of-lambda}.

\subsection{Models of Sec.~\ref{Sec:Solvable-Type2Case2-Type2BCase3}}

In Sec.~\ref{Sec:Solvable-Type2Case2-Type2BCase3}, the models with $x_+(v) \approx a/v$ with $a=4/h_c$ 
and $x_-\approx b/v^k$ with $k>2$ are considered.
In those models, all outgoing null geodesics inside the event horizon
approach the attractive separatrix with the behavior  $x=O(1/v^2)$.
Therefore, $X_-$ and $X$ give subdominant contribution to $S(v)$, and the leading-order behavior
is determined by $X_+(v)\approx (4/h_c)\log v$. Then, we have
\begin{equation}
S(v)\,\approx\, \left[-2+O(\varepsilon_+)\right]\log v,
\end{equation}
and thus, the finiteness of the affine parameter is guaranteed from Lemma~\ref{Lemma:finiteness-of-lambda}.

\pagebreak

% ****** End of file apssamp.tex ******

\printbibliography

@article{Ayon-Beato:2000mjt,
  title = {The {{Bardeen Model}} as a {{Nonlinear Magnetic Monopole}}},
  author = {Ayón-Beato, Eloy and García, Alberto},
  date = {2000-11},
  journaltitle = {Physics Letters B},
  shortjournal = {Physics Letters B},
  volume = {493},
  number = {1--2},
  eprint = {gr-qc/0009077},
  eprinttype = {arXiv},
  pages = {149--152},
  issn = {03702693},
  doi = {10.1016/S0370-2693(00)01125-4}
}

@article{Bardeen-2677,
  title = {Non-Singular General Relativistic Gravitational Collapse},
  author = {Bardeen, James},
  date = {1968-09},
  journaltitle = {Proceedings of the 5th International Conference on Gravitation and the Theory of Relativity},
  pages = {87},
  url = {https://ui.adsabs.harvard.edu/abs/1968qtr..conf...87B/abstract},
  urldate = {2025-06-06},
  langid = {english},
  language = {en}
}

@article{Bueno:2024dgm,
  title = {Regular Black Holes from Pure Gravity},
  author = {Bueno, Pablo and Cano, Pablo A. and Hennigar, Robie A.},
  date = {2025-02},
  journaltitle = {Physics Letters B},
  shortjournal = {Physics Letters B},
  volume = {861},
  pages = {139260},
  issn = {03702693},
  doi = {10.1016/j.physletb.2025.139260},
  langid = {english},
  language = {en}
}

@article{Bueno:2024eiga,
  title = {Dynamical {{Formation}} of {{Regular Black Holes}}},
  author = {Bueno, Pablo and Cano, Pablo A. and Hennigar, Robie A. and Murcia, Ángel J.},
  date = {2025-05-06},
  journaltitle = {Physical Review Letters},
  shortjournal = {Phys. Rev. Lett.},
  volume = {134},
  number = {18},
  pages = {181401},
  publisher = {American Physical Society},
  doi = {10.1103/PhysRevLett.134.181401},
  langid = {american},
  language = {en-US}
}

@online{Bueno:2025gjg,
  title = {Regular Black Holes from {{Oppenheimer-Snyder}} Collapse},
  author = {Bueno, Pablo and Cano, Pablo A. and Hennigar, Robie A. and Murcia, Ángel J. and Vicente-Cano, Aitor},
  date = {2025-05-14},
  eprint = {2505.09680},
  eprinttype = {arXiv},
  eprintclass = {gr-qc},
  doi = {10.48550/arXiv.2505.09680},
  langid = {american},
  language = {en-US},
  pubstate = {prepublished},
  version = {1}
}

@article{Carballo-Rubio:2018pmi,
  title = {On the Viability of Regular Black Holes},
  author = {Carballo-Rubio, Raúl and Di Filippo, Francesco and Liberati, Stefano and Pacilio, Costantino and Visser, Matt},
  date = {2018-07},
  journaltitle = {Journal of High Energy Physics},
  shortjournal = {J. High Energ. Phys.},
  volume = {2018},
  number = {7},
  pages = {23},
  issn = {1029-8479},
  doi = {10.1007/JHEP07(2018)023},
  langid = {english},
  language = {en}
}

@article{Chen:2014jwq,
    author = "Chen, Pisin and Ong, Yen Chin and Yeom, Dong-han",
    title = "{Black Hole Remnants and the Information Loss Paradox}",
    eprint = "1412.8366",
    archivePrefix = "arXiv",
    primaryClass = "gr-qc",
    reportNumber = "NORDITA-2014-149",
    doi = "10.1016/j.physrep.2015.10.007",
    journal = "Phys. Rept.",
    volume = "603",
    pages = "1--45",
    year = "2015"
}

@article{Chinaglia:2017uqd,
  title = {A Note on Singular and Non-Singular Black Holes},
  author = {Chinaglia, Stefano and Zerbini, Sergio},
  date = {2017-05-11},
  journaltitle = {General Relativity and Gravitation},
  shortjournal = {Gen Relativ Gravit},
  volume = {49},
  number = {6},
  pages = {75},
  issn = {1572-9532},
  doi = {10.1007/s10714-017-2235-6},
  langid = {english},
  language = {en}
}

@article{DiFilippo:2024mwm,
  title = {Inner-Extremal Regular Black Holes from Pure Gravity},
  author = {Di Filippo, Francesco and Kolář, Ivan and Kubizňák, David},
  date = {2025-02-24},
  journaltitle = {Physical Review D},
  shortjournal = {Phys. Rev. D},
  volume = {111},
  number = {4},
  pages = {L041505},
  publisher = {American Physical Society},
  doi = {10.1103/PhysRevD.111.L041505}
}

@article{Dymnikova:1992ux,
  title = {Vacuum Nonsingular Black Hole},
  author = {Dymnikova, Irina},
  date = {1992-03-01},
  journaltitle = {General Relativity and Gravitation},
  shortjournal = {Gen Relat Gravit},
  volume = {24},
  number = {3},
  pages = {235--242},
  issn = {1572-9532},
  doi = {10.1007/BF00760226},
  langid = {english},
  language = {en}
}

@article{Dymnikova:2003vt,
  title = {Spherically Symmetric Space-Time with the Regular de {{Sitter}} Center},
  author = {Dymnikova, Irina},
  date = {2003-07},
  journaltitle = {International Journal of Modern Physics D},
  shortjournal = {Int. J. Mod. Phys. D},
  volume = {12},
  number = {06},
  eprint = {gr-qc/0304110},
  eprinttype = {arXiv},
  pages = {1015--1034},
  issn = {0218-2718, 1793-6594},
  doi = {10.1142/S021827180300358X}
}

@article{Fan:2016hvf,
  title = {Construction of Regular Black Holes in General Relativity},
  author = {Fan, Zhong-Ying and Wang, Xiaobao},
  date = {2016-12-20},
  journaltitle = {Physical Review D},
  shortjournal = {Phys. Rev. D},
  volume = {94},
  number = {12},
  pages = {124027},
  publisher = {American Physical Society},
  doi = {10.1103/PhysRevD.94.124027}
}

@article{Frolov:2014jva,
  title = {Information Loss Problem and a `black Hole' Model with a Closed Apparent Horizon},
  author = {Frolov, Valeri P.},
  date = {2014-05},
  journaltitle = {Journal of High Energy Physics},
  shortjournal = {J. High Energ. Phys.},
  volume = {2014},
  number = {5},
  eprint = {1402.5446},
  eprinttype = {arXiv},
  eprintclass = {gr-qc, physics:hep-th},
  pages = {49},
  issn = {1029-8479},
  doi = {10.1007/JHEP05(2014)049},
  langid = {english},
  language = {en}
}

@article{Frolov:2024hhe,
  title = {Regular Black Holes Inspired by Quasitopological Gravity},
  author = {Frolov, Valeri P. and Koek, Alex and Soto, Jose Pinedo and Zelnikov, Andrei},
  date = {2025-02-12},
  journaltitle = {Physical Review D},
  shortjournal = {Phys. Rev. D},
  volume = {111},
  number = {4},
  pages = {044034},
  publisher = {American Physical Society},
  doi = {10.1103/PhysRevD.111.044034},
  langid = {american},
  language = {en-US}
}

@article{Hawking:1975vcx,
  title = {Particle Creation by Black Holes},
  author = {Hawking, S. W.},
  date = {1975-08-01},
  journaltitle = {Communications in Mathematical Physics},
  shortjournal = {Commun.Math. Phys.},
  volume = {43},
  number = {3},
  pages = {199--220},
  issn = {1432-0916},
  doi = {10.1007/BF02345020},
  langid = {english},
  language = {en}
}

@article{Hawking:1976ra,
  title = {Breakdown of Predictability in Gravitational Collapse},
  author = {Hawking, S. W.},
  date = {1976-11-15},
  journaltitle = {Physical Review D},
  shortjournal = {Phys. Rev. D},
  volume = {14},
  number = {10},
  pages = {2460--2473},
  publisher = {American Physical Society},
  doi = {10.1103/PhysRevD.14.2460}
}

@article{Hayward:2005gi,
  title = {Formation and {{Evaporation}} of {{Nonsingular Black Holes}}},
  author = {Hayward, Sean A.},
  date = {2006-01-26},
  journaltitle = {Physical Review Letters},
  shortjournal = {Phys. Rev. Lett.},
  volume = {96},
  number = {3},
  pages = {031103},
  issn = {0031-9007, 1079-7114},
  doi = {10.1103/PhysRevLett.96.031103},
  langid = {english},
  language = {en}
}

@article{Kodama:1979vm,
  title = {Inevitability of a {{Naked Singularity Associated}} with the {{Black Hole Evaporation}}},
  author = {Kodama, Hideo},
  date = {1979-11-01},
  journaltitle = {Progress of Theoretical Physics},
  shortjournal = {Progress of Theoretical Physics},
  volume = {62},
  number = {5},
  pages = {1434--1435},
  issn = {0033-068X},
  doi = {10.1143/PTP.62.1434}
}

@article{Konoplya:2024kih,
  title = {Dymnikova Black Hole from an Infinite Tower of Higher-Curvature Corrections},
  author = {Konoplya, R. A. and Zhidenko, A.},
  date = {2024-09-01},
  journaltitle = {Physics Letters B},
  shortjournal = {Physics Letters B},
  volume = {856},
  pages = {138945},
  issn = {0370-2693},
  doi = {10.1016/j.physletb.2024.138945},
  langid = {american},
  language = {en-US}
}

@article{Lan:2023cvz,
  title = {Regular Black Holes: {{A}} Short Topic Review},
  shorttitle = {Regular Black Holes},
  author = {Lan, Chen and Yang, Hao and Guo, Yang and Miao, Yan-Gang},
  date = {2023-09-05},
  journaltitle = {International Journal of Theoretical Physics},
  shortjournal = {Int J Theor Phys},
  volume = {62},
  number = {9},
  eprint = {2303.11696},
  eprinttype = {arXiv},
  eprintclass = {gr-qc, physics:hep-th},
  pages = {202},
  issn = {1572-9575},
  doi = {10.1007/s10773-023-05454-1},
  langid = {american},
  language = {en-US}
}

@article{Lesourd:2018ekn,
  title = {Causal Structure of Evaporating Black Holes},
  author = {Lesourd, Martin},
  date = {2019-01-24},
  journaltitle = {Classical and Quantum Gravity},
  shortjournal = {Class. Quantum Grav.},
  volume = {36},
  number = {2},
  eprint = {1808.07303},
  eprinttype = {arXiv},
  eprintclass = {gr-qc},
  pages = {025007},
  issn = {0264-9381, 1361-6382},
  doi = {10.1088/1361-6382/aaf5f8}
}

@article{Maeda:2021jdc,
  title = {Quest for Realistic Non-Singular Black-Hole Geometries: {{Regular-center}} Type},
  shorttitle = {Quest for Realistic Non-Singular Black-Hole Geometries},
  author = {Maeda, Hideki},
  date = {2022-11-18},
  journaltitle = {Journal of High Energy Physics},
  shortjournal = {J. High Energ. Phys.},
  volume = {2022},
  number = {11},
  eprint = {2107.04791},
  eprinttype = {arXiv},
  eprintclass = {gr-qc, physics:hep-th},
  pages = {108},
  issn = {1029-8479},
  doi = {10.1007/JHEP11(2022)108},
  langid = {english},
  language = {en}
}

@article{Penrose:1964wq,
  title = {Gravitational {{Collapse}} and {{Space-Time Singularities}}},
  author = {Penrose, Roger},
  date = {1965-01-18},
  journaltitle = {Physical Review Letters},
  shortjournal = {Phys. Rev. Lett.},
  volume = {14},
  number = {3},
  pages = {57--59},
  publisher = {American Physical Society},
  doi = {10.1103/PhysRevLett.14.57}
}

@article{Sueto:2023ztw,
  title = {Evaporation of a Nonsingular {{Reissner}}–{{Nordström}} Black Hole and the Information Loss Problem},
  author = {Sueto, Kensuke and Yoshino, Hirotaka},
  date = {2023-10-01},
  journaltitle = {Progress of Theoretical and Experimental Physics},
  shortjournal = {Progress of Theoretical and Experimental Physics},
  volume = {2023},
  number = {10},
  pages = {103E01},
  issn = {2050-3911},
  doi = {10.1093/ptep/ptad111}
}
\end{document}